\documentclass[reqno,oneside,10pt]{amsart}

\usepackage{amssymb,mathptmx,eucal,array,enumitem,setspace,geometry,color,multirow}
\usepackage{mathtools} 
\usepackage[noadjust]{cite}

\usepackage{tikz}
\usetikzlibrary{decorations.pathmorphing}
\tikzset{>=to} 
\tikzset{squig/.style={->, line join=round, decorate,
                       decoration={zigzag, segment length=4, amplitude=.9, post=lineto, post length=2pt}}} 

\newcommand{\dodu}[1]{\fill [black] (#1) circle (4pt);}

\definecolor{rouge}{rgb}{0.95,0.1,0.15}
\definecolor{forestgreen}{rgb}{0.13,0.54,0.13}
\definecolor{bleu}{rgb}{0.1,0.1,0.9}

\newcommand{\invisible}[1]{} 

\makeatletter 
\newcases{dlrcases}{\quad}{%
  $\m@th\displaystyle{##}$\hfil}{$\m@th\displaystyle{##}$\hfil}{\lbrace}{\rbrace}
\newcases{lrcases}{\quad}{%
  $\m@th{##}$\hfil}{$\m@th{##}$\hfil}{\lbrace}{\rbrace}
\newcases{dlrcases*}{\quad}{%
  $\m@th\displaystyle{##}$\hfil}{{##}\hfil}{\lbrace}{\rbrace}
\newcases{lrcases*}{\quad}{%
  $\m@th{##}$\hfil}{{##}\hfil}{\lbrace}{\rbrace}
\makeatother

\usepackage[pdfpagemode=UseNone, pdfstartview={XYZ null null null}]{hyperref}
\hypersetup{
colorlinks=true,
citecolor=red,
linkcolor=blue,
urlcolor=blue
}

\usepackage[capitalise,noabbrev]{cleveref} 

\setlist[itemize]{leftmargin=*}                  
\setlist[enumerate]{leftmargin=*,                
                    label=\textup{(\roman*)}}    

\DeclareSymbolFont{largesymbols}{OMX}{zplm}{m}{n} 

\let\originalleft\left     
\let\originalright\right
\renewcommand{\left}{\mathopen{}\mathclose\bgroup\originalleft}
\renewcommand{\right}{\aftergroup\egroup\originalright}

\numberwithin{equation}{section}

\newcolumntype{C}{>{$}c<{$}}              
\newcolumntype{D}{>{$\displaystyle}l<{$}} 
\allowdisplaybreaks

\newcommand{\func}[2]{#1 \left( #2 \right)}

\newcommand{\brac}[1]{\left( #1 \right)}
\newcommand{\sqbrac}[1]{\left[ #1 \right]}
\newcommand{\set}[1]{\left\{ #1 \right\}}
\newcommand{\st}{\mspace{5mu} : \mspace{5mu}}

\newcommand{\abs}[1]{\left\lvert #1 \right\rvert}

\newcommand{\inner}[2]{\bigl\langle #1 , #2 \bigr\rangle}
\newcommand{\bilin}[2]{\inner{#1}{#2}}

\newcommand{\fld}[1]{\mathbb{#1}} 
\newcommand{\ZZ}{\fld{Z}}

\newcommand{\CC}{\fld{C}}
\newcommand{\KK}{\fld{K}}

\newcommand{\ii}{\mathfrak{i}}
\newcommand{\ee}{\mathsf{e}}

\newcommand*{\petitr}{0.25}%
\newcommand*{\rr}{1.5}%
\newcommand*{\rrr}{2.0}%
\newcommand*{\rrrr}{2.5}%
\newcommand*{\rrrrr}{3.0}%
\newcommand{\alg}[1]{\mathcal{#1}}  
\newcommand{\Alg}[1]{\alg{A}_{#1}}
\newcommand{\tl}[1]{\mathsf{TL}_{#1}}
\newcommand{\dtl}[1]{\mathsf{dTL}_{#1}}

\newcommand{\grp}[1]{\mathsf{#1}}   

\newcommand{\Mod}[1]{\mathsf{#1}}   
\newcommand{\Stan}[1]{\Mod{S}_{#1}} 
\newcommand{\Radi}[1]{\Mod{R}_{#1}} 
\newcommand{\Irre}[1]{\Mod{I}_{#1}} 
\newcommand{\Proj}[1]{\Mod{P}_{#1}} 
\newcommand{\Inje}[1]{\Mod{J}_{#1}} 
\newcommand{\Cost}[1]{\Mod{C}_{#1}} 
\newcommand{\TheA}[1]{\Mod{A}_{#1}} %
\newcommand{\TheV}[1]{\Mod{V}_{#1}} %
\newcommand{\TheB}[2]{\Mod{B}_{#1}^{#2}} %
\newcommand{\TheT}[2]{\Mod{T}_{#1}^{#2}} %

\DeclareMathOperator{\rdsymb}{rd} 
\newcommand{\rd}[1]{\rdsymb \bigl[ #1 \bigr]} 

\newcommand{\Res}[1]{#1 \mbox{$\downarrow$} {}}
\newcommand{\Ind}[1]{#1 \raisebox{0.18em}{$\uparrow$} {}}

\newcommand{\ra}{\rightarrow}
\newcommand{\lra}{\longrightarrow}
\newcommand{\llra}{\longleftrightarrow}
\newcommand{\ira}{\hookrightarrow}    
\newcommand{\sra}{\twoheadrightarrow} 
\newcommand{\lira}{\ensuremath{\lhook\joinrel\relbar\joinrel\rightarrow}} 
\newcommand{\lsra}{\ensuremath{\relbar\joinrel\twoheadrightarrow}}        

\newcommand{\tauarrow}{\rightsquigarrow} 

\newcommand{\ses}[3]{0 \ra #1 \ra #2 \ra #3 \ra 0}

\newcommand{\dses}[5]{0 \lra #1 \overset{#2}{\lra} #3 \overset{#4}{\lra} #5 \lra 0}

\newcommand{\rdses}[5]{#1 \overset{#2}{\lra} #3 \overset{#4}{\lra} #5 \lra 0}

\newcommand{\blank}{{-}} 

\newcommand{\dual}[1]{#1^*}                
\newcommand{\twdu}[1]{#1^{\vee}}           
\newcommand{\delr}[1]{\delta_R^{#1}}       

\newcommand{\VectD}[1]{\dual{#1}}
\newcommand{\AlgD}[1]{#1^{t}}
\DeclareMathOperator{\ARTr}{Tr}
\newcommand{\ARtransp}[1]{\ARTr #1}
\newcommand{\ARTauSymbol}{\tau}
\newcommand{\ARTau}[1]{\ARTauSymbol #1}
\newcommand{\PowerARTau}[2]{\ARTauSymbol^{#1} #2}
\newcommand{\iARTau}[1]{\PowerARTau{-1}{#1}}

\DeclareMathOperator{\id}{id}

\newcommand{\eid}{\mathsf{eid}}
\newcommand{\oid}{\mathsf{oid}}
\newcommand{\lmod}{\Alg{}\textup{-mod}} 
\newcommand{\rmod}{\textup{mod-}\Alg{}} 

\DeclareMathOperator{\sign}{sign}
\DeclareMathOperator{\im}{im}
\DeclareMathOperator{\rad}{rad}
\DeclareMathOperator{\soc}{soc}
\DeclareMathOperator{\head}{hd}
\DeclareMathOperator{\Coker}{coker}
\DeclareMathOperator{\coker}{coker}

\DeclareMathOperator{\Hom}{Hom}

\DeclareMathOperator{\Ext}{Ext}
\DeclareMathOperator{\Tor}{Tor}

\theoremstyle{plain}
\newtheorem{theorem}{Theorem}[section]
\newtheorem{lemma}[theorem]{Lemma}
\newtheorem{proposition}[theorem]{Proposition}
\newtheorem{corollary}[theorem]{Corollary}
\newtheorem{definition}[theorem]{Definition}

\renewcommand{\cong}{\simeq}
\renewcommand{\ncong}{\not\simeq}

\newcommand{\AR}{Auslander-Reiten}

\begin{document}

\title[Restriction and induction of TL indecomposables]{Restriction and induction \\ of indecomposable modules \\ over the Temperley-Lieb algebras}

\author[J Bellet\^ete]{Jonathan Bellet\^ete}

\address[Jonathan Bellet\^ete]{
D\'{e}partement de Physique \\
Universit\'{e} de Montr\'{e}al \\
Qu\'{e}bec, Canada, H3C~3J7.
}

\email{jonathan.belletete@umontreal.ca}

\author[D Ridout]{David Ridout}

\address[David Ridout]{
School of Mathematics and Statistics \\
University of Melbourne \\
Parkville, Australia, 3010.
}

\email{david.ridout@unimelb.edu.au}

\author[Y Saint-Aubin]{Yvan Saint-Aubin}

\address[Yvan Saint-Aubin]{
D\'{e}partement de Math\'{e}matiques et de Statistique \\
Universit\'{e} de Montr\'{e}al \\
Qu\'{e}bec, Canada, H3C~3J7.
}

\email{yvan.saint-aubin@umontreal.ca}

\date{\today}


\begin{abstract}
Both the original Temperley-Lieb algebras $\tl{n}$ and their dilute counterparts $\dtl{n}$ form families of filtered algebras: $\tl{n}\subset \tl{n+1}$ and $\dtl{n}\subset\dtl{n+1}$, for all $n\ge 0$. For each such inclusion, the restriction and induction of every finite-dimensional indecomposable module over $\tl{n}$ (or $\dtl{n}$) is computed. To accomplish this, a thorough description of each indecomposable is given, including its projective cover and injective hull, some short exact sequences in which it appears, its socle and head, and its extension groups with irreducible modules. These data are also used to prove the completeness of the list of indecomposable modules, up to isomorphism. In fact, two completeness proofs are given, the first is based on elementary homological methods and the second uses Auslander-Reiten theory. The latter proof offers a detailed example of this algebraic tool that may be of independent interest.
\end{abstract}

\maketitle

\tableofcontents

\onehalfspacing

%
%

\section{Introduction} \label{sec:Intro}

One of the two main goals of this paper is to construct a complete list of non-isomorphic indecomposable modules of the Temperley-Lieb algebras $\tl{n}(\beta)$ and their dilute counterparts $\dtl{n}(\beta)$.  Both families of algebras are parametrised by a positive integer $n$ and a complex number $\beta$. These indecomposable modules are characterised using Loewy diagrams, from which socles and heads may be deduced, and their projective covers and injective hulls are determined. We also compute some of their extension groups and detail some of the short exact sequences in which they appear. Many of these results are new, but a subset of these extensive data has appeared in the literature, most of the time without proofs. Due to the increasing role played by the representation theory of these two families of algebras, both in mathematics and physics, we believe that a self-contained treatment is called for.

The second main goal concerns the fact that both the Temperley-Lieb and dilute Temperley-Lieb families of algebras have a natural filtration: $\tl{n}(\beta)\subset \tl{n+1}(\beta)$ and $\dtl{n}(\beta) \subset \dtl{n+1}(\beta)$, for $n\ge 1$. This first inclusion, for example, is realised by using all the generators $e_i$, $1\le i\le n$, of $\tl{n+1}(\beta)$, except $e_n$, to generate a subalgebra isomorphic to $\tl{n}(\beta)$. The restriction and induction of $\tl{n}(\beta)$- and $\dtl{n}(\beta)$-modules along these filtrations has grown to be a useful tool and the paper identifies the results, up to isomorphism, for all indecomposable modules. To the best of our knowledge, these results are also new.

The study of the representation theory of the Temperley-Lieb algebra was launched by Jones \cite{JonInd83}, Martin \cite{Martin} and Goodman and Wenzl \cite{GoodWenzl93}. They provided descriptions of the projective and irreducible $\tl{n}(\beta)$-modules. Shortly thereafter, Graham and Lehrer \cite{GraCel96} introduced a large class of algebras, the ``cellular algebras", that included both $\tl{n}(\beta)$ and $\dtl{n}(\beta)$.  They proved many general properties and obtained a beautiful description of the role that their ``cell representations'' play as intermediates between the irreducible and projective modules. Many cellular algebras have diagrammatic descriptions that are closely related to their applications in statistical mechanics.  In this setting, the cell representations are often referred to as the ``standard modules''. More recent works on the representation theory of $\tl{n}(\beta)$ and $\dtl{n}(\beta)$ have profited by combining Graham and Lehrer's insights into these standard modules with diagrammatic presentations \cite{KauInv90,WesRep95,RSA,BSA}.

The recent literature attests to the necessity of having a thorough understanding of the indecomposable modules of the Temperley-Lieb algebras. We outline here three examples that have arisen in the mathematical physics literature.  In \cref{app:physics}, we discuss in more detail three models of significant physical interest in which ``exotic'' indecomposable representations (by which we mean neither standard nor irreducible) of $\tl{n}(\beta)$ naturally appear.

First, due to the close relationship between statistical lattice models and conformal field theories, it is believed that an operation analogous to the fusion product of conformal field theory should exist at the level of the lattice.  One proposal for such a ``lattice fusion product" appeared in the work of Read and Saleur \cite{ReadSaleur} and explicit computations were subsequently carried out over the Temperley-Lieb algebra by Gainutdinov and Vasseur \cite{GainutdinovVasseur}.  Their method exploited the well known relationship between the representation theories of $\tl{n}(\beta)$ and the quantum group $U_q(\mathfrak{sl}_2)$, leading them to the following simple but crucial observation: The lattice fusion product of a $\tl{n}(\beta)$-module with a $\tl{m}(\beta)$-module does not depend on $n$ or $m$, if both are sufficiently large.\footnote{The fact that lattice fusion operates between modules over different algebras indicates that it should be thought of as an operation on modules over the Temperley-Lieb categories.} This observation is welcome since these integers have no natural interpretation in the continuum limit where conformal field theory is believed to take over. Further examples of lattice fusion products were computed in \cite{BelFus15}, for both $\tl{n}(\beta)$- and $\dtl{n}(\beta)$-modules, using powerful arguments based on restriction and induction. The classification of indecomposables is important here because the lattice fusion of an irreducible and a standard module may result in indecomposables that are neither irreducible, standard nor projective. These computations assume some of the results that will be proved in \cref{sec:restrictionInduction} (and thereby provide one of the primary motivations for the work reported in this paper).

As a second example illustrating the importance of having a complete list of indecomposable modules, we mention the work \cite{dimer2016} of Morin-Duchesne, Rasmussen and Ruelle, who used the map between dimer and XXZ spin configurations to introduce a new representation of $\tl{n}(\beta)$, with $\beta=0$. This representation was crucial for their analysis of dimer models and their contention \cite{MorInt16} that the continuum limit is described by a (logarithmic) conformal field theory of central charge $c=-2$. In most models described by the XXZ spin-chain, the Temperley-Lieb generators act on two neighbouring sites. In their new representation, the generators act on three consecutive sites and the decomposition of this representation into a direct sum of indecomposable modules involves modules that are, generically, neither irreducible, standard nor projective. We shall discuss this example in more detail in \cref{app:physics}.

Our last example points even further down the path that the present paper starts to explore. The Temperley-Lieb algebras and their dilute counterparts are only two of the families of diagrammatic algebras encountered in mathematics and physics. Other important variants of the Temperley-Lieb algebras include the affine (or periodic) families and the one- and two-boundary families. The one-boundary Temperley-Lieb algebra (or blob algebra) was introduced by Martin and Saleur \cite{blob} to describe statistical models wherein the physical field takes on various states along the domain boundary. A discussion of the one-boundary representation theory may be found in \cite{MarStr00}. Recent work \cite{MDRR} has shown that in order to properly account for certain examples of integrable boundary conditions, called Kac boundary conditions, one needs to consider instead a quotient of the one-boundary Temperley-Lieb algebra. By examining the bilinear forms on the standard modules of these quotients and studying numerical data of the statistical models, the authors inferred the structures of the Virasoro modules appearing in their conformal field-theoretic limits, finding that the Virasoro modules were identified as certain finitely generated submodules of Feigin-Fuchs modules. Even if these observations remain at the level of conjecture, the complexity of the modules over the one-boundary Temperley-Lieb algebras and their quotients seems to go well beyond that of the indecomposable modules that will be classified in the present paper. Hopefully, the methods developed here will also be fruitful for these other families.

The paper is organised as follows. \cref{sec:standardProjective} recalls the basic representation theory of the original and dilute Temperley-Lieb algebras. It also introduces (twisted) dual modules, gives the restriction and induction of the standard, irreducible and projective modules, and computes their Hom- and Ext-groups. \cref{sec:newFamilies} constructs recursively two families of new indecomposable modules and computes their extension groups with all the irreducibles. The calculations show that, together with the projectives and irreducibles, these new families form a complete list of non-isomorphic indecomposable modules (\cref{thm:classification}). \cref{sec:restrictionInduction} gives the restriction and induction of these new families of modules, relative to the inclusions described above. Finally, \cref{sec:ARquiver} repeats the construction of all indecomposable modules and the proof of their completeness using a more advanced tool, namely Auslander-Reiten theory. We hope that reading both \cref{sec:newFamilies,sec:ARquiver} in parallel will provide a comparison ground of standard methods for classifying indecomposable modules over finite-dimensional associative algebras and show their relative advantages.

%
%

\section*{Acknowledgements}

We thank Ingo Runkel for a careful reading of the manuscript.  DR also thanks Gus Lehrer for a discussion concerning cellular approaches to classifying indecomposable modules and Tony Licata for helpful pointers on quivers and zigzag algebras. JB holds scholarships from Fonds de recherche Nature et technologies (Qu\'ebec) and from the Facult\'e des \'etudes sup\'erieures et postdoctorales de l'Universit\'e de Montr\'eal, DR's research is partially funded by the Australian Research Council Discovery Projects DP1093910 and DP160101520, and YSA holds a grant from the Canadian Natural Sciences and Engineering Research Council. This support is gratefully acknowledged.

%
%

\section{Standard, irreducible and projective modules}\label{sec:standardProjective}

The definition of both the original and dilute Temperley-Lieb algebras, $\tl{n}(\beta)$ and $\dtl{n}(\beta)$, depends on a parameter $\beta$ taking values in a commutative ring. In what follows, this ring will always be the complex numbers $\CC$. Another parameter $q\in\mathbb C^\times$, related to the first by $\beta=q+q^{-1}$, is also used. The standard modules of these algebras may be introduced in several ways. In \cite{Martin}, bases for the standard modules are formed from walk diagrams, similar to Dyck paths. In \cite{Wenzl88,GoodWenzl93}, the ties between the Temperley-Lieb, Hecke and symmetric group algebras lead naturally to standard tableaux. Following early work \cite{KauInv90} on tangles and knots, both the algebra and the standard modules were given diagrammatic forms for which the action is simply concatenation of diagrams and the parameter $\beta$ appears when a closed loop is formed.  This resulted in simplified proofs of many structural results \cite{WesRep95} and led to the creation of the theory of cellular algebras \cite{GraCel96,GL-Aff}. This diagrammatic definition of standard modules was used in \cite{RSA,BSA} to construct the indecomposable projective modules --- these are also called the principal indecomposables --- and their irreducible quotients.  The present section summarises the properties upon which the construction of the remaining indecomposable modules and the computation of their restrictions and inductions will be based. (Many results are quoted in this section without proof; our presentation follows \cite{RSA,BSA} and we direct the reader to these works for the missing proofs.)

%
%
\subsection{Standards, irreducibles and projectives}\label{sub:basics}

We introduce a set of integers $\Lambda_n$ for each algebra $\tl{n}$ or $\dtl{n}$.  This set naturally parametrises the standard modules $\Stan{n,k}$, with $k \in \Lambda_n$. For $\dtl{n}$, this set is simply $\{0,1,\dots, n\}$; for $\tl{n}$, it is the subset of $\{0,1,\dots, n\}$ whose elements have the same parity as that of $n$. When $q$ is not a root of unity or when it is $\pm 1$, both algebras $\tl{n}(\beta)$ and $\dtl{n}(\beta)$ are semisimple. Then, the standard modules $\Stan{n,k}$, for $k\in\Lambda_n$, form a complete set of non-isomorphic irreducible modules of $\tl{n}$ ($\dtl{n}$). Since the algebras are semisimple, these irreducible modules exhaust the list of indecomposable modules. We will thus assume throughout that $q$ is a root of unity other than $\pm 1$.
\begin{center}
\emph{Let $\ell\ge 2$ be the smallest positive integer such that $q^{2\ell}=1$.}
\end{center}
Unless otherwise specified, all modules will be complex finite-dimensional left modules.

The set $\Lambda_n$ is partitioned as follows. If an element $k$ satisfies $k\equiv \ell-1\textrm{ mod }\ell$, then $k$ is said to be \emph{critical} and it forms its own class in the partition. If the element $k$ is not critical, then its class $[k]$ consists of the images (in $\Lambda_n$) of $k$ generated by reflections with respect to the critical integers. Here, if $k_c$ is a critical integer, then $2k_c-k$ is the reflection of $k$ through $k_c$.  The class of a non-critical $k$ thus contains precisely one integer between each pair of consecutive critical ones. We shall often need to refer to neighbouring elements in a non-critical class $[k]$. They will be ordered as $k_L<\dots < k^{--}<k^-<k<k^+<k^{++}<\dots <k_R$, so that $k_L\ge 0$ and $k_R\le n$ are the smallest and largest elements in $[k]\subset\Lambda_n$. The notation $k^j$ ($k^{-j}$) is also used to refer to the $j$-th element to the right of $k$ (to its left) so that, for example, $k^{--}=k^{-2}$ and $k^{+++}=k^3$. We shall often refer to non-critical classes as \emph{orbits}.

As an example of a partition, we take $\ell=4$ and $n=12$ so that the critical classes for $\dtl{12}$ are $[3]=\{3\}$, $[7]=\{7\}$ and $[11]=\{11\}$, whereas the (non-critical) orbits are $\{0,6,8\}$, $\{1,5,9\}$ and $\{2,4,10,12\}$. Note that the partition for $\tl{12}$, with $\ell=4$, consists of just two non-critical classes, namely $\{0,6,8\}$ and $\{2,4,10,12\}$. These are easily obtained from the diagram
\begin{equation*}
\begin{tikzpicture}[baseline={(current bounding box.center)},every node/.style={fill=white,rectangle,inner sep=2pt},scale=1/2]
	\draw[rounded corners] (4,0) -- (4,-0.75) -- (8,-0.75) -- (8,0) ;
	\draw[rounded corners] (20,0) -- (20,-0.75) -- (8,-0.75) -- (8,0) ;
	\draw[rounded corners] (24,0) -- (24,-0.75) -- (20,-0.75) -- (20,0) ;
	\draw[rounded corners] (2,0) -- (2,-1) -- (10,-1) -- (10,0) ;
	\draw[rounded corners] (18,0) -- (18,-1) -- (10,-1) -- (10,0) ;
	\draw[rounded corners] (0,0) -- (0,-1.25) -- (12,-1.25) -- (12,0) ;
	\draw[rounded corners] (16,0) -- (16,-1.25) -- (12,-1.25) -- (12,0) ;
	\draw[dashed] (6,-1.5) -- (6,0.9);
	\draw[dashed] (14,-1.5) -- (14,0.9);
	\draw[dashed] (22,-1.5) -- (22,0.9);
	\node at (0,0) {$0$};
	\node at (2,0) {$1$};
	\node at (4,0) {$2$};
	\node at (6,0) {$3$};
	\node at (8,0) {$4$};
	\node at (10,0) {$5$};
	\node at (12,0) {$6$};
	\node at (14,0) {$7$};
	\node at (16,0) {$8$};
	\node at (18,0) {$9$};
	\node at (20,0) {$10$};
	\node at (22,0) {$11$};
	\node at (24,0) {\phantom{,}$12$,};
\end{tikzpicture}
\end{equation*}
where the dashed lines indicate the critical $k$. In what follows, we shall also find it convenient to reflect about a critical $k$ when the result does not belong to $\Lambda_n$, extending the notation $k^-$, $k^+$, and so on, in the obvious fashion.  For instance, if $k=9$ in the above example, then $k^+=13 \notin \Lambda_n$.

The partition of the set $\Lambda_n$ into classes under reflection is intimately related to the existence of distinguished central elements. Both $\tl{n}$ and $\dtl{n}$ have a central element $F_n$ whose exact form will not be needed (see \cite{RSA} for its definition for $\tl{n}$ and \cite{BSA} for $\dtl{n}$). Its crucial property here is the following (\cite[Prop.~A.2]{RSA} and \cite[Prop.~B.3]{BSA}):
\begin{proposition}
The element $F_n$ acts as scalar multiplication by $\delta_k=q^{k+1}+q^{-k-1}$ on the standard module $\Stan{n,k}$.
\end{proposition}
\noindent It is easily shown that $k$ and $k'$ belong to the same orbit in $\Lambda_n$ if and only if they have the same parity and $\delta_k=\delta_{k'}$. Thus, all irreducible and standard modules labelled by an element of $[k]$ have the same $F_n$-eigenvalue. Note that $F_n$ has only one eigenvalue when acting on an arbitrary indecomposable module, although it need not then act as a multiple of the identity.

The dilute algebra $\dtl{n}$ has other important central elements. In particular, there are central idempotents $\eid$ and $\oid$ for which $\eid+\oid$ is the unit of $\dtl{n}$. (The subalgebra $\eid\cdot\dtl{n}$ ($\oid\cdot\dtl{n}$) contains elements whose diagrammatic representations have an even (odd) number of vacancies on each of their sides; see \cite{BSA}.) Any indecomposable module $\Mod{M}$ can be given a parity. This is defined to be even if $\eid\cdot \Mod{M}=\Mod{M}$ and odd otherwise. The standard module $\Stan{n,k}$ over $\dtl{n}$ has the parity of $n-k$.
\begin{proposition} The composition factors of an indecomposable $\tl{n}$- or $\dtl{n}$-module are irreducible modules $\Irre{n,k_i}$ whose labels $k_i$ belong to a single orbit $[k]$ and, for $\dtl{n}$, have the same parity.
\end{proposition}
\noindent We will construct these irreducible modules shortly.

It follows that if $\Mod{M}$ and $\Mod{N}$ are indecomposable, then $\Hom(\Mod{M},\Mod{N})=0$ whenever $\Mod{M}$ and $\Mod{N}$ have distinct parities (for $\dtl{n}$) or their $F_n$-eigenvalues are distinct.\footnote{We remark that the elements of Hom-groups will be understood, unless otherwise specified, to be $\tl{n}$- or $\dtl{n}$-module homomorphisms, as appropriate.}  Similarly, $\Ext(\Mod{M},\Mod{N})=0$ under the same conditions, as otherwise there would exist indecomposable modules possessing more than one $F_n$-eigenvalue or parity (see \cref{sec:Ext} for a primer on these extension groups and \cref{lem:trivialExtensions} for additional conditions that imply $\Ext(\Mod{M},\Mod{N}) =0$). This is how these central elements will be used hereafter.

\medskip

We now recall the theorems describing the basic structure of the standard modules. When not stated explicitly, these results hold for the standards of both $\tl{n}$ and $\dtl{n}$, as long as the index $k$ on $\Stan{n,k}$ belongs to the set $\Lambda_n$ of the algebra.
\begin{center}
\emph{We will use the symbol $\Alg{n}$ to stand for either $\tl{n}$ or $\dtl{n}$.  Moreover, we will generally omit the label $n$ on \\ the algebra $\Alg{n}$, its modules, and the set $\Lambda_n$, except when two different values are needed in the same statement.}
\end{center}

The structure of the standard modules $\Stan{k} \equiv \Stan{n,k}$ is conveniently investigated by employing a symmetric bilinear form $\bilin{\cdot}{\cdot}$ naturally defined on each.  This form is invariant with respect to the algebra action in the sense that there exists an involutive antiautomorphism ${}^{\dag}$ of $\Alg{} \equiv \Alg{n}$ such that
\begin{equation}
\bilin{x}{u \cdot y} = \bilin{u^{\dag} \cdot x}{y}, \qquad \text{for all \(u \in \Alg{}\) and \(x,y \in \Stan{k}\).}
\end{equation}
We refer to \cite[Sec.~3]{RSA}, for $\Alg{} = \tl{}$, and \cite[Sec.~4.1]{BSA}, for $\Alg{} = \dtl{}$, for the definitions of $\bilin{\cdot}{\cdot}$ and ${}^{\dag}$, mentioning only that ${}^{\dag}$ corresponds diagrammatically to a reflection. The point is that the invariance of this bilinear form means that its radical $\Radi{k}$ is a submodule of $\Stan{k}$.  We denote the quotient by $\Irre{k} = \Stan{k} / \Radi{k}$.
\begin{proposition} \label{prop:Standard}
\leavevmode
\begin{enumerate}
\item The standard modules $\Stan{k}$ are indecomposable with $\Hom(\Stan{k},\Stan{k}) \cong \CC$. \label{it:StanIndec}
\item $\Radi{k}$ is the maximal proper submodule of $\Stan{k}$, for all $k \in \Lambda$, unless $\Alg{} = \tl{}$, $\beta = 0$ and $k=0$.  In this latter case, the form $\bilin{\cdot}{\cdot}$ is identically zero, so $\Radi{0} = \Stan{0}$ and $\Irre{0} = 0$. \label{it:RadMax}
\item If $k$ is critical, then the form $\bilin{\cdot}{\cdot}$ is non-degenerate, so $\Radi{k} = 0$ and $\Stan{k} = \Irre{k}$ is irreducible. \label{it:StanIrr}
\item If $k$ is non-critical, then $\Radi{k} \cong \Irre{k^+}$, for all $k<k_R$, and $\Radi{k_R} = 0$. The standard module $\Stan{k_R}$ is therefore irreducible:  $\Stan{k_R} \cong \Irre{k_R}$. \label{it:RadIrr}
\end{enumerate}
\end{proposition}
\noindent We remark that unless $\Alg{} = \tl{}$, $\beta = 0$ and $k=0$, the module $\Irre{k}$ is irreducible (and non-zero). In fact, all the irreducible modules, up to isomorphism, may be obtained in this fashion.
\begin{proposition} \label{prop:StComplete}
For all $n\ge 1$ and $\beta \in \CC$, the irreducible modules $\Irre{k}$, with $k\in\Lambda$, form a complete set of non-isomorphic irreducible modules, except in the case of $\tl{n}$ with $\beta=0$ and $n$ even. In this latter case, $\Irre{0} = 0$ and the set of irreducibles $\Irre{k}$ with $k\in\Lambda \setminus \set{0}$ is complete.
\end{proposition}
\noindent We deal with the annoying possibility that $\Irre{k}$ may be zero by  introducing \cite{GraCel96} a set $\Lambda_{n,0} \equiv \Lambda_0$ which is defined to be $\Lambda \setminus \set{0}$, if $\Alg{} = \tl{}$, $\beta=0$ and $n$ is even, and $\Lambda$ otherwise.  The set $\set{\Irre{k} \st k \in \Lambda_0}$ is then a complete set of non-isomorphic irreducible modules.

We conclude our discussion of standard modules by noting that \cref{prop:Standard} implies that the following sequence is exact and non-split for $k$ non-critical and $k<k_R$:
\begin{equation}\label{eq:StanExact}
\dses{\Irre{k^+}}{}{\Stan{k}}{}{\Irre{k}}.
\end{equation}
This sequence is also exact for $k=k_R$ if we assume that any module with a label $k \notin \Lambda$ is zero. For $k=k_R$, so $k^+ \notin \Lambda$, it reads $\ses{0}{\Stan{k_R}}{\Irre{k_R}}$, which simply states that $\Stan{k_R}$ is its own irreducible quotient.  With this understanding, the composition factors of the non-critical standard module $\Stan{k}$ are therefore $\Irre{k}$ and $\Irre{k^+}$.  This leads to a quick proof of the following result, needed in \cref{sub:restriction}.
\begin{proposition} \label{prop:HomSkSk-}
For each non-critical $k \in \Lambda$ with $k \neq k_L$, $\Hom(\Stan{k},\Stan{k^-}) \cong \CC$.
\end{proposition}
\begin{proof}
By \eqref{eq:StanExact}, $\Stan{k}$ and $\Stan{k^-}$ only have one composition factor in common:  $\Irre{k}$.  It follows that if $f \in \Hom(\Stan{k},\Stan{k^-})$, then $\Stan{k} / \ker f \cong \im f$ is either $0$ or isomorphic to $\Irre{k}$.  Any non-zero $f$ is therefore unique up to rescaling.  But, $\Irre{k}$ is a quotient of $\Stan{k}$ and a submodule of $\Stan{k^-}$, so there exists a non-zero $f$ given by the composition $\Stan{k} \sra \Irre{k} \cong \Radi{k^-} \ira \Stan{k^-}$.
\end{proof}

\medskip

We turn now to the structure of the projective modules.  For each $k \in \Lambda_0$, let $\Proj{k} \equiv \Proj{n,k}$ be the projective cover of the irreducible module $\Irre{k}$. (The module $\Proj{k}$ is thus indecomposable.) When $\Irre{k} = 0$, thus $\Alg{} = \tl{}$ with $n$ even and $\beta = k = 0$, we could define $\Proj{k} = 0$ as well, but this is more trouble than it is worth.
\begin{proposition} \label{prop:Proj}
\leavevmode
\begin{enumerate}
\item If $k \in \Lambda_0$ is critical, then the irreducible $\Irre{k}$ is projective: $\Irre{k}\simeq\Stan{k}\simeq\Proj{k}$.
\item If $k \in \Lambda_0$ is non-critical, then the following sequence is exact:
\begin{equation}\label{eq:ProjExact}
\dses{\Stan{k^{-}}}{}{\Proj{k}}{}{\Stan{k}}.
\end{equation}
\end{enumerate}
\end{proposition}
\noindent If $k=k_L$, then $k^- \notin \Lambda$, so $\Stan{k^-}$ is understood to be zero. The sequence, in this case, reads $\ses{0}{\Proj{k_L}}{\Stan{k_L}}$, indicating that $\Stan{k_L}$ is projective and, due to \eqref{eq:StanExact}, has (at most) two composition factors, $\Irre{k}$ and $\Irre{k^+}$. If $k=k_R \neq k_L$, then $\Stan{k_R}\simeq\Irre{k_R}$, by \cref{prop:Standard}\ref{it:RadIrr}, and $\Proj{k_R}$ has exactly three composition factors, $\Irre{k^-}$ and two copies of $\Irre{k}$. For all other non-critical $k\in\Lambda_0$, the projective $\Proj{k}$ has precisely four composition factors, namely $\Irre{k^-}$, $\Irre{k}$ (twice) and $\Irre{k^+}$.\footnote{When $\Alg{} = \tl{}$ with $n$ even and $\beta=0$, these two statements must be modified for $\Proj{2}$ because $\Irre{0} = 0$:  $\Proj{2,2}$ has only two composition factors and, for $n>2$, $\Proj{n,2}$ has three. The sequence \eqref{eq:ProjExact} remains exact and non-split in both these cases.} We will discuss this further in \cref{sub:Loewy}.

Finally, we record the following result for future use.
\begin{proposition} \label{prop:HomSP}
For all $k \in \Lambda_0$, $\Hom(\Stan{k},\Proj{k}) \cong \CC$.
\end{proposition}
\begin{proof}
If $k$ is critical or $k=k_L$, then $\Proj{k} \cong \Stan{k}$, so $\Hom(\Stan{k},\Proj{k}) \cong \CC$ by \cref{prop:Standard}\ref{it:StanIndec}.  So, we may assume that $k$ is non-critical with $k>k_L$.  Then, $\Proj{k}$ has precisely two composition factors isomorphic to $\Irre{k}$.

Let $f \in \Hom(\Stan{k},\Proj{k})$ and consider the composition $\Stan{k} \xrightarrow{f} \Proj{k} \overset{\pi}{\sra} \Irre{k}$, where $\pi$ is the canonical quotient map.  If $\pi f \neq 0$, then it is a surjection, hence there exists a homomorphism $g$ making
\begin{equation}
\begin{tikzpicture}[baseline={(current bounding box.center)},scale=1.5]
\node (S) at (0,0) [] {\(\Stan{k}\)};
\node (I) at (1,0) [] {\(\Irre{k}\)};
\node (O) at (2,0) [] {\(0\)};
\node (P) at (0,1) [] {\(\Proj{k}\)};
\draw [->] (S) -- node [below] {\(\scriptstyle \pi f\)} (I);
\draw [->] (I) -- (O);
\draw [->] (P) -- node [right] {\(\scriptstyle \pi\)} (I);
\draw [dashed,->] (P) -- node [left] {\(\scriptstyle g\)} (S);
\end{tikzpicture}
\end{equation}
commute, by the projectivity of $\Proj{k}$.  Now, $\pi fg = \pi$ requires that $f$ be surjective, as $\ker \pi$ is the maximal proper submodule of $\Proj{k}$.  However, this is impossible as $\Proj{k}$ has more composition factors than $\Stan{k}$.

We must therefore have $\pi f = 0$, hence that $\im f \subseteq \ker \pi$.  Now, $\ker \pi$ has only one composition factor isomorphic to $\Irre{k}$ and it is a submodule:  $\Irre{k} \cong \Radi{k^-} \subseteq \Stan{k^-} \subseteq \ker \pi$.  Thus, $\im f$ is either $0$ or isomorphic to $\Irre{k}$, by \eqref{eq:StanExact}.  $f$ is therefore unique up to rescaling and the (non-zero) composition $\Stan{k} \sra \Irre{k} \ira \Stan{k^-} \ira \Proj{k}$ gives the result.
\end{proof}

%
%
\subsection{Their restriction and induction}\label{sub:restriction}

This subsection describes how the families of modules introduced so far behave under restriction and induction. We first fix an inclusion of algebras $\Alg{n} \ira \Alg{n+1}$ whose image is the subalgebra of diagrams whose $(n+1)$-th nodes are connected by an identity strand (see \cite[Sec.~4]{RSA} and \cite[Sec.~3.4]{BSA} for details).  If $\Mod{M}$ is an $\Alg{n}$-module, its restriction to an $\Alg{n-1}$-module will be denoted by $\Res{\Mod{M}}$ and its induction to an $\Alg{n+1}$-module by $\Ind{\Mod{M}}$.
\begin{proposition}\label{prop:StanRest}
\leavevmode
\begin{enumerate}
\item The restriction of the standard module $\Stan{n+1,k}$ satisfies the exact sequence \label{it:SResES}
\begin{subequations}
\begin{gather}
\dses{\Stan{n,k-1}}{}{\Res{\Stan{n+1,k}}}{}{\Stan{n,k+1}} \phantom{\oplus \Stan{n,k}} \qquad \text{(\(\Alg{n}=\tl{n}\)),} \label{eq:StanRestdTLn} \\
\dses{\Stan{n,k-1}\oplus\Stan{n,k}}{}{\Res{\Stan{n+1,k}}}{}{\Stan{n,k+1}} \qquad \text{(\(\Alg{n}=\dtl{n}\)).} \label{eq:StanRestTLn}
\end{gather}
\end{subequations}
\item For all non-critical $k$, these sequences split: \label{it:SResSplit}
\begin{equation}
\Res{\Stan{n+1,k}} \cong
\begin{cases*}
\Stan{n,k-1}\oplus\Stan{n,k+1}, & if \(\Alg{n} = \tl{n}\), \\
\Stan{n,k-1}\oplus\Stan{n,k}\oplus\Stan{n,k+1}, & if \(\Alg{n} = \dtl{n}\).
\end{cases*}
\end{equation}
\item For critical $k$, the result is almost always projective:
\begin{equation}
\Res{\Stan{n+1,k}} \cong
\begin{cases*}
\Proj{n,k+1}, & if \(\Alg{n} = \tl{n}\) and \(k \neq n+1\), \\
\Proj{n,k} \oplus \Proj{n,k+1}, & if \(\Alg{n} = \dtl{n}\) and \(k \neq n, n+1\).
\end{cases*}
\end{equation}
The exceptions are $\Res{\Stan{n+1,n+1}} \cong \Stan{n,n}$ for $\tl{n}$ and $\dtl{n}$, as well as $\Res{\Stan{n+1,n}} \cong \Stan{n,n-1} \oplus \Stan{n,n}$ for $\dtl{n}$. \label{it:SResCrit}
\item The induction and restriction of standard modules are related by \label{it:SResInd}
\begin{equation}\label{eq:RestIndu}
\Ind{\Stan{n-1,k}}\simeq\Res{\Stan{n+1,k}}, \qquad \text{for all }n\ge 2,
\end{equation}
except when $\beta=0$ for the module $\Ind{\Stan{2,0}}$ over $\tl{3}$. In that case,
\begin{equation}
\Ind{\Stan{2,0}} \cong \Stan{3,1} \oplus \Stan{3,3}, \qquad \text{but} \qquad
\Res{\Stan{4,0}} \cong \Stan{3,1}.
\end{equation}
\end{enumerate}
\end{proposition}
\noindent Again, if one of the indices of the direct summands does not belong to $\Lambda_n$, then this summand is understood to be $0$.  We remark that the submodule $\Stan{n,k}$ of $\Res{\Stan{n+1,k}}$, for $\Alg{n} = \dtl{n}$, is always a direct summand, because its parity differs from that of $\Stan{n,k-1}$ and $\Stan{n,k+1}$.  This proposition appeared in \cite{RSA} for $\tl{n}$ and in \cite{BSA} for $\dtl{n}$ (see also \cite{WesRep95} where a part of the proposition was first stated). Analogous results were proved in \cite{BSA} for the restriction and induction of the irreducibles of $\dtl{n}$. The proofs can be extended straightforwardly to $\tl{n}$.
\begin{proposition}\label{prop:IrreRest}
Suppose that $k \in \Lambda_{n+1,0}$.
\begin{enumerate}
\item If $\Alg{n}=\dtl{n}$ and $\Radi{n+1,k} \neq 0$, then
\begin{equation} \label{eq:IrreRest}
\Res{\Irre{n+1,k}} \simeq
\begin{cases*}
\Irre{n,k-1} \oplus \Irre{n,k}, & if \(k+1\) is critical, \\
\Irre{n,k-1} \oplus \Irre{n,k} \oplus \Irre{n,k+1}, & otherwise.
\end{cases*}
\end{equation}
If, moreover, $\Radi{n-1,k} \neq 0$, then $\Ind{\Irre{n-1,k}}\simeq\Res{\Irre{n+1,k}}$.
\item If $\Alg{n}=\tl{n}$, then the same statements hold if one deletes the $\Irre{n,k}$ appearing on the right-hand side of \eqref{eq:IrreRest}.
\end{enumerate}
\end{proposition}
\noindent Of course, $\Radi{n \pm 1,k} = 0$ implies that $\Irre{n \pm 1,k} = \Stan{n \pm 1,k}$, so these results were already given in \cref{prop:StanRest}.  This proposition also gives the behaviour of the critical principal indecomposables under restriction and induction, because then $\Proj{n,k}\simeq\Stan{n,k}$. The non-critical cases are rather more delicate and are covered by the following result.
%
%
\begin{proposition}\label{prop:ProjRest}
\leavevmode
\begin{enumerate}
\item For non-critical $k \in \Lambda_{n+1,0}$, the principal indecomposables satisfy
\begin{equation}\label{eq:ProjRest}
\Res{\Proj{n+1,k}}\simeq
\begin{cases*}
\Proj{n,k-1}\oplus\Proj{n,k+1}, & if \(\Alg{n}=\tl{n}\), \\
\Proj{n,k-1}\oplus\Proj{n,k}\oplus\Proj{n,k+1}, & if \(\Alg{n}=\dtl{n}\).
\end{cases*}
\end{equation}
Here, any modules with indices not in $\Lambda_{n,0}$ are not set immediately to zero. We first make the following corrections to the right-hand side of the above formula:
\begin{itemize}
\item If $k+1$ is critical, then $\Proj{n,k+1}$ is replaced by $\Proj{n,k+1}\oplus \Proj{n,k^--1} \cong \Stan{n,k+1}\oplus \Stan{n,k^--1}$.
\item If $k-1$ is critical, then $\Proj{n,k-1}$ is replaced by $2 \: \Proj{n,k-1} \cong 2 \: \Stan{n,k-1}$.
\item Any remaining $\Proj{n,k'}$ with $k'>n$ non-critical is replaced by $\Stan{n,k'^{-}}$.
\end{itemize}
Now, any modules with indices not in $\Lambda_{n,0}$ are set to zero. \label{it:ProjRest}
\item In all cases (with $n\ge 2$), $\Ind{\Proj{n-1,k}}\simeq\Res{\Proj{n+1,k}}$. \label{it:ProjRestInd}
\end{enumerate}
\end{proposition}
\begin{proof}
We work with $\dtl{n}$ for definiteness, the argument for $\tl{n}$ being identical after removing modules whose indices are not in $\Lambda_n$.  We will moreover omit the routine checks that nothing untoward happens for $k = k_L$, $k_R$, as long as $k+1 \le n$.

Since restriction is an exact functor, \eqref{eq:ProjExact} gives the exactness of
\begin{equation}
\dses{\Res{\Stan{n+1,k^-}}}{}{\Res{\Proj{n+1,k}}}{}{\Res{\Stan{n+1,k}}},
\end{equation}
hence, by \eqref{eq:StanExact}, that of
\begin{equation}\label{eq:star}\dses{\Stan{n,k^--1}\oplus\Stan{n,k^-}\oplus\Stan{n,k^-+1}}{}{\Res{\Proj{n+1,k}}}{}{\Stan{n,k-1}\oplus\Stan{n,k}\oplus\Stan{n,k+1}}.
\end{equation}
Considerations of parity and $F_n$-eigenvalues now force $\Res{\Proj{n+1,k}}$ to decompose as ${\Mod{M}}_{-1}\oplus {\Mod{M}}_{0}\oplus {\Mod{M}}_{+1}$, where $F_n$ has eigenvalue $\delta_{k+i}$ on ${\Mod{M}}_i$, $i\in\{-1,0,+1\}$. The exact sequence \eqref{eq:star} therefore decomposes into three exact sequences:
\begin{equation} \label{es:ProjRest}
\dses{\Stan{n,k^--i}}{}{{\Mod{M}}_{i}}{}{\Stan{n,k+i}},\qquad i \in \set{-1,0,+1}.
\end{equation}
Note that $(k+i)^-=k^--i$ since neighbouring elements of an orbit are obtained from one another by reflection (see the beginning of \cref{sub:basics}).  The goal is to prove that ${\Mod{M}}_i\simeq \Proj{n,k+i}$, for each $i\in\{-1,0,+1\}$, taking into account the replacements noted in the statement of the proposition. These replacements are easily dealt with:  Let $i = \pm 1$ and assume that $k+i$ is critical with $\ell\neq 2$. Then, the exact sequence \eqref{es:ProjRest} shows that $\Mod{M}_i$ is $\Stan{n,k+1}\oplus \Stan{n,k^--1}$, if $i=+1$, and a direct sum of two copies of the projective $\Stan{n,k-1} = \Stan{n,k^-+1}$, if $i=-1$. Furthermore, if $k+i > n$, then $\Stan{n,k+i} \simeq 0$ and \eqref{eq:star} tells us that $\Res{\Proj{n+1,k}}$ contains $\Stan{n-1,k^--i}$. In the case $\ell=2$, the four summands $\Stan{n,k^--1}$, $\Stan{n,k^-+1}\simeq \Stan{n,k-1}$ and $\Stan{n,k+1}$ of \eqref{eq:star} are all projective and they will thus be direct summands of ${\Res{\Proj{n+1,k}}}$.  This takes care of the replacements.  From now on, we will therefore fix $i \in \set{-1,0,+1}$ and assume that $k+i\le n$ is not critical.

Consider the following diagram in which the rows are exact by \eqref{eq:ProjExact} and \eqref{es:ProjRest}:
\begin{equation}
\begin{tikzpicture}[baseline={(current bounding box.center)},scale=1/3]
\node (t1) at (5,5) [] {\(0\)};
\node (t2) at (10,5) [] {\(\Stan{n,k^--i}\)};
\node (t3) at (15,5) [] {\(\Proj{n,k+i}\)};
\node (t4) at (20,5) [] {\(\Stan{n,k+i}\)};
\node (t5) at (25,5) [] {\(0\)\phantom{.}};
\node (b1) at (5,0) [] {\(0\)};
\node (b2) at (10,0) [] {\(\Stan{n,k^--i}\)};
\node (b3) at (15,0) [] {\(\Mod{M}_i\)};
\node (b4) at (20,0) [] {\(\Stan{n,k+i}\)};
\node (b5) at (25,0) [] {\(0\).};
\draw [->] (t1) -- (t2);
\draw [->] (t2) to node [above] {\(\scriptstyle \alpha\)} (t3);
\draw [->] (t3) to node [above] {\(\scriptstyle \beta\)}  (t4);
\draw [->] (t4) -- (t5);
\draw [->] (b1) -- (b2);
\draw [->] (b2) to node [above] {\(\scriptstyle \gamma\)}  (b3);
\draw [->] (b3) to node [above] {\(\scriptstyle \delta\)}  (b4);
\draw [->] (b4) -- (b5);
\draw [dashed,->] (t2) to node [left] {\(\scriptstyle g\)} (b2);
\draw [dashed,->] (t3) to node [left] {\(\scriptstyle f\)} (b3);
\draw [->] (t4) to node [right] {\(\scriptstyle \id\)} (b4);
\draw [dashed,->] (t4) to node [left] {\(\scriptstyle h\)} (b3);
\end{tikzpicture}
\end{equation}
A map $f$ must exist making the right square commute, $\delta f=\beta$, as $\Proj{n,k+i}$ is projective.  Moreover, $f$ then satisfies $\delta f\alpha=\beta\alpha=0$. As $\im f\alpha \subseteq \ker \delta = \im \gamma$ and $\gamma$ is injective, one may construct a unique homomorphism $g$ that also makes the left square commute. Because of \cref{prop:Standard}\ref{it:StanIndec}, the map $g$ is either $0$ or an isomorphism. If it is the latter, then the short five lemma implies that $f$ is also an isomorphism, hence that $\Mod{M}_i\simeq \Proj{n,k+i}$ as desired. The proof of \ref{it:ProjRest} will thus be complete if the case $g=0$ can be ruled out.

Suppose then that $g=0$. Because $0=\gamma g=f\alpha$, we have $\ker f\supseteq\im \alpha=\ker\beta$, so there exists a homomorphism $h$ such that $h\beta=f$. The commutativity of the right square then gives $\delta h\beta=\delta f=\beta$, so that $\delta h$ acts as the identity on $\im \beta = \Stan{n,k+i}$.  Therefore, $h$ splits the bottom row so that $\Mod{M}_i\cong\Stan{n,k^--i}\oplus\Stan{n,k+i}$.  It follows that
\begin{equation} \label{eq:ProjRestCon}
\Hom(\Stan{n,k+i},\Mod{M}_i)\simeq\Hom(\Stan{n,k+i},\Stan{n,k^--i})\oplus\Hom(\Stan{n,k+i},\Stan{n,k+i})=\mathbb C\oplus \mathbb C,
\end{equation}
the first Hom-group following from \cref{prop:HomSkSk-} and the second following from \cref{prop:Standard}\ref{it:StanIndec}.  However,
\begin{align}
\Hom(\Stan{n,k+i},{\Mod{M}}_i)&\overset{(1)}\simeq\Hom(\Stan{n,k+i},{\Mod{M}}_{-1}\oplus{\Mod{M}}_0\oplus{\Mod{M}}_1)\simeq\Hom(\Stan{n,k+i},\Res{\Proj{n+1,k}})\notag\\
&\overset{(2)}\simeq\Hom(\Ind{\Stan{n,k+i}},{\Proj{n+1,k}})\simeq\Hom(\Stan{n+1,k+i-1}\oplus\Stan{n+1,k+i}\oplus\Stan{n+1,k+i+1},{\Proj{n+1,k}})\notag\\
&\overset{(3)}\simeq \Hom(\Stan{n+1,k},{\Proj{n+1,k}}),
\end{align}
where, for $(1)$, adding the two other modules does not change the Hom-group as their parities or $F_{n+1}$-eigenvalues are different, $(2)$ is Frobenius reciprocity, and $(3)$ again follows from parity and $F_n$-eigenvalue considerations. But, \cref{prop:HomSP} gives $\Hom(\Stan{n,k+i},{\Mod{M}}_i)\simeq\Hom(\Stan{n+1,k},{\Proj{n+1,k}})\simeq\mathbb C$, contradicting \eqref{eq:ProjRestCon}.  This rules out $g=0$.

For \ref{it:ProjRestInd}, the isomorphism of $\Ind{\Proj{n-1,k}}$ and $\Res{\Proj{n+1,k}}$ follows by comparing with the induction results of \cite{RSA,BSA}.
\end{proof}
Recall that induction functors are always right-exact.  Here, we show that the functor $\Ind{}$ is not left-exact in general.  Choose $k\in \Lambda_n$ such that $k^{++} = n+1$ or $n+2$, as parity dictates, so that $\Stan{n,k^{+}} \simeq \Irre{n,k^{+}}$ and $\ses{\Stan{n,k^{+}}}{\Stan{n,k}}{\Irre{n,k}}$ is exact.  Inducing results in the following exact sequence:
\begin{equation} \label{es:IndNotExact}
\rdses{\Ind{\Stan{n,k^{+}}}}{}{\Ind{\Stan{n,k}}}{}{\Ind{\Irre{n,k}}}.
\end{equation}
By \cref{prop:StanRest}, $\Ind{\Stan{n,k^{+}}}$ has a direct summand isomorphic to $\Stan{n+1,k^{+}+1}$ which has a submodule isomorphic to $\Irre{n+1,k^{++}-1}$.  Because $k^{++} = n+1$ or $n+2$, this submodule is $\Irre{n+1,n}$ or $\Irre{n+1,n+1}$, hence is non-zero. However, the same proposition shows that $\Ind{\Stan{n,k}}$ does not have any composition factor isomorphic to $\Irre{n+1,n}$ or $\Irre{n+1,n+1}$, so the left-most map of \eqref{es:IndNotExact} cannot be an inclusion and we conclude that $\Ind{}$ is not left-exact.

%
%
\subsection{Their duals} \label{sec:Dual}

Given a left module $\Mod{M}$ over $\Alg{} \equiv \Alg{n}$, the dual vector space $\Hom_{\CC} (\Mod{M},\CC)$ is naturally equipped with the structure of a right $\Alg{}$-module which we shall denote by $\dual{\Mod{M}}$.  However, the involutive antiautomorphism ${}^{\dag}$ may be used to twist this action to obtain a left module structure on $\Hom (\Mod{M},\CC)$; the corresponding twisted dual will be denoted by $\twdu{\Mod{M}}$.  Explicitly, this twisted action is
\begin{equation}
\func{\brac{uf}}{m} = \func{f}{u^{\dag} m}; \qquad u \in \Alg{}, \quad f \in \twdu{\Mod{M}}, \quad m \in \Mod{M}.
\end{equation}
We will have occasion to consider (untwisted) dual modules $\dual{\Mod{M}}$ in \cref{sec:restrictionInduction,sec:ARquiver}.  Until then, it is understood that when we speak of a dual module, it is the twisted dual that we are referring to. As above, all modules are therefore assumed to be complex finite-dimensional left modules, unless otherwise specified.

\begin{proposition} \label{prop:IrreDual}
Every irreducible $\Alg{}$-module is self-dual:  $\twdu{\Irre{k}} \cong \Irre{k}$, for all $k \in \Lambda_0$.
\end{proposition}
\begin{proof}
By \cref{prop:Standard,prop:StComplete}, every irreducible module may be constructed as the quotient of a standard module by the radical of its invariant bilinear form.  Therefore, every irreducible module carries an invariant \emph{non-degenerate} bilinear form $\bilin{\cdot}{\cdot}$.

For each $x \in \Irre{}$, define $f_x \in \twdu{\Irre{}}$ by $\func{f_x}{y} = \bilin{x}{y}$.  Then, $x \mapsto f_x$ is a module homomorphism:
\begin{equation}
\func{\brac{u f_x}}{y} = \func{f_x}{u^{\dag} y} = \bilin{x}{u^{\dag} y} = \bilin{ux}{y} = \func{f_{ux}}{y}.
\end{equation}
Moreover, it is injective as $f_x = 0$ implies that $0 = \func{f_x}{y} = \bilin{x}{y}$ for all $y \in \Irre{}$, hence $x=0$ by non-degeneracy.  Since $\dim \twdu{\Irre{}} = \dim \Irre{}$, this map is an isomorphism.
\end{proof}

\begin{proposition} \label{prop:DoubleDual}
Duality is reflexive:  $\Mod{M} \cong \twdu{(\twdu{\Mod{M}})}$.
\end{proposition}
\begin{proof}
The proof is quite similar.  Define the map $x \in \Mod{M} \mapsto \phi_x \in \twdu{(\twdu{\Mod{M}})}$ by $\func{\phi_x}{f} = \func{f}{x}$, for all $f \in \twdu{\Mod{M}}$.  This is a module homomorphism,
\begin{equation}
\func{\brac{u \phi_x}}{f} = \func{\phi_x}{u^{\dag} f} = \func{\brac{u^{\dag} f}}{x} = \func{f}{u^{\dag \dag} x} = \func{f}{ux} = \func{\phi_{ux}}{f},
\end{equation}
and it is injective,
\begin{equation}
\phi_x = 0 \quad \Rightarrow \quad \func{\phi_x}{f} = 0 \quad \text{for all \(f \in \twdu{\Mod{M}}\)} \quad \Rightarrow \quad \func{f}{x} = 0 \quad \text{for all \(f \in \twdu{\Mod{M}}\)} \quad \Rightarrow \quad x=0,
\end{equation}
hence it is an isomorphism for dimensional reasons.
\end{proof}

\begin{proposition} \label{prop:DualExact}
Duality is an exact contravariant functor:  The sequence
\begin{equation} \label{es:Generic}
\dses{\Mod{L}}{\iota}{\Mod{M}}{\pi}{\Mod{N}}
\end{equation}
is exact if and only if the sequence
\begin{equation}
\dses{\twdu{\Mod{N}}}{\twdu{\pi}}{\twdu{\Mod{M}}}{\twdu{\iota}}{\twdu{\Mod{L}}}
\end{equation}
is exact, where $\twdu{\iota}$ and $\twdu{\pi}$ are defined by
\begin{equation}
\func{\func{\twdu{\iota}}{g}}{\ell} = \func{g}{\func{\iota}{\ell}}, \quad \text{for all \(\ell \in \Mod{L}\) and \(g \in \twdu{\Mod{M}}\);} \qquad
\func{\func{\twdu{\pi}}{h}}{m} = \func{h}{\func{\pi}{m}}, \quad \text{for all \(m \in \Mod{M}\) and \(h \in \twdu{\Mod{N}}\).}
\end{equation}
\end{proposition}
\begin{proof}
Assume that \eqref{es:Generic} is exact.  We first check that $\twdu{\pi}$ is a module homomorphism.  This follows because
\begin{equation}
\func{\brac{u \func{\twdu{\pi}}{h}}}{m} = \func{\func{\twdu{\pi}}{h}}{u^{\dag} m} = \func{h}{\func{\pi}{u^{\dag} m}} = \func{h}{u^{\dag} \func{\pi}{m}} = \func{\brac{uh}}{\func{\pi}{m}} = \func{\func{\twdu{\pi}}{uh}}{m},
\end{equation}
for all $m \in \Mod{M}$ and $h \in \twdu{\Mod{N}}$.  The check for $\twdu{\iota}$ is similar.

Second, we prove that $\twdu{\pi}$ is injective.  Take $h \in \ker \twdu{\pi}$, so that $0 = \func{\func{\twdu{\pi}}{h}}{m} = \func{h}{\func{\pi}{m}}$ for all $m \in \Mod{M}$.  This implies that $h=0$ because $\pi$ is a surjection.

Third, we show that $\ker \twdu{\iota} = \im \twdu{\pi}$.  Exactness gives $\func{\func{\twdu{\iota}}{\func{\twdu{\pi}}{h}}}{\ell} = \func{\func{\twdu{\pi}}{h}}{\func{\iota}{\ell}} = \func{h}{\func{\pi}{\func{\iota}{\ell}}} = 0$, for all $\ell \in \Mod{L}$ and $h \in \twdu{\Mod{N}}$.  Thus, $\im \twdu{\pi} \subseteq \ker \twdu{\iota}$.  To prove the reverse inclusion, we take $g \in \ker \twdu{\iota}$, so that $0 = \func{\func{\twdu{\iota}}{g}}{\ell} = \func{g}{\func{\iota}{\ell}}$ for all $\ell \in \Mod{L}$.  Thus, $g$ annihilates $\im \iota = \ker \pi$.  Now, define a functional $h \in \twdu{\Mod{N}}$ by $\func{h}{n} = \func{g}{m}$, where $\func{\pi}{m} = n$.  This is well defined, because $\func{\pi}{m} = \func{\pi}{m'}$ implies that $m-m' \in \ker \pi$, hence that $\func{g}{m-m'} = 0$, thus $\func{g}{m} = \func{g}{m'}$.  But, $\func{\func{\twdu{\pi}}{h}}{m} = \func{h}{\func{\pi}{m}} = \func{g}{m}$ for all $m \in \Mod{M}$, so we see that $g \in \im \twdu{\pi}$.  This proves that $\ker \twdu{\iota} \subseteq \im \twdu{\pi}$.

Last, $\twdu{\iota}$ is surjective by the rank-nullity theorem of linear algebra:
\begin{align}
\dim \im \twdu{\iota} &= \dim \twdu{\Mod{M}} - \dim \ker \twdu{\iota} = \dim \Mod{M} - \dim \im \twdu{\pi} = \dim \Mod{M} - \dim \twdu{\Mod{N}} = \dim \Mod{M} - \dim \Mod{N} = \dim \Mod{L} \notag \\
&= \dim \twdu{\Mod{L}}.
\end{align}
This completes the proof, \cref{prop:DoubleDual} providing the converse.
\end{proof}
\noindent Indeed, taking (twisted) duals defines an autoequivalence of the category of finite-dimensional $\Alg{}$-modules.  Such equivalences are always exact.

\medskip

We will refer to the duals of the standard modules as \emph{costandard modules}, denoting them by $\Cost{k} = \twdu{\Stan{k}}$. From \cref{prop:IrreDual,prop:DualExact}, we learn that the dual of the exact sequence \eqref{eq:StanExact} is
\begin{equation} \label{es:ICI}
\dses{\Irre{k}}{}{\Cost{k}}{}{\Irre{k^+}}
\end{equation}
and that this sequence is also exact and non-split, by \cref{prop:DoubleDual}.  Note that this failure to split uses the basic fact that $\twdu{(\Mod{M} \oplus \Mod{N})} \simeq \twdu{\Mod{M}} \oplus \twdu{\Mod{N}}$ as modules.  It follows that a standard module $\Stan{k}$ is self-dual if and only if it is irreducible.  The corresponding result for the principal indecomposables is as follows:
\begin{proposition} \label{prop:ProjDual}
The critical principal indecomposables are self-dual, as are the non-critical $\Proj{k}$ with $k>k_L$.  The non-critical $\Proj{k_L}$ are self-dual if and only if they are irreducible.
\end{proposition}
\begin{proof}
The case where $k$ is critical or where $k$ is non-critical with $k=k_L$, hence $\Proj{k} \cong \Stan{k}$, has already been dealt with.  We therefore assume that $k$ is critical and prove that $\Proj{k+i}$ is self-dual, for all $n\ge k+i$, where $i=1, \dots, \ell-1$.  The proof is by induction and we shall detail it for $\Alg{} = \dtl{}$, omitting the simple modifications required for $\tl{}$.

The base cases of this induction are $i=0$ and $1$.  The former is the critical case already dealt with, so we turn to $i=1$.  The key tool for this case, and indeed for the extension to higher $i$, is the realisation that duality commutes with restriction: $\Res{(\twdu{\Mod{M}})}$ and $\twdu{(\Res{\Mod{M}})}$ are the same vector space with the same algebra action.  For $i=1$, we use the criticality of $\Proj{k}$ and \cref{prop:StanRest}\ref{it:SResCrit} twice to arrive at
\begin{equation}
\twdu{\Proj{n,k+1}} \oplus \Proj{n,k} \cong \twdu{\Proj{n,k+1}} \oplus \twdu{\Proj{n,k}} \cong \twdu{(\Res{\Stan{n+1,k}})} = \Res{(\twdu{\Stan{n+1,k}})} = \Res{\Stan{n+1,k}} = \Proj{n,k+1} \oplus \Proj{n,k},
\end{equation}
hence the self-duality of $\Proj{k+1}$, for all $n \ge k+1$ (noting that $k \neq n, n+1$).  Assuming that self-duality holds for $\Proj{k+i}$ and $\Proj{k+i-1}$, where $i \le \ell-2$, a similar calculation using \cref{prop:ProjRest} now proves that $\Proj{k+i+1}$ is self-dual.
\end{proof}

For $k$ non-critical with $k > k_L$, the dual of the non-split exact sequence \eqref{eq:ProjExact} is therefore
\begin{equation} \label{es:CPC}
\dses{\Cost{k}}{}{\Proj{k}}{}{\Cost{k^-}},
\end{equation}
which is likewise exact and non-split.  For $k = k_L$, $\Cost{k^-} = 0$ but $\Cost{k_L}$ is not projective, in general, because $\Proj{k_L}$ need not be self-dual.  Rather, $\Cost{k_L} \cong \twdu{\Proj{k_L}}$ is \emph{injective}.
\begin{corollary} \label{coro:Inje}
\leavevmode
\begin{enumerate}
\item When $k \in \Lambda$ is critical, the standard module $\Stan{k}$ is projective and injective. \label{it:InjCrit}
\item When $k \in \Lambda_0$ is non-critical and $k>k_L$, the projective module $\Proj{k}$ is injective. \label{it:InjProj}
\item When $k \in \Lambda_0$ is non-critical and $k=k_L$, the costandard module $\Cost{k}$ is injective. \label{it:InjCost}
\end{enumerate}
\end{corollary}
\begin{proof}
These follow from the exactness of duality, the modules being self-dual for \ref{it:InjCrit} and \ref{it:InjProj}, and the reflexivity of duality for \ref{it:InjCost} (\cref{prop:DualExact,prop:ProjDual,prop:DoubleDual}).  For $\Proj{}$ being projective is equivalent to $\Mod{M} \ra \Proj{} \ra 0$ splitting for all $\Mod{M}$, which is equivalent to $0 \ra \twdu{\Proj{}} \ra \Mod{N}$ splitting for all $\Mod{N}$, which means that $\twdu{\Proj{}}$ is injective.
\end{proof}
\noindent We remark that for $\Alg{} = \tl{}$ with $n$ even and $\beta = 0$, the costandard module $\Cost{0}$ coincides with $\Irre{2}$, by \eqref{eq:StanExact}, which is not injective.

For each $k \in \Lambda_0$, we let $\Inje{k}$ denote the \emph{injective hull} of $\Irre{k}$.  The critical standard modules are irreducible, projective and injective:  $\Irre{k} = \Stan{k} = \Proj{k} = \Inje{k}$.  Each critical $\Stan{k}$ is therefore the unique indecomposable object in its block, this block being semisimple.  These statements continue to hold for non-critical $k$ if the class $[k] \subset \Lambda$ contains a single element:  $[k]=\set{k}$.  Otherwise, the non-critical blocks are non-semisimple and we identify the injective hulls and projective covers of the modules introduced so far in the following proposition.
\begin{proposition} \label{prop:ProjInj}
If $k \in \Lambda_0$ is critical, then $\Irre{k} = \Stan{k} = \Proj{k} = \Inje{k}$. If $k \in \Lambda_0$ is non-critical, then:
\begin{enumerate}
\item The module $\Inje{k}$ is $\Proj{k}$, if $k>k_L$, and $\Cost{k_L}$, if $k=k_L$.
\item The projective cover of $\Irre{k}$, $\Stan{k}$ and $\Proj{k}$ is $\Proj{k}$.  That of $\Cost{k}$ is $\Proj{k}$, if $k=k_R$, and otherwise it is $\Proj{k^+}$.  That of $\Inje{k}$ is $\Proj{k^+}$, if $k=k_L$, and otherwise it is $\Proj{k}$. \label{it:IdentProj}
\item The injective hull of $\Irre{k}$, $\Cost{k}$ and $\Inje{k}$ is $\Inje{k}$.  That of $\Stan{k}$ is $\Inje{k}$, if $k=k_R$, and otherwise it is $\Inje{k^+}$.  That of $\Proj{k}$ is $\Inje{k^+}$, if $k=k_L$, and otherwise it is $\Inje{k}$. \label{it:IdentInj}
\end{enumerate}
\end{proposition}
\begin{proof} All injective and projective modules appearing in \cref{coro:Inje} are indecomposable. Therefore, these injectives (projectives) will automatically be the injective hull (projective cover) of their submodules (quotients). Moreover, the proof of the previous corollary shows that the injective hull of a module $\Mod{M}$ is the projective cover of its dual $\twdu{\Mod{M}}$; the statement \ref{it:IdentInj} is thus a reformulation of \ref{it:IdentProj}.
\end{proof}
\noindent We remark that the definition of the projective cover (injective hull) of $\Mod{M}$ should include the surjective (injective) homomorphism:  $\Proj{} \sra \Mod{M}$ ($\Mod{M} \ira \Inje{}$). However, we shall see in \cref{prop:SIPHomo} that these homomorphisms are unique, up to rescaling, in each case addressed by \cref{prop:ProjInj}, hence we shall usually leave them implicit.

%
%
\subsection{Interlude:  Ext-groups} \label{sec:Ext}

In order to classify indecomposable modules, we need to know when we can ``stitch'' existing modules together to build bigger ones. This knowledge is encoded in the extension group $\Ext^{1}(\Mod{M},\Mod{N})$, where $\Mod{M}$ and $\Mod{N}$ are the modules being stitched. We begin with a quick definition of these groups, usually referred to as $\Ext$-groups for short, before showing how they may be computed and what their relation is to indecomposable modules. In this section, we assume that $\Alg{}$ is an arbitrary finite-dimensional associative algebra over some field $\KK$.

We have noted in the previous sections that the operations of restricting and taking the dual of a module are exact. This means that, if $\ses{\Mod{K}}{\Mod{L}}{\Mod{M}}$ is a short exact sequence, applying these operations (functors) to each of the constituent modules gives another short exact sequence (see \cref{prop:DualExact} for an example). Not all functors are exact however. At the end of \cref{sub:restriction}, we observed that induction is right-exact, but not left-exact, by giving an explicit example. The functors under study in the present section are the Hom-functors $\Hom_{\Alg{}}(\Mod{N},\blank)$ and $\Hom_{\Alg{}}(\blank,\Mod{N})$. The first is covariant and the second contravariant.  Both are left-exact, but neither need be right-exact in general. In a sense, extension groups measure the failure of these functors to be right-exact. (A complete discussion of $\Ext$-groups, covering what is needed here, can be found, for example, in chapter III of \cite{HiltonStamm} and chapter IX of \cite{Assem-Fr}.)

Let $\Mod{N}$ be an $\Alg{}$-module. The $n$-th extension functors  $\Ext^n_{\Alg{}}(\Mod{N}, \blank)$ and $\Ext^n_{\Alg{}}(\blank,\Mod{N},)$ are the $n$-th right derived functors of $\Hom_{\Alg{}}(\Mod{N},\blank)$ and  $\Hom_{\Alg{}}(\blank,\Mod{N})$, respectively. This means that for any short exact sequence of $\Alg{}$-modules
\begin{equation}
\dses{\Mod{K}}{}{\Mod{L}}{}{\Mod{M}},
\end{equation}
there exist two long exact sequences of $\KK$-vector spaces called the \emph{Hom-Ext long exact sequences}:\footnote{As we shall see, these Hom- and Ext-groups are actually $\KK$-vector spaces because we work over a field.  If we were instead working over a general commutative ring, then the Hom- and Ext-groups would only be abelian groups, whence their names.}
\begin{subequations} \label{es:HomExt}
\begin{equation} \label{es:lcohomology}
\begin{tikzpicture}[baseline={(current bounding box.center)},xscale=2.5,every node/.style={text depth=0}]
\node (h0) at (0.25,1) [] {\(0\)};
\node (h1) at (1,1) [] {\(\Hom_{\Alg{}}(\Mod{N},\Mod{K})\)};
\node (h2) at (2,1) [] {\(\Hom_{\Alg{}}(\Mod{N},\Mod{L})\)};
\node (h3) at (3,1) [] {\(\Hom_{\Alg{}}(\Mod{N},\Mod{M})\)};
\node (e1) at (1,0) [] {\(\Ext_{\Alg{}}^1(\Mod{N},\Mod{K})\)};
\node (e2) at (2,0) [] {\(\Ext_{\Alg{}}^1(\Mod{N},\Mod{L})\)};
\node (e3) at (3,0) [] {\(\Ext_{\Alg{}}^1(\Mod{N},\Mod{M})\)};
\node (e4) at (3.75,0) [] {\(\cdots\)};
\node (E0) at (0.35,-1) [] {\(\cdots\)};
\node (E1) at (1,-1) [] {\(\Ext_{\Alg{}}^n(\Mod{N},\Mod{K})\)};
\node (E2) at (2,-1) [] {\(\Ext_{\Alg{}}^n(\Mod{N},\Mod{L})\)};
\node (E3) at (3,-1) [] {\(\Ext_{\Alg{}}^n(\Mod{N},\Mod{M})\)};
\node (E4) at (3.75,-1) [] {\(\cdots\)};
\draw [->] (h0) -- (h1);
\draw [->] (h1) -- (h2);
\draw [->] (h2) -- (h3);
\draw[->] (h3.east) arc (90:-90:0.25) -- (e1.west |- 1,0.5) arc (90:270:0.25);
\draw [->] (e1) -- (e2);
\draw [->] (e2) -- (e3);
\draw [->] (e3) -- (e4);
\draw[->] (e4.east) arc (90:-90:0.25) -- (E0.west |- 1,-0.5) arc (90:270:0.25);
\draw [->] (E0) -- (E1);
\draw [->] (E1) -- (E2);
\draw [->] (E2) -- (E3);
\draw [->] (E3) -- (E4);
\end{tikzpicture}
\end{equation}
and
\begin{equation} \label{es:rcohomology}
\begin{tikzpicture}[baseline={(current bounding box.base)},xscale=2.5,every node/.style={text depth=0}]
\node (h0) at (0.25,1) [] {\(0\)};
\node (h1) at (1,1) [] {\(\Hom_{\Alg{}}(\Mod{M},\Mod{N})\)};
\node (h2) at (2,1) [] {\(\Hom_{\Alg{}}(\Mod{L},\Mod{N})\)};
\node (h3) at (3,1) [] {\(\Hom_{\Alg{}}(\Mod{K},\Mod{N})\)};
\node (e1) at (1,0) [] {\(\Ext_{\Alg{}}^1(\Mod{M},\Mod{N})\)};
\node (e2) at (2,0) [] {\(\Ext_{\Alg{}}^1(\Mod{L},\Mod{N})\)};
\node (e3) at (3,0) [] {\(\Ext_{\Alg{}}^1(\Mod{K},\Mod{N})\)};
\node (e4) at (3.75,0) [] {\(\cdots\)};
\node (E0) at (0.35,-1) [] {\(\cdots\)};
\node (E1) at (1,-1) [] {\(\Ext_{\Alg{}}^n(\Mod{M},\Mod{N})\)};
\node (E2) at (2,-1) [] {\(\Ext_{\Alg{}}^n(\Mod{L},\Mod{N})\)};
\node (E3) at (3,-1) [] {\(\Ext_{\Alg{}}^n(\Mod{K},\Mod{N})\)};
\node (E4) at (3.75,-1) [] {\(\cdots\)};
\node at (3.875,-1) {\(\vphantom{\cdots}\) .};
\draw [->] (h0) -- (h1);
\draw [->] (h1) -- (h2);
\draw [->] (h2) -- (h3);
\draw[->] (h3.east) arc (90:-90:0.25) -- (e1.west |- 1,0.5) arc (90:270:0.25);
\draw [->] (e1) -- (e2);
\draw [->] (e2) -- (e3);
\draw [->] (e3) -- (e4);
\draw[->] (e4.east) arc (90:-90:0.25) -- (E0.west |- 1,-0.5) arc (90:270:0.25);
\draw [->] (E0) -- (E1);
\draw [->] (E1) -- (E2);
\draw [->] (E2) -- (E3);
\draw [->] (E3) -- (E4);
\end{tikzpicture}
\end{equation}
\end{subequations}
In what follows, we shall only be concerned with the first extension groups $\Ext^{1}_{\Alg{}}(\blank,\blank)$; we therefore simply write $\Ext_{\Alg{}}$ for $\Ext^{1}_{\Alg{}}$. We shall also omit the subscript $\Alg{}$, as we do for Hom-groups, when it is clear from the context.

The following results are crucial tools for computing extension groups: $\Ext(\Mod{P},\blank) = 0$ if and only if $\Mod{P}$ is projective, while $\Ext(\blank, \Mod{J}) = 0$ if and only if $\Mod{J}$ is injective. To see how these are used, let $\Mod{K}$ and $\Mod{M}$ be $\Alg{}$-modules and let
\begin{equation}\label{es:projpres}
 \dses{\Mod{K}}{}{\Mod{J}}{}{\Mod{L}}\qquad\text{and}\qquad
 \dses{\Mod{R}}{}{\Mod{P}}{}{\Mod{M}}
 \end{equation}
be an injective presentation of $\Mod{K}$ and a projective presentation of $\Mod{M}$, respectively.\footnote{A injective (projective) presentation is a short exact sequence of the form \eqref{es:projpres}, where the middle term is injective (projective).} Once specialised to these cases, the exact sequences \eqref{es:lcohomology} and \eqref{es:rcohomology} truncate, becoming
\begin{subequations}
\begin{gather}
0 \lra \Hom(\Mod{N},\Mod{K}) \lra \Hom(\Mod{N},\Mod{J}) \lra \Hom(\Mod{N},\Mod{L}) \lra \Ext(\Mod{N},\Mod{K}) \lra 0, \label{es:lcohomology2} \\
0 \lra \Hom(\Mod{M},\Mod{N}) \lra \Hom(\Mod{P},\Mod{N}) \lra \Hom(\Mod{R},\Mod{N}) \lra \Ext(\Mod{M},\Mod{N}) \lra 0. \label{es:rcohomology2}
\end{gather}
\end{subequations}
If the relevant $\Hom$ groups are known, then $\Ext(\Mod{N},\Mod{K}) $ and $\Ext(\Mod{M},\Mod{N}) $ can be identified.

The reason to introduce these extension groups here is the following: The (first) extension group $\Ext(\Mod{M},\Mod{N})$ describes, roughly speaking, the inequivalent ways to ``stitch'' the modules $\Mod{M}$ and $\Mod{N}$ together to obtain a new module $\Mod{E}$ with $\Mod{N}$ isomorphic to a submodule of $\Mod{E}$ and $\Mod{M}$ isomorphic to the quotient $\Mod{E} / \Mod{N}$. In other words, $\Ext(\Mod{M},\Mod{N})$ characterises the (inequivalent) short exact sequences $\ses{\Mod{N}}{\Mod{E}}{\Mod{M}}$. We then say that $\Mod{E}$ is an extension of $\Mod{M}$ by $\Mod{N}$. To state what it means for two extensions to be equivalent, let $\Mod{E}$ and $\Mod{E'}$ be two extensions of $\Mod{M}$ by $\Mod{N}$:
\begin{equation} \label{es:ee'}
	e \colon \dses{\Mod{N}}{f}{\Mod{E}}{g}{\Mod{M}}\qquad \text{and} \qquad e' \colon \dses{\Mod{N}}{f'}{\Mod{E'}}{g'}{\Mod{M}}.
\end{equation}
Then, $e$ and $e'$ are said to be equivalent if there exists $h \colon \Mod{E} \to \Mod{E'}$ such that $f' = hf$ and $g = g'h$. It can be shown that this is an equivalence relation and that the set of inequivalent extensions of $\Mod{M}$ by $\Mod{N}$ is in one-to-one correspondence with $\Ext(\Mod{M},\Mod{N})$. In this correspondence, the origin of the vector space $\Ext(\Mod{M},\Mod{N})$ corresponds to the split extension $\Mod{E}=\Mod{N}\oplus\Mod{M}$. Moreover, multiplying $e$ by $\alpha \in \KK^{\times}$ yields the extensions
\begin{equation}
	\alpha e \colon \dses{\Mod{N}}{\alpha f}{\Mod{E}}{g}{\Mod{M}}\qquad \text{or} \qquad \dses{\Mod{N}}{f}{\Mod{E}}{\alpha g}{\Mod{M}},
\end{equation}
which are easily seen to be equivalent.  Finally, the sum $e+e'$ is defined through an operation sometimes known as the Baer sum, completing the $\KK$-vector space structure on the set of inequivalent extensions.  As we will have no need of this sum, we omit its definition and refer the reader to any standard text on homological algebra, for example \cite[Sec.~3.4]{Weibel}.

We remark that if the two extensions $e$ and $e'$ of \eqref{es:ee'} are equivalent, then the middle modules $\Mod{E}$ and $\Mod{E}'$ are isomorphic, by the short five lemma.  However, the converse is not true:  $e$ and $\alpha e$ are not equivalent, for $\alpha \neq 1$, despite both their middle modules being $\Mod{E}$. In particular, if $\Ext(\Mod{M}, \Mod{N}) \cong \KK$, then there are precisely two extensions, up to isomorphism, one split and one non-split.  We shall use this conclusion many times in what follows.

%
%
\subsection{Their Hom- and Ext-groups}\label{sub:extension}

We now have enough information to compute the homomorphism groups between the modules $\Irre{k}$, $\Stan{k}$, $\Cost{k}$, $\Proj{k}$ and $\Inje{k}$. This is quite straightforward, but there are various cases that have to be considered for each pair of module types.  For example, $\Hom(\Irre{k},\Stan{k'}) = \delta_{k',k^-} \CC$ because the image can only be the submodule $\Radi{k'} \cong \Irre{{k'}^+}$ (or zero), unless $k'=k'_R$, in which case $\Stan{k'}=\Irre{k'}$ and $\Hom(\Irre{k},\Stan{k'_R}) = \delta_{k,k'_R} \CC$.  We can avoid much of this case analysis by agreeing to use the following conventions for the rest of this subsection:
\begin{itemize}
\item We will not consider $\Inje{k}$ explicitly because, for every $k$, $\Inje{k}$ is either $\Proj{k}$ or $\Cost{k}$.
\item The index $k$ in $\Stan{k}$, $\Cost{k}$ and $\Proj{k}$ will be assumed to be non-critical because $\Proj{k} \cong \Cost{k} \cong \Stan{k} \cong \Irre{k}$ for critical $k$.
\item The index $k$ in $\Stan{k}$ and $\Cost{k}$ will exclude $k=k_R$ because $\Cost{k_R} \cong \Stan{k_R} \cong \Irre{k_R}$.
\item The index $k$ in $\Proj{k}$ will exclude $k=k_L$ because $\Proj{k_L} \cong \Stan{k_L}$.
\item When $\Alg{} = \tl{}$, $\beta = 0$ and $n$ is even, the index $k$ in $\Stan{k}$ and $\Cost{k}$ will exclude $k=0$ because $\Cost{0} = \Stan{0} = \Irre{2}$.
\end{itemize}
In the case $\Alg{} = \tl{}$, $\beta = 0$ and $n$ even, we will also obviously exclude $\Irre{0} = 0$.  The computation of Hom-groups is then straightforward and uses the exact sequences \eqref{eq:StanExact}, \eqref{eq:ProjExact} and their duals.

\begin{proposition}\label{prop:SIPHomo}
With these conventions, the groups $\Hom(\Mod{M},\Mod{N})$ between the irreducible, standard, costandard and projective modules are summarised in the following table:
\begin{center}
\begin{tabular}{CC|CCCC}
\multicolumn{2}{C|}{\multirow{2}{*}{$\Hom(\Mod{M},\Mod{N})$}} & \multicolumn{4}{C}{\Mod{N}} \\
     &  & \ \qquad \Irre{k'} \qquad \ & \ \qquad \Stan{k'} \qquad \ & \ \qquad \Cost{k'} \qquad \ & \ \qquad \Proj{k'} \qquad \ \\
\hline
\multicolumn{1}{C}{\multirow{4}{*}{$\Mod{M}$}} & \Irre{k} & \delta_{k',k}\CC & \delta_{k',k^-}\CC & \delta_{k',k}\CC & \delta_{k',k}\CC \\
\multicolumn{1}{C}{}                           & \Stan{k} & \delta_{k',k}\CC & \brac{\delta_{k',k}+\delta_{k',k^-}}\CC & \delta_{k',k}\CC & \brac{\delta_{k',k}+\delta_{k',k^+}}\CC \\
\multicolumn{1}{C}{}                           & \Cost{k} & \delta_{k',k^+}\CC & \delta_{k',k}\CC & \brac{\delta_{k',k}+\delta_{k',k^+}}\CC & \brac{\delta_{k',k}+\delta_{k',k^+}}\CC \\
\multicolumn{1}{C}{}                           & \Proj{k} & \delta_{k',k}\CC & \brac{\delta_{k',k}+\delta_{k',k^-}}\CC & \brac{\delta_{k',k}+\delta_{k',k^-}}\CC & \brac{2\:\delta_{k',k}+\delta_{k',k^-}+\delta_{k',k^+}}\CC\ .
\end{tabular}
\end{center}
\end{proposition}
\noindent We remark that the entries of this table are related by the vector space isomorphisms $\Hom(\twdu{\Mod{M}},\twdu{\Mod{N}}) \cong \Hom(\Mod{N},\Mod{M})$, implemented by sending $\phi \colon \Mod{N} \to \Mod{M}$ to $\twdu{\phi} \colon \twdu{\Mod{M}} \to \twdu{\Mod{N}}$, where $\twdu{\phi}(f)(n) = f(\phi(n))$ (see \cref{prop:DualExact}).

The extension groups involving these modules are now straightforward to compute.  We assume the same conventions on the index $k$ as for the Hom-groups and note that this means that each $\Proj{k}$ is projective and injective, hence any Ext-group involving a $\Proj{k}$ is zero.
\begin{proposition}\label{prop:SIExte}
With these conventions, the groups $\Ext(\Mod{M},\Mod{N})$ between the irreducible, standard and costandard modules may be summarised, with three exceptions, in the following table:
\begin{center}
\begin{tabular}{CC|CCC}
\multicolumn{2}{C|}{\multirow{2}{*}{$\Ext(\Mod{M},\Mod{N})$}} & \multicolumn{3}{C}{\Mod{N}} \\
     &  & \ \qquad \Irre{k'} \qquad \ & \ \qquad \Stan{k'} \qquad \ & \ \qquad \Cost{k'} \qquad \ \\
\hline
\multicolumn{1}{C}{\multirow{3}{*}{$\Mod{M}$}} & \Irre{k} & \brac{\delta_{k',k^-}+\delta_{k',k^+}}\CC & (\delta_{k',k^-} \delta_{k,k_R} + \delta_{k',k^{--}})\CC & \delta_{k',k^+}\CC \\
\multicolumn{1}{C}{}                           & \Stan{k} & \delta_{k',k^-}\CC & \brac{\delta_{k',k^-}+\delta_{k',k^{--}}}\CC & 0 \\
\multicolumn{1}{C}{}                           & \Cost{k} & (\delta_{k',k^+} \delta_{k',k'_R} + \delta_{k',k^{++}})\CC & 0 & (\delta_{k',k^+} + \delta_{k',k^{++}})\CC \ .
\end{tabular}
\end{center}
\smallskip
The exceptions occur for $\Alg{} = \tl{2}$ and $\beta=0$, for which $\Ext(\Irre{2,2},\Irre{2,2})\simeq \CC$ instead of $0$, and for $\Alg{} = \tl{}$ and $\beta=0$, for which $\Ext(\Irre{2},\Cost{2}) \cong \Ext(\Stan{2},\Irre{2}) \cong \CC$ instead of $0$.
\end{proposition}
\begin{proof}
We compute these extension groups using the Hom-Ext long exact sequences \eqref{es:HomExt}.  To see how this works, consider $\Ext(\Stan{k},\Stan{k'})$, for $k$ non-critical, $k \neq k_R$ and $k' \neq k'_R$.  We start from the short exact sequence \eqref{eq:ProjExact}, which is a projective presentation of $\Stan{k}$, and examine the contravariant Hom-Ext long exact sequence \eqref{es:rcohomology2}:
\begin{equation}
0 \lra \Hom(\Stan{k},\Stan{k'}) \lra \Hom(\Proj{k},\Stan{k'}) \lra \Hom(\Stan{k^-},\Stan{k'}) \lra \Ext(\Stan{k},\Stan{k'}) \lra 0.
\end{equation}
Here, we have noted that the rightmost term is $\Ext(\Proj{k},\Stan{k'}) = 0$ because of projectivity.  Substituting in the homomorphism groups from \cref{prop:SIPHomo} gives
\begin{equation}
0 \lra \brac{\delta_{k',k}+\delta_{k',k^-}}\CC \lra \brac{\delta_{k',k}+\delta_{k',k^-}}\CC \lra \brac{\delta_{k',k^-}+\delta_{k',k^{--}}}\CC \lra \Ext(\Stan{k},\Stan{k'}) \lra 0.
\end{equation}
The Ext-group is therefore zero unless $k'=k^-$ or $k^{--}$.  In these cases, the exact sequence becomes
\begin{equation}
0 \lra \CC \lra \CC \lra \CC \lra \Ext(\Stan{k},\Stan{k^-}) \lra 0 \quad \text{and} \quad
0 \lra 0 \lra 0 \lra \CC \lra \Ext(\Stan{k},\Stan{k^{--}}) \lra 0,
\end{equation}
respectively, and the Euler-Poincar\'{e} principle gives $\Ext(\Stan{k},\Stan{k^-}) \simeq \Ext(\Stan{k},\Stan{k^{--}}) \simeq \mathbb C$.

The computations for $\Ext(\Irre{k},\Stan{k'})$, $\Ext(\Irre{k},\Cost{k'})$, $\Ext(\Stan{k},\Irre{k'})$, $\Ext(\Stan{k},\Cost{k'})$, $\Ext(\Cost{k},\Irre{k'})$, $\Ext(\Cost{k},\Stan{k'})$ and $\Ext(\Cost{k},\Cost{k'})$ are almost identical, utilising \eqref{eq:StanExact} or its dual \eqref{es:ICI} (which is an injective presentation of $\Cost{k}$).  We only remark that, for $\Ext(\Irre{k},\Stan{k'})$, one has to consider the case $k'=k_R^-$ separately because then $\Hom(\Irre{k},\Stan{{k'}^+}) = \Hom(\Irre{k},\Irre{{k'}^+})$.  There is a similar case to consider for $\Ext(\Cost{k},\Irre{k'})$.

The computation for $\Ext(\Irre{k},\Irre{k'})$ is slightly different in that we instead start from the short exact sequence \eqref{eq:StanExact}, which is not a projective or injective presentation, to derive the Hom-Ext long exact sequence
\begin{equation}
\begin{tikzpicture}[baseline={(current bounding box.center)},xscale=2.5,yscale=2]
\node (h0) at (0.35,0.5) [] {\(0\)};
\node (h1) at (1,0.5) [] {\(\Hom(\Irre{k},\Irre{{k'}})\)};
\node (h2) at (2,0.5) [] {\(\Hom(\Stan{k},\Irre{{k'}})\)};
\node (h3) at (3,0.5) [] {\(\Hom(\Irre{k^+},\Irre{{k'}})\)};
\node (e1) at (1,0) [] {\(\Ext(\Irre{k},\Irre{{k'}})\)};
\node (e2) at (2,0) [] {\(\Ext(\Stan{k},\Irre{{k'}})\)};
\node (e3) at (3,0) [] {\(\Ext(\Irre{k^+},\Irre{{k'}})\)};
\node (e4) at (3.75,0) [] {\(\cdots \vphantom{\Hom(\Stan{k},\Irre{{k'}})}\)\ .};
\draw [->] (h0) -- (h1);
\draw [->] (h1) -- (h2);
\draw [->] (h2) -- (h3);
\draw[->] (h3.east) arc (90:-90:0.125) -- (e1.west |- 1,0.25) arc (90:270:0.125);
\draw [->] (e1) -- (e2);
\draw [->] (e2) -- (e3);
\draw [->] (e3) -- (e4);
\end{tikzpicture}
\end{equation}
Substituting for $\Hom(\Irre{k^+},\Irre{{k'}})$ and $\Ext(\Stan{k},\Irre{{k'}})$, we learn that $\Ext(\Irre{k},\Irre{k'}) = 0$, unless $k' = k^{\pm}$, and that the result is $\CC$ in either of the interesting cases.

Finally, the exceptions noted in the table are all related to the degenerate structure of the projective $\Proj{2}$ when $\Alg{} = \tl{}$, $\beta = 0$ and $n$ is even.  Specifically, the exact sequence \eqref{eq:ProjExact} becomes
\begin{equation}\label{eq:unbearable}
\dses{\Irre{2}}{}{\Proj{2}}{}{\Stan{2}},
\end{equation}
explaining the non-triviality of $\Ext(\Stan{2},\Irre{2})$ and $\Ext(\Irre{2},\Cost{2})$.  When $n=2$, $\Stan{2} \cong \Irre{2}$, explaining the non-triviality of $\Ext(\Irre{2},\Irre{2})$.  The dimensionality of these Ext-groups is easily verified using projective presentations and Hom-Ext long exact sequences, as above.
\end{proof}

It is easy to show that the entries of the table in \cref{prop:SIExte} are related by the vector space isomorphisms $\Ext(\twdu{\Mod{M}},\twdu{\Mod{N}}) \cong \Ext(\Mod{N},\Mod{M})$.  This follows by noting that a projective presentation $\ses{\Mod{R}}{\Mod{P}}{\Mod{M}}$ of $\Mod{M}$ gives an injective presentation $\ses{\twdu{\Mod{M}}}{\twdu{\Mod{P}}}{\twdu{\Mod{R}}}$ of $\twdu{\Mod{M}}$, hence the exact sequence \eqref{es:lcohomology2} becomes
\begin{equation}
0 \lra \Hom(\twdu{\Mod{N}},\twdu{\Mod{M}}) \lra \Hom(\twdu{\Mod{N}},\twdu{\Mod{P}}) \lra \Hom(\twdu{\Mod{N}},\twdu{\Mod{R}}) \lra \Ext(\twdu{\Mod{N}},\twdu{\Mod{M}}) \lra 0,
\end{equation}
upon replacing $\Mod{N}$ by $\twdu{\Mod{N}}$.  As $\Hom(\twdu{\Mod{N}},\twdu{\Mod{K}}) \cong \Hom(\Mod{K},\Mod{N})$, comparing with \eqref{es:rcohomology2} completes the proof.

An example of an application of these extension groups is the following result.
\begin{corollary}\label{coro:StanUnique}
If a module $\Mod{M}$ is indecomposable with exact sequence $\ses{\Irre{k^+}}{\Mod{M}}{\Irre{k}}$, then $\Mod{M}\simeq \Stan{k}$.  Similarly, if $\Mod{N}$ is indecomposable with exact sequence $\ses{\Irre{k}}{\Mod{N}}{\Irre{k^+}}$, then $\Mod{N}\simeq \Cost{k}$.
\end{corollary}
\begin{proof}
Since $\Ext(\Irre{k},\Irre{k^+}) \cong \CC$, there is a single isomorphism class of non-trivial extensions of $\Irre{k}$ by $\Irre{k^+}$.  Since $\Stan{k}$ is indecomposable and is therefore one such extension, $\Mod{M}$ must be isomorphic to $\Stan{k}$. The second statement now follows by duality.
\end{proof}
\noindent It follows from \cref{prop:SIExte} and \cref{coro:StanUnique} that we have classified all $\Alg{}$-modules with two composition factors, up to isomorphism.  The complete list consists of the direct sums of two irreducibles, the (reducible) standard and costandard modules, and the projective $\tl{2}$-module $\Proj{2,2}$ at $\beta = 0$.

We end this subsection by proving two useful lemmas. The first limits the number of non-trivial extension groups.
\begin{lemma}\label{lem:trivialExtensions}
\leavevmode
\begin{enumerate}
\item Let $\Mod{M}$ and $\Mod{N}$ be $\Alg{}$-modules whose composition factors have indices $k^i$ and $k^j$, for some reference index $k$, where $i$ and $j$ run over (multi)sets $I$ and $J$, respectively, of integers. If $\abs{i-j}>1$ for all $i\in I$ and $j\in J$, then $\Ext(\Mod{M},\Mod{N})\simeq \Ext(\Mod{N},\Mod{M})\simeq 0$. \label{it:Ext=0}
\item Let $\Mod{M}$ be an  $\Alg{}$-module such that $\Ext(\Mod{M},\Mod{I})\simeq 0$ ($\Ext(\Mod{I},\Mod{M})\simeq 0$) for all semisimple modules $\Mod{I}$. Then $\Ext(\Mod{M}, \Mod{N})\simeq 0$ ($\Ext(\Mod{N},\Mod{M})\simeq 0$) for all modules $\Mod{N}$. \label{it:ExtByIrr}
\end{enumerate}
\end{lemma}
\begin{proof}
The proof of \ref{it:Ext=0} is by double induction, first on the length of $\Mod{M}$, that is the number of its composition factors (including multiplicities), then on the length of $\Mod{N}$.  First, suppose that $\Mod{N} \cong \Irre{k^j}$ is irreducible.  If the length of $\Mod{M}$ is $1$, then $\Mod{M}$ is also irreducible and \cref{prop:SIExte} gives the result. If its length is greater than $1$, let $\Irre{}$ be an irreducible submodule of $\Mod{M}$.  Then, $\ses{\Irre{}}{\Mod{M}}{\Mod{M}/\Mod{I}}$ is exact and the covariant Hom-Ext long exact sequence \eqref{es:lcohomology} says that so is $\Ext(\Irre{k^j},\Irre{})\rightarrow \Ext(\Irre{k^j},\Mod{M})\rightarrow\Ext(\Irre{k^j},\Mod{M}/\Irre{})$. The two extreme extension groups are $0$ by the induction hypothesis, hence $\Ext(\Irre{k^j},\Mod{M})=0$ as well. The contravariant Hom-Ext long exact sequence \eqref{es:rcohomology} similarly yields $\Ext(\Mod{M},\Irre{k^i}) = 0$.  This now provides the base case for a similar induction on the length of $\Mod{N}$ which completes the proof of \ref{it:Ext=0}.  A similar induction argument proves \ref{it:ExtByIrr}.
\end{proof}

A non-trivial extension group $\Ext(\Mod{N},\Mod{L})$ implies the existence of a module $\Mod{M}$  such that the exact sequence
\begin{equation} \label{es:LMN}
\dses{\Mod{L}}{}{\Mod{M}}{}{\Mod{N}}
\end{equation}
does not split. This, however, does not prove that $\Mod{M}$ is indecomposable. Indeed, \cref{prop:extensionRevealed} will give examples of non-trivial extensions that are decomposable.  The second lemma gives an easy criterion to prove the indecomposability of a given extension.
\begin{lemma}\label{lem:indecExtension}
Let $\Mod{L}$, $\Mod{M}$ and $\Mod{N}$ be the $\Alg{}$-modules appearing in the short exact sequence \eqref{es:LMN}. Suppose furthermore that $\Mod{L}$ and $\Mod{N}$ are indecomposable and that $\Hom(\Mod{L},\Mod{N})=0$. Then, $\Mod{M}$ is decomposable if and only if \eqref{es:LMN} splits.
\end{lemma}
\begin{proof}
Since the definition of the short exact sequence splitting is that $\Mod{M}\simeq\Mod{L}\oplus\Mod{N}$, we only need show that decomposability implies splitting under these hypotheses. Suppose then that $\Mod{M}$ is decomposable. This implies that there exists a non-trivial projection $q \colon \Mod{M} \to \Mod{M}$, meaning that the morphism $q$ satisfies $q^2=q$ but is neither zero nor the identity.

Consider now the diagram
\begin{equation}\label{cd:IndecExt}
\begin{tikzpicture}[baseline={(current bounding box.center)},scale=1/3]
\node (t1) at (5,5) [] {$0$};
\node (t2) at (10,5) [] {$\Mod{L}$};
\node (t3) at (15,5) [] {$\Mod{M}$};
\node (t4) at (20,5) [] {$\Mod{N}$};
\node (t5) at (25,5) [] {$0$\phantom{,}};
\node (b1) at (5,0) [] {$0$};
\node (b2) at (10,0) [] {$\Mod{L}$};
\node (b3) at (15,0) [] {$\Mod{M}$};
\node (b4) at (20,0) [] {$\Mod{N}$};
\node (b5) at (25,0) [] {$0$,};
\draw [->] (t1) -- (t2);
\draw [->] (t2) to node [above] {$f$} (t3);
\draw [->] (t3) to node [above] {$g$}  (t4);
\draw [->] (t4) -- (t5);
\draw [->] (b1) -- (b2);
\draw [->] (b2) to node [above] {$f$}  (b3);
\draw [->] (b3) to node [above] {$g$}  (b4);
\draw [->] (b4) -- (b5);
\draw [dashed,->] (t2) to node [left] {$p$} (b2);
\draw [->] (t3) to node [left] {$q$} (b3);
\draw [dashed,->] (t4) to node [right] {$r$} (b4);
\end{tikzpicture}
\end{equation}
in which both rows are exact. Because $\Hom(\Mod{L},\Mod{N})\simeq 0$, the composition $gqf$ must vanish, implying that there exist unique morphisms $p$ and $r$ that make \eqref{cd:IndecExt} commute.  Now, $f$ is injective and $fpp=qfp=qqf=qf=fp$, so it follows that $p^2=p$. Similarly, $g$ being surjective implies that $r^2=r$. The indecomposability of $\Mod{L}$ and $\Mod{N}$ therefore requires that $p$ and $r$ must be either the zero morphism or the identity morphism. There are thus four subcases to study.

If both $p$ and $r$ are the identity, then so is $q$ by the short five lemma. If, contrarily, both $p$ and $r$ are zero, then we have $qf=0$, hence $\im f \subseteq \ker q$, and $gq=0$, hence $\im q \subseteq \ker g$.  Combining these with exactness now gives $\im q \subseteq \ker q$, which yields $q=q^2=0$, a contradiction.

The remaining two cases are more interesting. If $p$ is the identity and $r$ is zero, then the snake lemma gives $\ker q\simeq \Mod{N}$. On the other hand, the commuting left square of \eqref{cd:IndecExt} implies that $q$ acts as the identity on $\im f\simeq \Mod{L}$. The two eigenspaces of $q$ are therefore isomorphic to $\Mod{L}$ and $\Mod{N}$, hence $\Mod{M}\simeq \Mod{L}\oplus \Mod{N}$ and the sequence \eqref{es:LMN} splits. A similar argument shows that \eqref{es:LMN} also splits when $p$ is zero and $r$ is the identity.
\end{proof}
We remark that the conclusion of \cref{lem:indecExtension} is also true under the hypotheses that $\Mod{L}$ and $\Mod{N}$ are indecomposable and that $\Mod{M}$ has no subquotient\footnote{We recall that a subquotient of a module $\Mod{M}$ is a submodule of a quotient of $\Mod{M}$ or, equivalently, a quotient of a submodule of $\Mod{M}$.} isomorphic to a direct sum of two isomorphic irreducibles.  For then the submodules of $\Mod{M}$ obey the distributive laws \cite{BenRep91}
\begin{equation}
(\Mod{A} + \Mod{B}) \cap \Mod{C} \cong (\Mod{A} \cap \Mod{C}) + (\Mod{B} \cap \Mod{C}), \qquad
(\Mod{A} \cap \Mod{B}) + \Mod{C} \cong (\Mod{A} + \Mod{C}) \cap (\Mod{B} + \Mod{C}),
\end{equation}
from which the lemma follows rather trivially.  However, the hypotheses of \cref{lem:indecExtension} given above have the advantage that they only require knowledge of $\Mod{L}$ and $\Mod{N}$, and not of the subquotient structure of the extension $\Mod{M}$ itself.

We will generally use \cref{lem:indecExtension} when the two indecomposable modules $\Mod{N}$ and $\Mod{L}$ have no composition factors in common, so the Hom-groups between them necessarily vanish. Then the extension \eqref{es:LMN} will be non-trivial if and only if $\Mod{M}$ is indecomposable.

To illustrate a typical usage of this result, consider the non-trivial extension group $\Ext(\Stan{k^+},\Stan{k^-})\simeq \mathbb C$ of \cref{prop:SIExte} (for $k$ non-critical).  As $\Stan{k^+}$ and $\Stan{k^-}$ are indecomposable modules with no composition factors in common, $\Hom(\Stan{k^-}, \Stan{k^+}) = 0$ and thus any non-trivial extension will be indecomposable by \cref{lem:indecExtension}.  This example demonstrates the existence of indecomposable $\Alg{}$-modules with (three or) four composition factors $\Irre{k^-}$, $\Irre{k}$, $\Irre{k^+}$ and $\Irre{k^{++}}$. We remark that the only indecomposable modules with four composition factors that have been encountered thus far are projective and that these projectives always have a composition factor of multiplicity two. The indecomposable extension described above cannot be one of these projectives, hence must be a new indecomposable $\Alg{}$-module. This observation will be the starting point of \cref{sec:newFamilies}.

%
%
\subsection{Their Loewy diagrams}\label{sub:Loewy}

It is sometimes convenient to visualise the structure of non-semisimple modules diagrammatically, particularly when computing Hom-groups.  A convention popular in the mathematical physics community is to represent the structure as a graph in which the vertices are the composition factors of the module $\Mod{M}$ and the arrows represent the ``action of the algebra''.  More precisely, an arrow is drawn from the factor $\Mod{I}$ to the factor $\Mod{I}'$ if $\Mod{M}$ has a subquotient isomorphic to a non-trivial extension of $\Mod{I}$ by $\Mod{I}'$.  In principle, one can also decorate the arrow with an extra label if $\dim \Ext (\Mod{I},\Mod{I}') > 1$, but \cref{prop:SIExte} ensures that this never happens for $\Alg{} = \tl{}$ or $\dtl{}$.

The utility of this arrow notation is evidently limited.  For example, suppose that $\Mod{M}$ has a composition factor $\Mod{I}$ appearing as a submodule with multiplicity $2$:  $\Mod{I} \overset{\iota_1}{\ira} \Mod{M}$ and $\Mod{I} \overset{\iota_2}{\ira} \Mod{M}$.  Then, it has an infinite number of submodules isomorphic to $\Mod{I}$, corresponding to linear combinations of $\iota_1$ and $\iota_2$ (modulo rescalings).  An arrow indicating a submodule $\Mod{L} \subset \Mod{M}$ that is isomorphic to a non-trivial extension in $\Ext(\Mod{I}',\Mod{I})$ will then also require a label to identify which linear combination of $\iota_1$ and $\iota_2$ describes the submodule $\Mod{I} \subset \Mod{L}$.  The labelling of the arrows can therefore be unpleasantly complicated in general.  However, this issue also turns out to not be a problem for the indecomposable modules of $\tl{}$ or $\dtl{}$, so we will always decorate our diagrams with arrows in order to maximise the information conveyed.

We will refer to these structure graphs as \emph{Loewy diagrams}.  For the reasons already mentioned, the Loewy diagrams defined by mathematicians tend not to have arrows; instead, the composition factors are arranged in horizontal layers that have structural meaning.  To make this precise, one introduces the \emph{radical} $\rad\Mod{M}$ and the \emph{socle} $\soc\Mod{M}$ of a module $\Mod{M}$ as the intersection of its maximal proper submodules and the sum of its simple submodules, respectively.\footnote{For the algebras $\tl{}$ and $\dtl{}$, this notion of radical generalises that which was introduced in \cref{sub:basics} for the standard modules $\Stan{k}$.}  The \emph{head} $\head \Mod{M}$ of $\Mod{M}$ is then the quotient by the radical:  $\head\Mod{M} = \Mod{M} / \rad\Mod{M}$.  Just as $\soc\Mod{M}$ is the (unique) maximal semisimple submodule of $\Mod{M}$, $\head\Mod{M}$ is the (unique) maximal semisimple quotient of $\Mod{M}$.

Radicals and socles lead to important examples of filtrations.  Given a module $\Mod{M}$, define its \emph{radical series} and \emph{socle series} to be the following strictly descending and strictly ascending chains of submodules:
\begin{equation}
\begin{gathered}
\Mod{M}=\rad^0\Mod{M} \supset \rad^1\Mod{M} \supset \rad^2\Mod{M} \supset \dots \supset \rad^{n-1}\Mod{M} \supset \rad^n\Mod{M}=0, \\
0=\soc^0\Mod{M} \subset \soc^1\Mod{M} \subset \soc^2\Mod{M} \subset \dots \subset \soc^{n-1}\Mod{M} \subset \soc^n\Mod{M}=\Mod{M}.
\end{gathered}
\end{equation}
Here, $\rad^j\Mod{M}$ and $\soc^j\Mod{M}$ are defined recursively, for $j \ge 1$, to be $\rad(\rad^{j-1}\Mod{M})$ and the unique submodule satisfying $\soc^j\Mod{M} / \soc^{j-1}\Mod{M} = \soc(\Mod{M} / \soc^{j-1}\Mod{M})$, respectively.  Note that both chains contain the same number $n$ of non-zero submodules (this number is called the \emph{Loewy length} of the module $\Mod{M}$) and that the successive quotients $\rad^j\Mod{M} / \rad^{j+1}\Mod{M}$ and $\soc^j\Mod{M} / \soc^{j-1}\Mod{M}$ are maximal semisimple.

For $\Alg{} = \tl{}$ or $\dtl{}$, the radical and socle series of any given $\Alg{}$-module coincide: $\rad^j\Mod{M} = \soc^{n-j}\Mod{M}$.  We may therefore draw \emph{the} Loewy diagram so that its composition factors are partitioned (uniquely) into horizontal layers according to the following convention:  the $j$-th layer (counting from bottom to top) indicates the composition factors that appear in the maximal semisimple quotient $\soc^j\Mod{M} / \soc^{j-1}\Mod{M} = \rad^{n-j}\Mod{M} / \rad^{n-j+1}\Mod{M}$.  In addition to arranging our composition factors thusly, we shall also use arrows as a means to indicate further refinements to the substructure. It moreover proves convenient to arrange the composition factors so that their indices increase from left to right, as in the non-critical orbits of \cref{sub:basics}.

As an example of the ``annotated'' Loewy diagrams that we shall use, we present the diagrams for the standard modules:
\begin{equation} \label{ld:Stan}
\begin{matrix}\Stan{k} \\ \text{(\(k\) critical)}\end{matrix}:\ \ \ \Irre k
\qquad\qquad \begin{matrix}\Stan{k}\\ \text{(\(k\) non-critical)}\end{matrix} :\
\begin{tikzpicture}[baseline={(current bounding box.center)},scale=1/2.5]
\node (l1) at (0,2) [] {\(\Irre{k}\)};
\node (l2) at (2,0) [] {\(\Irre{k^+}\)};
\draw [->] (l1) -- (l2);
\end{tikzpicture}
.
\end{equation}
The left diagram obviously reflects the fact that critical standard modules are irreducible; the arrow in the right diagram means that the action of $\Alg{}$ can map any element of $\Stan{k}$ associated with the composition factor $\Irre{k}$ to an element of the factor $\Irre{k^+}$, but not vice versa.  The factor $\Irre{k^+}$ is then a submodule of $\Stan{k}$ and the factor $\Irre{k}$ represents the quotient $\Stan{k} / \Irre{k^+}$, as in the exact sequence \eqref{eq:StanExact}.  Note that we are employing the convention that modules with $k \notin \Lambda$ are zero: when $k=k_R$, the diagram on the right degenerates to that of the left because $\Irre{k_R^+} = 0$. In the language introduced above, the standard modules $\Stan{k}$, with $k$ non-critical and $k \neq k_R$, have $\soc\Stan{k} = \rad\Stan{k} = \Radi{k} \cong \Irre{k^+}$ and $\head\Stan{k} \cong \Irre{k}$.

The short exact sequence \eqref{es:ICI} then gives the Loewy diagrams of the costandard modules:
\begin{equation} \label{ld:Cost}
\begin{matrix}\Cost{k} \\ \text{(\(k\) critical)}\end{matrix}:\ \ \ \Irre k
\qquad\qquad \begin{matrix}\Cost{k}\\ \text{(\(k\) non-critical)}\end{matrix} :\
\begin{tikzpicture}[baseline={(current bounding box.center)},scale=1/2.5]
\node (l1) at (0,0) [] {\(\Irre{k}\)};
\node (l2) at (2,2) [] {\(\Irre{k^+}\)};
\draw [->] (l2) -- (l1);
\end{tikzpicture}
.
\end{equation}
This illustrates the general rule that the Loewy diagram for $\twdu{\Mod{M}}$ is obtained from that of $\Mod{M}$ by reversing all arrows and flipping the diagram upside-down.  In principle, one should also replace each composition factor by its dual as well, but for $\Alg{} = \tl{}$ and $\dtl{}$, every irreducible is self-dual (\cref{prop:IrreDual}). Moreover, because duality is an exact contravariant functor (\cref{prop:DualExact}), it exchanges a module's radical and socle series.  In particular, it swaps the socle with the head:  $\soc(\twdu{\Mod{M}}) \cong \head\Mod{M}$.

The Loewy diagrams for the projective modules are also easily constructed. The following cases are easy:
\begin{subequations} \label{ld:ProjEasy}
\begin{equation}
\begin{matrix}\Proj{k}\\ \text{(\(k\) critical)}\end{matrix}:\ \ \ \Irre k
\qquad\qquad \begin{matrix}\Proj{k}\\ \text{(\(k\) non-critical with \(k=k_L\))}\end{matrix} :\
\begin{tikzpicture}[baseline={(current bounding box.center)},scale=1/2.5]
\node (l1) at (0,2) [] {\(\Irre{k}\)};
\node (l2) at (2,0) [] {\(\Irre{k^+}\)};
\draw [->] (l1) -- (l2);
\end{tikzpicture}
.
\end{equation}
The diagram for $k$ non-critical and larger than $k_L$ follows from the exact sequence \eqref{eq:ProjExact} and its dual \eqref{es:CPC}, recalling that these projectives are self-dual (\cref{prop:ProjDual}).  The submodules and quotients from these exact sequences lead to the following Loewy diagram:
\begin{equation} \label{ld:ProjDiamond}
\begin{matrix}\Proj{k}\\ \text{(\(k\) non-critical with \(k>k_L\))}\end{matrix}\ :\
\begin{tikzpicture}[baseline={(current bounding box.center)},scale=1/2.5]
\node (l1) at (0,2) [] {\(\Irre{k^-}\)};
\node (m1) at (2,4) [] {\(\Irre{k}\)};
\node (m2) at (2,0) [] {\(\Irre{k}\)};
\node (r1) at (4,2) [] {\(\Irre{k^+}\)};
\draw [->] (m1) -- (l1);
\draw [->] (l1) -- (m2);
\draw [->] (m1) -- (r1);
\draw [->] (r1) -- (m2);
\end{tikzpicture}
.
\end{equation}
There is no arrow between the two $\Irre{k}$ factors because such a self-extension would be a direct sum: $\Ext(\Irre{k},\Irre{k}) = 0$ (the single exception for $\tl{2}$, $\beta=0$ and $k=2$ is not relevant here as $\Proj{2,2}$ has only two composition factors) and any arrows between the $\Irre{k^{\pm}}$ are likewise ruled out by extension groups.  More fundamentally, such an arrow would contradict the fact that $\Proj{k}$ has submodules isomorphic to $\Stan{k^-}$ and $\Cost{k}$. For example, an arrow from $\Irre{k^-}$ to $\Irre{k^+}$ in \eqref{ld:ProjDiamond} would mean that the submodule generated by the $\Irre{k^-}$ factor is not isomorphic to $\Stan{k^-}$.

We mention the degenerate case $k=k_R$, for which $\Irre{k_R^+} = 0$.  This factor and its incident arrows are therefore removed from the Loewy diagram \eqref{ld:ProjDiamond} to obtain that of $\Proj{k_R}$.  A different degeneration occurs when $\Alg{} = \tl{}$, with $n$ even and $\beta = 0$, as then $\Irre{k^-} = 0$ for $k=2$.  Moreover, both degenerations occur simultaneously if, in addition, $n=2$.  For completeness, we draw the Loewy diagrams for each of these cases:
\begin{equation} \label{ld:ProjDegenerate}
\begin{matrix}\Proj{k}\\ \text{(\(k\) non-critical with} \\ \text{\(k=k_R\) and \(n \neq 2\))}\end{matrix}\ :\
\begin{tikzpicture}[baseline={(current bounding box.center)},scale=1/2.5]
\node (l1) at (0,2) [] {\(\Irre{k^-}\)};
\node (m1) at (2,4) [] {\(\Irre{k}\)};
\node (m2) at (2,0) [] {\(\Irre{k}\)};
\draw [->] (m1) -- (l1);
\draw [->] (l1) -- (m2);
\end{tikzpicture}
\qquad\quad\begin{matrix}\Proj{2} \\ \text{(\(n \neq 2\) even,}\\ \text{\(\Alg{} = \tl{}\), \(\beta=0\))}\end{matrix}\ :\
\begin{tikzpicture}[baseline={(current bounding box.center)},scale=1/2.5]
\node (m1) at (2,4) [] {\(\Irre{2}\)};
\node (m2) at (2,0) [] {\(\Irre{2}\)};
\node (r1) at (4,2) [] {\(\Irre{4}\)};
\draw [->] (m1) -- (r1);
\draw [->] (r1) -- (m2);
\end{tikzpicture}
\qquad\quad\begin{matrix}\Proj{2,2}\\ \text{(\(\Alg{} = \tl{2}\), \(\beta=0\))}\end{matrix}\ :\
\begin{tikzpicture}[baseline={(current bounding box.center)},scale=1/2]
\node (m1) at (2,3) [] {\(\Irre{2,2}\)};
\node (m2) at (2,1) [] {\(\Irre{2,2}\)};
\draw [->] (m1) -- (m2);
\end{tikzpicture}
.
\end{equation}
\end{subequations}

It may be useful to translate the Loewy diagrams \eqref{ld:ProjDiamond} and \eqref{ld:ProjDegenerate} of the $\Proj{k}$, with $k$ non-critical and $k \neq k_L$, into the language of radicals and socles.  The Loewy length of these modules is $3$ (except for $\Proj{2,2}$, for $\Alg{} = \tl{}$ and $\beta = 0$).  The (unique) Loewy series for $\Proj{k}$ takes the form
\begin{equation}
0 \subset \Irre{k} \subset \TheV{k^-} \subset \Proj{k},
\end{equation}
where $\Irre{k}$ is (isomorphic to) the socle of $\Proj{k}$ and the, as yet undescribed, module $\TheV{k^-}$ is its radical.\footnote{The module $\TheV{k^-}$ will be studied in \cref{sec:ProjInjPres} where the reason for the chosen notation will become apparent.} The semisimple quotient $\TheV{k} / \Irre{k} = \rad\Proj{k} / \rad^2\Proj{k} = \soc^2\Proj{k} / \soc\Proj{k}$ is (isomorphic to) $\Irre{k^-} \oplus \Irre{k^+}$.  Finally, the head is also (isomorphic to) $\Irre{k}$, consistent with the self-duality of $\Proj{k}$.  We summarise this as follows:
\begin{equation}
\begin{tikzpicture}[baseline={(current bounding box.center)},scale=1/2.0]
\node (l1) at (0.2,2) [] {\(\Irre{k^-}\)};
\node (m1) at (2,4) [] {\(\Irre{k}\)};
\node (m2) at (2,0.2) [] {\(\Irre{k}\)};
\node (r1) at (4.1,2) [] {\(\Irre{k^+}\)};
\draw [->] (m1) -- (l1);
\draw [->] (l1) -- (m2);
\draw [->] (m1) -- (r1);
\draw [->] (r1) -- (m2);
\draw[bleu,densely dotted,thick,rounded corners] (2,4) circle (0.5); 
\draw[rouge,densely dashed,thick,rounded corners] (1,0.25) -- (0,1.25) -- (-0.75,2) -- (0,2.75) -- (2,0.75) -- (4,2.75) -- (4.75,2) -- (2,-0.75) -- (1,0.25); 
\draw[black,thick,rounded corners] (2,0.2) circle (0.5); 
\node (theR) at (-1.5,1.) [] {\(\rad \Proj k\)}; 
\node (flecheR) at (-0.5,2.) [] {};
\draw [->] (theR) to [bend right=-25] (flecheR);
\node (theH) at (-0.3,3.5) [] {\(\head \Proj k\)}; 
\node (flecheH) at (1.7,4.2) [] {};
\draw [->] (theH) to [bend right=-15] (flecheH);
\node (theS) at (-0.3,-0.5) [] {\(\soc \Proj k\)}; 
\node (flecheS) at (1.7,0.2) [] {};
\draw [->] (theS) to [bend right=15] (flecheS);
\end{tikzpicture}
\ .
\end{equation}

We close this section by describing a suggestive use for Loewy diagrams: they help in determining the structure of non-trivial extensions. In \cref{sub:extension}, we observed that \cref{prop:SIExte} predicted the existence of indecomposable modules with three and four composition factors. For example, the extension groups $\Ext(\Stan{k},\Irre{k^-})$, $\Ext(\Irre{k^+},\Stan{k^-})$ and $\Ext(\Stan{k^+},\Stan{k^-})$ are each isomorphic to $\CC$. We use dashed arrows in these drawings to indicate the extension itself; solid arrows describing the submodule and quotient:
\begin{equation} \label{ld:ProjSubquotients}
\begin{tikzpicture}[baseline={(current bounding box.center)},scale=0.9]
\node (k-) at (0,0) [] {\(\Irre{k^-}\)};
\node (k) at (1,1) [] {\(\Irre{k}\)};
\node (k+) at (2,0) [] {\(\Irre{k^+}\)};
\draw [->,dashed] (k) -- (k-);
\draw [->] (k) -- (k+);
\node at (1,-1) [] {\(\Ext(\Stan{k},\Irre{k^-})\)};
\end{tikzpicture}
\qquad\qquad
\begin{tikzpicture}[baseline={(current bounding box.center)},scale=0.9]
\node (k--) at (0,1) [] {\(\Irre{k^-}\)};
\node (k-) at (1,0) [] {\(\Irre{k}\)};
\node (k) at (2,1) [] {\(\Irre{k^+}\)};
\draw [->,dashed] (k) -- (k-);
\draw [->] (k--) -- (k-);
\node at (1,-1) [] {\(\Ext(\Irre{k^+},\Stan{k^-})\)};
\end{tikzpicture}
\qquad\qquad
\begin{tikzpicture}[baseline={(current bounding box.center)},scale=0.9]
\node (k) at (0,1) [] {\(\Irre{k^-}\)};
\node (k+) at (1,0) [] {\(\Irre{k}\)};
\node (k++) at (2,1) [] {\(\Irre{k^+}\)};
\node (k+++) at (3,0) [] {\(\Irre{k^{++}}\)};
\draw [->,dashed] (k++) -- (k+);
\draw [->] (k) -- (k+);
\draw [->] (k++) -- (k+++);
\node at (1.5,-1) [] {\(\Ext(\Stan{k^+},\Stan{k^-})\)};
\end{tikzpicture}
.
\end{equation}
We shall prove in \cref{sub:projectiveBT} that the Loewy diagrams of these extensions are obtained from these diagrams by replacing the dashed arrows with solid ones.  Moreover, representatives for the non-trivial isomorphism classes of the first two $\Ext$-groups will be constructed, in \cref{sec:ProjInjPres}, as subquotients of the principal indecomposables, as their Loewy diagrams suggest.  In contrast, the third Loewy diagram is new, indicating that there are more indecomposables to be discovered beyond the subquotients of the projectives and injectives.

%
%

\section{A complete set of indecomposable modules}\label{sec:newFamilies}

This section constructs a complete set of classes of indecomposable modules for $\tl{n}$ and $\dtl{n}$, up to isomorphism, using relatively elementary theory. \cref{sec:ARquiver} obtains the same set using a more advanced tool, namely Auslander-Reiten theory. There, the main results of this theory will be reviewed (without proof) and then applied to $\tl{n}$ and $\dtl{n}$. It will turn out that both of these algebras are representation-finite, meaning that their inequivalent indecomposable (finite-dimensional) modules are finite in number.

Here, we pick up from the observation that closed \cref{sub:extension}. The table of extension groups of \cref{prop:SIExte} proves the existence of non-trivial extensions with three and four composition factors that, by \cref{lem:indecExtension}, are indecomposable. Our first step will exploit this observation further. It will define, recursively, a family of indecomposable modules, the existence of the next member being granted by a non-trivial extension group involving the present one. The second step will reveal the structure of these modules which we will summarise by computing their socles, radicals and heads.  This information suffices to draw their Loewy diagrams. The third step uses this information to construct projective and injective presentations of the new indecomposables.  These presentations allow us to compute their extension groups with irreducible modules in the fourth step. \cref{lem:trivialExtensions} tells us that these groups will detect if \emph{any} further extensions are possible. The fifth and last step will prove that the modules introduced in the first step form, together with the projective and critical standard modules, a complete set of inequivalent indecomposable modules.

%
%

\subsection{The modules $\TheB{n,k}l$ and $\TheT{n,k}l$} \label{sub:definitionBT}

Let $k\in\Lambda_0$ be non-critical and define $\TheB k0\equiv \Irre k$. We state the obvious fact that this module is irreducible, hence indecomposable, with a single composition factor $\Irre k$. From $\TheB k0$, we construct recursively a family of indecomposable modules $\TheB k{2j}$, $j=0, 1, \dots$, as follows (the range of $j$ will be clarified below). The $j$-th step constructs the module $\TheB k{2j}$ which has $2j+1$ composition factors $\Irre{k}$, $\Irre{k^+}$, $\Irre{k^{++}}$, \dots, $\Irre{k^{2j}}$. The following step of the recursion then shows that the extension group $\Ext(\Stan{k^{2j+1}},\TheB k{2j})$ is isomorphic to $\mathbb C$ and we name the non-trivial extension $\TheB k{2(j+1)}$.

Note that the first step of this recursive definition follows easily from \cref{prop:SIExte}. Indeed, it shows that the extension group $\Ext(\Stan{k^+},\TheB{k}{0}) = \Ext(\Stan{k^+},\Irre k)$ is $\mathbb C$, hence that the short exact sequence $\ses{\Irre k}{\TheB k{2}}{\Stan{k^+}}$ completely characterises the non-trivial extension $\TheB k2$, up to isomorphism.  Moreover, since the composition factors of $\Irre k$ and $\Stan{k^+}$ are distinct, \cref{lem:indecExtension} proves that $\TheB k{2}$ is indecomposable.  Finally, its composition factors are clearly $\Irre k$, $\Irre{k^+}$ and $\Irre{k^{++}}$, as required.

Suppose then that the $j$-th recursive step has been completed, so that there exists an indecomposable module $\TheB k{2j}$ with composition factors $\Irre{k}$, $\Irre{k^+}$, \dots, $\Irre{k^{2j}}$ and non-split short exact sequence
\begin{equation}\label{eq:exactBk2j}
\dses{\TheB k{2(j-1)}}{}{\TheB k{2j}}{}{\Stan{k^{2j-1}}}.
\end{equation}
Our goal is to compute $\Ext(\Stan{k^{2j+1}},\TheB k{2j})$.  To do this, we first show that $\Hom(\Stan{k^{2j}}, \TheB k{2j}) \cong \CC$.  This follows from the Hom-Ext long exact sequence \eqref{es:lcohomology}, derived from \eqref{eq:exactBk2j}:
\begin{equation}
0\lra
\Hom(\Stan{k^{2j}},\TheB k{2(j-1)})\lra
\Hom(\Stan{k^{2j}},\TheB k{2j})\lra
\Hom(\Stan{k^{2j}},\Stan{k^{2j-1}})\lra
\Ext(\Stan{k^{2j}},\TheB k{2(j-1)})\lra \cdots.
\end{equation}
Indeed, $\Stan{k^{2j}}$ and $\TheB k{2(j-1)}$ have no common composition factor, hence their Hom-group is zero, and their Ext-group is also zero because their composition factors are sufficiently separated (see \cref{lem:trivialExtensions}\ref{it:Ext=0}).  As $\Hom(\Stan{k^{2j}},\Stan{k^{2j-1}}) \cong \CC$, by \cref{prop:SIPHomo}, we obtain $\Hom(\Stan{k^{2j}}, \TheB k{2j})\simeq \CC$, as claimed.

The desired Ext-group is now computed from the exact sequence
\begin{equation}
0=\Hom(\Proj{k^{2j+1}},\TheB k{2j})\lra
\Hom(\Stan{k^{2j}},\TheB k{2j})\lra
\Ext(\Stan{k^{2j+1}},\TheB k{2j})\lra
\Ext(\Proj{k^{2j+1}},\TheB k{2j})=0.
\end{equation}
This follows from the Hom-Ext sequence \eqref{es:rcohomology}, based on \eqref{eq:ProjExact}, by noting that the first Hom-group is zero, because $\Irre{k^{2j+1}}$ is not a composition factor of $\TheB k{2j}$, and that the last Ext-group is zero, because $\Proj{k^{2j+1}}$ is projective.  It follows that $\Ext(\Stan{k^{2j+1}},\TheB k{2j}) \cong \Hom(\Stan{k^{2j}}, \TheB k{2j}) \cong \CC$. This result leads to the definition of $\TheB{k}{2(j+1)}$ as any representative of the corresponding non-trivial extension. Its composition factors are clearly $\Irre{k}$, $\Irre{k^+}$, \dots, $\Irre{k^{2j}}$, $\Irre{k^{2j+1}}$, $\Irre{k^{2(j+1)}}$.  Indecomposability follows from \cref{lem:indecExtension} as usual, hence the $(j+1)$-th step of the recursion is complete.

We remark that this recursion may be continued as long as there are integers larger than $k^{2j}$ in the (non-critical) orbit of $k$, that is, as long as $k^{2j+1}\in\Lambda$. In the case where $k^{2j+1} = k_R$, so that $k^{2(j+1)}\notin\Lambda$, the above computation remains valid, but we shall denote the resulting module by $\TheB k{2j+1}$ to underline the fact that it only contains $2j+2$ composition factors, instead of the $2j+3$ factors possessed by the other $\TheB k{2(j+1)}$.

A similar recursive construction can be used to construct a second family of modules $\TheB{k}{2j+1}$, $j=0, 1, \dots$, starting now from $\TheB k1\equiv \Cost k$, for any non-critical $k$ in $\Lambda_0$ smaller than $k_R^-$. The module $\TheB {k}{2j+1}$ is then defined, if $k^{2j+1}<k_R$, to be a non-trivial extension described by $\Ext(\Cost{k}, \TheB{k^{++}}{2j-1})\simeq \mathbb C$ (when $k \in \Lambda_0$ of course). The computation of this Ext-group is again done recursively, uses now the sequence \eqref{es:CPC}, and copies otherwise the previous argument.  Note that, in the first family, composition factors are added to the right of existing ones, but in this new family they are added to the left. The composition factors of $\TheB{k}{2j+1}$ are thus $\Irre{k},\Irre{k^{+}},\dots,\Irre{k^{2j}},\Irre{k^{2j+1}}$ and are even in number. For this second family, the process stops whenever the index of the costandard module $\Cost{k}$ that would be used to extend $\TheB{k^{++}}{2j-1}$ falls outside $\Lambda_0$. The constraints $k<k_R^-$ for $\TheB k1$ and $k^{2j+1}<k_R$ for $\TheB {k}{2j+1}$ ensure that $\Irre{k_R}$ is never a composition factor of an indecomposable module of this second family. Note that the first family contains modules $\TheB{k}{2j+1}$ with an even number of composition factors, but that they all have $\Irre{k_R}$ as composition factor. In this way, these constructions never produce modules with the same labels using different means:  the notation $\TheB{k}{l}$ is well defined.

Finally, the duals of the modules that we have constructed above will be denoted by $\TheT{k}{j} \equiv \twdu{(\TheB{k}{j})}$.  As duality is exact contravariant (\cref{prop:DualExact}), this is equivalent to dualising the above inductive definitions:  $\TheT{k}{2j}$ is realised through the non-trivial extensions of $\Ext(\TheT{k}{2(j-1)},\Cost{k^{2j-1}})$ and $\TheT{k}{2j+1}$ through those of $\Ext(\TheT{k^{++}}{2j-1},\Stan{k})$.
\begin{proposition}\label{prop:3point1}
Given the recursive constructions above, the $\TheB kl$ and $\TheT kl$ are indecomposable and appear in the following non-split exact sequences:
\begin{subequations} \label{es:DefBandT}
\begin{align}
&\dses{\TheB k{2(j-1)}}{}{\TheB k{2j}}{}{\Stan{k^{2j-1}}}, & &\dses{\Cost{k^{2j-1}}}{}{\TheT k{2j}}{}{\TheT k{2(j-1)}}, \\
&\dses{\TheB{k^{++}}{2j-1}}{}{\TheB k{2j+1}}{}{\Cost k}, & &\dses{\Stan k}{}{\TheT k{2j+1}}{}{\TheT{k^{++}}{2j-1}} & &\text{(\(k^{2j+1} \neq k_R\)).}
\shortintertext{Moreover, the non-split exact sequences for $\TheB{k}{2j+1}$ and $\TheT{k}{2j+1}$, when $k^{2j+1}=k_R$, are instead}
&\dses{\TheB k{2j}}{}{\TheB k{2j+1}}{}{\Irre{k^{2j+1}}}, & &\dses{\Irre{k^{2j+1}}}{}{\TheT k{2j+1}}{}{\TheT k{2j}} & &\text{(\(k^{2j+1} = k_R\)).}
\end{align}
\end{subequations}
The $\TheB kl$ and $\TheT kl$ have $l+1$ composition factors, namely $\Irre k$, $\Irre{k^+}$, \dots, $\Irre{k^l}$, and they represent mutually non-isomorphic classes of indecomposable modules, except for $\TheB k0=\TheT k0=\Irre k$.
\end{proposition}
\begin{proof}
Only the last statement remains to be proved. The indices of the composition factors of the $\TheB{k}{l}$ and the $\TheT{k}{l}$ are consecutive integers, from $k$ to $k^l$, in the non-critical orbit $[k]$. If two of these modules have the same indices, then they must have the same values of $k$ and $l$. Thus, only $\TheB kl$ and $\TheT kl$ could be isomorphic. For $l>0$, $\TheB kl$ and $\TheT kl$ are both reducible, but indecomposable, so they do not coincide with their socles. Let us first study the case with $l$ even. Because $\TheB k0\simeq\Irre k$ and $\TheB k{2(j-1)}\subset \TheB k{2j}$, it follows that $\Irre k$ is in the socle of $\TheB k{2j}$. Since $\TheT k{2j} = \twdu{(\TheB k{2j})}$, $\Irre k$ is in its head (\cref{sub:Loewy}). Thus, $\TheB k{2j}$ and $\TheT k{2j}$ cannot be isomorphic for $j>0$. The same argument also takes care of the pair $\TheB k{2j+1}$ and $\TheT k{2j+1}$ when $k^{2j+1}=k_R$. Finally, the argument for $l$ odd copies the previous one, but uses the irreducible $\Irre{k^{2j+1}}$ that is in the head of $\TheB k{2j+1}$ but in the socle of $\TheT {k}{2j+1}$.
\end{proof}

%
%
\subsection{Their Loewy diagrams}\label{sub:LoewyBT}

This section clarifies the structure of the new modules $\TheB kl$ and $\TheT kl$ by identifying their socles, radicals and heads (see \cref{sub:Loewy}). The Loewy diagrams for the $\TheB kl$ and $\TheT kl$ are easily drawn from these data. Moreover, their injective hulls and projective covers will be obtained as immediate consequences of \cref{prop:ProjInj}.
\begin{proposition}\label{prop:socRadHdBT}
The modules $\TheB kl$ and $\TheT kl$, for $l\ge 1$, all have Loewy length $2$.  The head, socle and radical of each are given in the following table:
\begin{center}
\begin{tabular}{C|DDDD}
\Mod{M} & \TheB{k}{2j} & \TheB{k}{2j+1} & \TheT{k}{2j} & \TheT{k}{2j+1} \\
\hline
\head \Mod{M} & \bigoplus_{i=0}^{j-1} \Irre{k^{2i+1}} & \bigoplus_{i=0}^j \Irre{k^{2i+1}} & \bigoplus_{i=0}^j \Irre{k^{2i}} & \bigoplus_{i=0}^j \Irre{k^{2i}} \\
\soc \Mod{M} = \rad \Mod{M} & \bigoplus_{i=0}^j \Irre{k^{2i}} & \bigoplus_{i=0}^j \Irre{k^{2i}} & \bigoplus_{i=0}^{j-1} \Irre{k^{2i+1}} & \bigoplus_{i=0}^j \Irre{k^{2i+1}} \ .
\end{tabular}
\end{center}
Of course, $\TheB{k}{0} \cong \TheT{k}{0} \cong \Irre{k}$, hence in this case, the radical is $0$ and the socle and the head are $\Irre{k}$.
\end{proposition}
\noindent Note that this table is consistent with the general result $\soc\Mod{M} \cong \head \twdu{\Mod{M}}$.

Before giving a proof, we use the proposition to draw the Loewy diagrams of the $\TheB kl$ and $\TheT kl$, with $l>0$, thus revealing their ``zigzag'' structure:
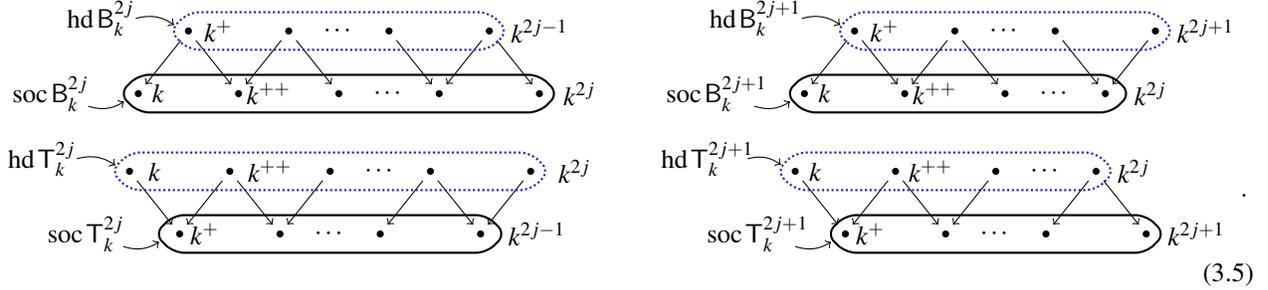
\begin{equation}
\begin{gathered}
\begin{tikzpicture}[baseline={(current bounding box.center)},scale=1/3]
\node (m3) at (6,2.5) [] {};
\node (m5) at (10,2.5) [] {};
\node (m7) at (14,2.5) [] {};
\node (m9) at (18,2.5) [] {};
\node (b2) at (4,0) [] {};
\node (b4) at (8,0) [] {};
\node (b6) at (12,0) [] {};
\node (b8) at (16,0) [] {};
\node (b10) at (20,0) [] {};
\dodu{b2}\dodu{b4}\dodu{b6}\dodu{b8}\dodu{b10}
\dodu{m3}\dodu{m5}\dodu{m7}\dodu{m9}
\draw [->] (m3) -- (b2);
\draw [->] (m3) -- (b4); \draw [->] (m5) -- (b4);
\draw [->] (m5) -- (b6); \draw [->] (m7) -- (b8);
\draw [->] (m9) -- (b8); \draw [->] (m9) -- (b10);
\node at (4.8,0) {$k$};
\node at (9.2,0) {$k^{++}$};
\node at (7.2,2.5) {$k^+$};
\node at (20.0,2.5) {$k^{2j-1}$};
\node at (21.6,0) {$k^{2j}$};
\node at (12,2.5) {$\cdots$};
\node at (14,0) {$\cdots$};
\node (theM) at (2.5,3) [] {$\head\TheB k{2j}$};
\node (flecheM) at (5.8,2.4) [] {};
\node (theA) at (0.5,0.0) [] {$\soc\TheB k{2j}$};
\node (flecheA) at (3.8,0) [] {};
\draw [->] (3.9,3) to [bend right=-25] (flecheM);
\draw [->] (2,-0.5) to [bend right=25] (flecheA);
\draw[bleu,densely dotted,thick,rounded corners] (8,1.75) -- (6,1.75) -- (5.25,2.5) -- (6,3.25) -- (18,3.25) -- (18.75,2.5) -- (18,1.75) -- (8.05,1.75) ;
\draw[black,thick,rounded corners] (6,-.75) -- (4,-.75) -- (3.25,0) -- (4,.75) -- (20,.75) -- (20.75,0) -- (20,-.75) -- (6,-.75) ;
\end{tikzpicture}
\quad\quad
\begin{tikzpicture}[baseline={(current bounding box.center)},scale=1/3]
\node (m3) at (6,2.5) [] {};
\node (m5) at (10,2.5) [] {};
\node (m7) at (14,2.5) [] {};
\node (m9) at (18,2.5) [] {};
\node (b2) at (4,0) [] {};
\node (b4) at (8,0) [] {};
\node (b6) at (12,0) [] {};
\node (b8) at (16,0) [] {};
\node (b10) at (20,0) [] {};
\dodu{b2}\dodu{b4}\dodu{b6}\dodu{b8}
\dodu{m3}\dodu{m5}\dodu{m7}\dodu{m9}
\draw [->] (m3) -- (b2);
\draw [->] (m3) -- (b4); \draw [->] (m5) -- (b4);
\draw [->] (m5) -- (b6); \draw [->] (m7) -- (b8);
\draw [->] (m9) -- (b8); 
\node at (4.8,0) {$k$};
\node at (9.2,0) {$k^{++}$};
\node at (7.2,2.5) {$k^+$};
\node at (20.0,2.5) {$k^{2j+1}$};
\node at (17.8,0) {$k^{2j}$};
\node at (12,2.5) {$\cdots$};
\node at (14,0) {$\cdots$};
\node (theM) at (2.05,3) [] {$\head\TheB k{2j+1}$};
\node (flecheM) at (5.8,2.4) [] {};
\node (theA) at (0.5,0.0) [] {$\soc\TheB k{2j+1}$};
\node (flecheA) at (3.8,0) [] {};
\draw [->] (3.9,3) to [bend right=-25] (flecheM);
\draw [->] (2,-0.5) to [bend right=25] (flecheA);
\draw[bleu,densely dotted,thick,rounded corners] (8,1.75) -- (6,1.75) -- (5.25,2.5) -- (6,3.25) -- (18,3.25) -- (18.75,2.5) -- (18,1.75) -- (8.05,1.75) ;
\draw[black,thick,rounded corners] (6,-.75) -- (4,-.75) -- (3.25,0) -- (4,.75) -- (16.25,.75) -- (17.0,0) -- (16.25,-.75) -- (6,-.75) ;
\end{tikzpicture}
\\
\begin{tikzpicture}[baseline={(current bounding box.center)},scale=1/3]
\node (m1) at (2,2.5) [] {};
\node (m3) at (6,2.5) [] {};
\node (m5) at (10,2.5) [] {};
\node (m7) at (14,2.5) [] {};
\node (m9) at (18,2.5) [] {};
\node (b2) at (4,0) [] {};
\node (b4) at (8,0) [] {};
\node (b6) at (12,0) [] {};
\node (b8) at (16,0) [] {};
\dodu{b2}\dodu{b4}\dodu{b6}\dodu{b8}
\dodu{m3}\dodu{m1}\dodu{m5}\dodu{m7}\dodu{m9}
\draw [->] (m1) -- (b2);
\draw [->] (m3) -- (b2);
\draw [->] (m3) -- (b4); \draw [->] (m5) -- (b4);
\draw [->] (m7) -- (b6);
\draw [->] (m7) -- (b8); \draw [->] (m9) -- (b8);
\node at (3.0,2.5) {$k$};
\node at (5.0,0) {$k^{+}$};
\node at (7.6,2.5) {$k^{++}$};
\node at (19.75,2.5) {$k^{2j}$};
\node at (18.25,0) {$k^{2j-1}$};
\node at (12,2.5)  {$\cdots$};
\node at (10,0)  {$\cdots$};
\node (theM) at (-1.5,3) [] {$\head\TheT k{2j}$};
\node (flecheM) at (1.8,2.4) [] {};
\node (theA) at (0.25,0.0) [] {$\soc\TheT k{2j}$};
\node (flecheA) at (3.55,0) [] {};
\draw [->] (-0.1,3) to [bend right=-25] (flecheM);
\draw [->] (1.75,-0.5) to [bend right=25] (flecheA);
\draw[bleu,densely dotted,thick,rounded corners] (4,1.75) -- (2,1.75) -- (1.25,2.5) -- (2,3.25) -- (18,3.25) -- (18.75,2.5) -- (18,1.75) -- (4.05,1.75) ;
\draw[black,thick,rounded corners] (5.75,-.75) -- (3.75,-.75) -- (3.0,0) -- (3.75,.75) -- (16.25,.75) -- (17.0,0) -- (16.25,-.75) -- (5.75,-.75) ;
\end{tikzpicture}
\quad\quad
\begin{tikzpicture}[baseline={(current bounding box.center)},scale=1/3]
\node (m3) at (6,0) [] {};
\node (m5) at (10,0) [] {};
\node (m7) at (14,0) [] {};
\node (m9) at (18,0) [] {};
\node (b2) at (4,2.5) [] {};
\node (b4) at (8,2.5) [] {};
\node (b6) at (12,2.5) [] {};
\node (b8) at (16,2.5) [] {};
\dodu{b2}\dodu{b4}\dodu{b6}\dodu{b8}
\dodu{m3}\dodu{m5}\dodu{m7}\dodu{m9}
\draw [<-] (m3) -- (b2);
\draw [<-] (m3) -- (b4); \draw [<-] (m5) -- (b4);
\draw [<-] (m5) -- (b6); 
\draw [<-] (m7) -- (b8); \draw [<-] (m9) -- (b8);
\node at (4.8,2.5) {$k$};
\node at (9.4,2.5) {$k^{++}$};
\node at (12,0) {$\cdots$};
\node at (14,2.5) {$\cdots$};
\node at (7.0,0) {$k^+$};
\node at (17.5,2.5) {$k^{2j}$};
\node at (20.0,0) {$k^{2j+1}$};
\node (theM) at (0.5,3) [] {$\head\TheT k{2j+1}$};
\node (flecheM) at (3.8,2.4) [] {};
\node (theA) at (2.5,0.0) [] {$\soc\TheT k{2j+1}$};
\node (flecheA) at (5.8,0) [] {};
\draw [->] (2.1,3) to [bend right=-25] (flecheM);
\draw [->] (4,-0.5) to [bend right=25] (flecheA);
\draw[bleu,densely dotted,thick,rounded corners] (8,1.75) -- (4,1.75) -- (3.25,2.5) -- (4,3.25) -- (16,3.25) -- (16.75,2.5) -- (16,1.75) -- (8.05,1.75) ;
\draw[black,thick,rounded corners] (8,-.75) -- (6,-.75) -- (5.25,0) -- (6,.75) -- (18,.75) -- (18.75,0) -- (18,-.75) -- (8,-.75) ;
\end{tikzpicture}
\ .
\end{gathered}
\end{equation}
To lighten the notation, we have replaced the composition factors $\Irre{k'}$ by dots, labelling (some of) them by the corresponding index $k'$ (the missing labels should be clear). As before, the socle forms the bottom row and it is circled by a solid line, while the head forms the top row and is circled by a dotted line. There can only be arrows from the head to the socle between neighbouring composition factors as the extension groups $\Ext(\Irre k, \Irre{k'})$ (\cref{prop:SIExte}) forbid other possibilities. Finally, none of the arrows indicated in these diagrams may be omitted as the result would indicate a decomposable module. We remark that the notation $\TheB kl$ ($\TheT kl$) was chosen as a reminder that the composition factor $\Irre k$ appears in the bottom (top) layer of the Loewy diagram.  The composition factor with the highest index is $\Irre{k^l}$.
\begin{proof}[Proof of \cref{prop:socRadHdBT}]
We provide details for the first family of modules $\TheB k{2j}$, those for the second family $\TheB{k}{2j+1}$ being similar and the results for the $\TheT{k}{l}$ then being obtained by duality.

The socle of $\TheB k{2j}$ is obtained by induction on $j$. For $j=0$, we have $\TheB k0\equiv \Irre k$, hence $\soc\TheB{k}{0} = \Irre{k}$ as required. For general $j$, we begin with the short exact sequence
\begin{equation}
\dses{\TheB k{2(j-1)}}{\iota}{\TheB k{2j}}{\pi}{\Stan{k^{2j-1}}}.
\end{equation}
Since the image of a semisimple module is semisimple, $\iota(\soc \TheB k{2(j-1)})\subseteq \soc\TheB k{2j}$ and so the latter must contain the composition factors $\Irre{k}$, $\Irre{k^{++}}$, \dots, $\Irre{k^{2(j-1)}}$. Moreover, $\Mod{L}\subseteq \Mod{M}$ implies that $\soc\Mod{L}=\Mod{L}\cap\soc \Mod{M}$, so $\soc\TheB k{2j}$ cannot contain any of the composition factors $\Irre{k^+}$, $\Irre{k^{+++}}$, \dots, $\Irre{k^{2j-3}}$. If $\Irre{k^{2j-1}}$ were in $\soc\TheB k{2j}$, then $\pi$ would map it into $\soc\Stan{k^{2j-1}} \cong \Irre{k^{2j}}$, hence $\Irre{k^{2j-1}}$ would be in $\ker \pi = \im \iota$.  This contradicts the fact that $\TheB{k}{2(j-1)}$ does not have $\Irre{k^{2j-1}}$ as a composition factor, so it follows that $\Irre{k^{2j-1}}$ is not in $\soc\TheB k{2j}$.  Suppose finally that $\Irre{k^{2j}}$ is not in $\soc\TheB k{2j}$.  Then, we would have
\begin{equation}
\soc\TheB k{2j}\simeq \soc\TheB k{2(j-1)}\simeq \Irre{k} \oplus \Irre{k^{++}} \oplus \dots \oplus \Irre{k^{2(j-1)}}.
\end{equation}
As the injective hull of a module coincides with that of its socle, that of $\TheB k{2j}$ would now be $\Inje{k} \oplus \Inje{k^{++}} \oplus \dots \oplus \Inje{k^{2(j-1)}}$. But, this hull does not have $\Irre{k^{2j}}$ as a composition factor, whereas $\TheB k{2j}$ does, another contradiction. We therefore conclude that the socle of $\TheB k{2j}$ is $\Irre{k} \oplus \Irre{k^{++}} \oplus \dots \oplus \Irre{k^{2j}}$, as required.

We now prove that the radical and socle of $\TheB k{2j}$ coincide, for $j>0$, a consequence of this being that the head is the direct sum of the composition factors that are not in the socle:  $\head\TheB k{2j} \cong \Irre{k^+} \oplus \dots \oplus \Irre{k^{2j-1}}$.  The proof will follow immediately upon constructing an injection
\begin{equation} \label{eq:l'injection}
\frac{\TheB{k}{2j}}{\soc\TheB{k}{2j}} \ira \frac{\rad\Inje{}[\TheB{k}{2j}]}{\soc\Inje{}[\TheB{k}{2j}]},
\end{equation}
where $\Inje{}[\TheB{k}{2j}] \cong \Inje{k} \oplus \Inje{k^{++}} \oplus \dots \oplus \Inje{k^{2j}}$ denotes the injective hull of $\TheB{k}{2j}$.  To see why, recall from \cref{prop:ProjInj} that the injective modules all have Loewy lengths at most $3$.  The quotient $\rad\Inje{}[\TheB{k}{2j}] / \soc\Inje{}[\TheB{k}{2j}]$ therefore has Loewy length at most $1$, meaning that it is semisimple.  The injection \eqref{eq:l'injection} will therefore establish that $\TheB{k}{2j} / \soc\TheB{k}{2j}$ is semisimple, hence that $\TheB{k}{2j}$ has Loewy length $2$ (it does not coincide with its socle, for $j>0$), so its radical equals its socle.

It remains to construct the injection \eqref{eq:l'injection}. Suppose first that $k\neq k_L$. Then, $\Inje{k^{2i}} \cong \Proj{k^{2i}}$ has two composition factors isomorphic to $\Irre{k^{2i}}$, one contributing to the socle and the other to the head.  As $\TheB k{2j}$ has a single composition factor isomorphic to $\Irre{k^{2i}}$, for each $0\le i\le j$, any morphism $f\colon\TheB k{2j}\to \Inje{}[\TheB{k}{2j}] \cong \Inje{k} \oplus \Inje{k^{++}} \oplus \dots \oplus \Inje{k^{2j}}$ will send this composition factor to that of the \emph{socle} of $\Inje{k^{2i}}$ (or to zero).  Indeed, this composition factor belongs to the socle of $\TheB{k}{2j}$, so it follows that $\soc\TheB{k}{2j}$ is mapped into $\soc\Inje{}[\TheB{k}{2j}]$.  Moreover, the image of $f$ thus never includes the composition factors corresponding to the heads of the $\Inje{k^{2i}}$, hence it will lie in $\rad\Inje{}[\TheB{k}{2j}]$.  This shows that any $f\colon\TheB k{2j}\to \Inje{}[\TheB{k}{2j}]$ will induce a map as in \eqref{eq:l'injection}.  To find an injective map and complete the proof, it suffices to take $f$ injective (which is always possible by the definition of injective hulls) because then $f$ maps $\soc\TheB{k}{2j}$ onto $\soc\Inje{}[\TheB{k}{2j}]$.

If $k=k_L$, the injective $\Inje{k_L}$ is isomorphic to $\Cost{k_L}$ whose Loewy length is $2$. In this case, the argument goes through if $\rad\Inje{}[\TheB{k}{2j}] \cong \rad(\Inje{k} \oplus \Inje{k^{++}} \oplus \dots \oplus \Inje{k^{2j}})$ is replaced throughout by $\Inje{k} \oplus \rad(\Inje{k^{++}} \oplus \dots \oplus \Inje{k^{2j}})$.
\end{proof}

As the projective cover $\Proj{}[\Mod{M}]$ (injective hull $\Inje{}[\Mod{M}]$) of a module $\Mod{M}$ is isomorphic to that of its head (socle), \cref{prop:socRadHdBT} immediately identifies these data for the $\TheB{k}{l}$ and $\TheT{k}{l}$.
\begin{corollary} \label{cor:BTProjInj}
The projective covers and injective hulls of the $\TheB{k}{l}$ and $\TheT{k}{l}$, with $l>0$, are as follows:
\begin{center}
\begin{tabular}{C|DDDD}
\Mod{M} & \TheB{k}{2j} & \TheB{k}{2j+1} & \TheT{k}{2j} & \TheT{k}{2j+1} \\
\hline
\Proj{}[\Mod{M}] & \bigoplus_{i=0}^{j-1} \Proj{k^{2i+1}} & \bigoplus_{i=0}^j \Proj{k^{2i+1}} & \bigoplus_{i=0}^j \Proj{k^{2i}} & \bigoplus_{i=0}^j \Proj{k^{2i}} \\
\Inje{}[\Mod{M}] & \bigoplus_{i=0}^j \Inje{k^{2i}} & \bigoplus_{i=0}^j \Inje{k^{2i}} & \bigoplus_{i=0}^{j-1} \Inje{k^{2i+1}} & \bigoplus_{i=0}^j \Inje{k^{2i+1}} \ .
\end{tabular}
\end{center}
As $\TheB{k}{0} \cong \TheT{k}{0} \cong \Irre{k}$, the projective covers and injective hulls, when $l=0$, are $\Proj{k}$ and $\Inje{k}$, respectively.
\end{corollary}
We remark that, as in \cref{sec:Dual}, we are again neglecting to specify the surjective (injective) morphisms that complete the description of these projective covers (injective hulls).  However, it is easy to check that these morphisms are unique, up to a scaling factor for each projective (injective) indecomposable appearing in the cover (hull).

\subsection{Their projective and injective presentations} \label{sec:ProjInjPres}

This section establishes projective and injective presentations of the $\TheB{k}{l}$ and $\TheT{k}{l}$.  More precisely, we determine the kernel (cokernel) of the projection (inclusion) from each of these modules to its projective cover (injective hull).  This will require precise relations between these families and the indecomposable projectives and injectives.  Define then, for each non-critical $k\neq k_L,k_R$, the $\Alg{}$-modules $\TheA{k^-}$ and $\TheV{k^-}$ to be $\Proj k/\soc\Proj k$ and $\rad\Proj k$, respectively. These definitions immediately yield the following short exact sequences
\begin{equation} \label{es:DefAV}
\dses{\Irre{k}}{}{\Proj{k}}{}{\TheA{k^-}}, \qquad \dses{\TheV{k^-}}{}{\Proj{k}}{}{\Irre{k}},
\end{equation}
because $\soc \Proj{k} \cong \Proj{k} / \rad \Proj{k} \cong \Irre{k}$ for all non-critical $k$ with $k \neq k_L$.
\begin{lemma}\label{lem:theAandV}
For $k\neq k_L$, $k_R$, we have the following non-split short exact sequences:
\begin{subequations}
\begin{align}
&\dses{\Irre{k^-}}{}{\TheA{k^-}}{}{\Stan{k}}, & &\dses{\Stan{k^-}}{}{\TheV{k^-}}{}{\Irre{k^+}},\label{es:AV}\\
&\dses{\Irre{k^+}}{}{\TheA{k^-}}{}{\Cost{k^-}}, & &\dses{\Cost{k}}{}{\TheV{k^-}}{}{\Irre{k^-}}\label{esdual:AV}.
\end{align}
\end{subequations}
In particular, we have the identifications $\TheA{k^-} \cong \TheB{k^-}{2}$ and $\TheV{k^-} \cong \TheT{k^-}{2}$.
\end{lemma}
\begin{proof}
We first prove that the sequence
\begin{equation} \label{es:IAS}
\dses{\Irre{k^-}}{}{\TheA{k^-}}{}{\Stan{k}}
\end{equation}
is exact.
The top row of the diagram
\begin{equation} \label{cd:Ak}
\begin{tikzpicture}[scale=1/3,baseline=(current bounding box.center)]
\node (t1) at (5,5) [] {$0$};
\node (t2) at (10,5) [] {$\Irre{k}$};
\node (t3) at (15,5) [] {$\Stan{k^-}$};
\node (t4) at (20,5) [] {$\Irre{k^-}$};
\node (t5) at (25,5) [] {$0$};
\node (b1) at (5,0) [] {$0$};
\node (b2) at (10,0) [] {$\Irre{k}$};
\node (b3) at (15,0) [] {$\Proj{k}$};
\node (b4) at (20,0) [] {$\TheA{k^-}$};
\node (b5) at (25,0) [] {$0$};
\draw [->] (t1) -- (t2);
\draw [->] (t2) to node [above] {$\iota$} (t3);
\draw [->] (t3) to node [above] {$\pi$}  (t4);
\draw [->] (t4) -- (t5);
\draw [->] (b1) -- (b2);
\draw [->] (b2) to node [above] {$\bar{\iota}$}  (b3);
\draw [->] (b3) to node [above] {$\bar{\pi}$}  (b4);
\draw [->] (b4) -- (b5);
\draw [->] (t2) to node [left] {$\id$} (b2);
\draw [->] (t3) to node [left] {$i$} (b3);
\draw [->] ([xshift=-0.2cm]t4.south) to node [left] {$\phi$} ([xshift=-0.2cm]b4.north);
\draw [dashed,->] ([xshift=0.2cm]b4.north) to node [right] {$\bar{\phi}$} ([xshift=0.2cm]t4.south);
\end{tikzpicture}
\end{equation}
is the exact sequence \eqref{eq:StanExact}. Let $i$ denote the inclusion of \eqref{eq:ProjExact}, so that $\bar{\iota} \equiv i \iota$ is injective.  The left square thus commutes and we may choose $\bar{\pi}$ so that the bottom row is exact, because of $\Hom(\Irre{k},\Proj{k}) \cong \CC$ (\cref{prop:SIPHomo}) and \cref{es:DefAV}.  Because $\bar{\pi} i \iota = \bar{\pi} \bar{\iota} = 0$, one may now define $\phi$ so that the right square commutes. The snake lemma then gives
\begin{equation}
\ker i = 0 \lra \ker \phi \lra \Coker \id = 0 \lra \Coker i = \Stan{k} \lra \Coker \phi \lra 0,
\end{equation}
hence $\ker \phi \simeq 0$ and $\Coker\phi \simeq \Stan{k}$, which settles the exactness of \eqref{es:IAS}.

We next show that \eqref{es:IAS} does not split.  If it did, then there would exist $\bar\phi:\TheA{k^-}\rightarrow \Irre{k^-}$ such that $\bar\phi\phi=\id$ on $\Irre{k^-}$. But then $\bar{\phi}\bar{\pi}$ cannot be zero (both are surjective), contradicting $\Hom(\Proj{k},\Irre{k^-}) = 0$ (\cref{prop:SIPHomo}).  As non-split extensions of $\Stan{k}$ by $\Irre{k^-}$ are unique up to isomorphism (\cref{prop:SIExte}), it follows from the definitions in \cref{sub:definitionBT} that $\TheA{k^-} \cong \TheB{k^-}{2}$.

The method is easily adapted to prove the remaining short exact sequences. The only conceptual difference for the two in \eqref{esdual:AV} is that we use the duals of \eqref{eq:StanExact} and \eqref{eq:ProjExact}, remembering that duality is exact contravariant (\cref{prop:DualExact}), and that the dual of the first sequence becomes the bottom row of the diagram rather than the top.
\end{proof}

Before stating the main result of this section, \cref{prop:cokerBT}, we need another lemma. It indicates that the two recursively defined families, the $\TheB{k}{2j}$ and $\TheB{k}{2j+1}$, are intimately related by proving the exactness of two sequences.  The first sequence shows that the first family may be constructed recursively by extending a costandard module (as in the definition given for the second family in \cref{sub:definitionBT}).  The second sequence then shows that the first family may be constructed by extending members of the second family by an irreducible.

\begin{lemma}\label{lem:tiesBetweenFamilies}
For $j\ge 1$ such that $k^{2j}\in\Lambda$, the following short sequences are exact:
\begin{equation}\label{eq:BevenWithCostAndIrre}
\dses{\TheB{k^{++}}{2(j-1)}}{}{\TheB{k}{2j}}{}{\Cost k}\qquad\text{and}\qquad
\dses{\Irre{k^{2j}}}{}{\TheB k{2j}}{}{\TheB{k}{2j-1}}.
\end{equation}
\end{lemma}
\begin{proof}
We consider the exactness of the first of the two sequences, proceeding by induction on $j$. The first sequence of \eqref{esdual:AV} establishes the case $j=1$. Assume then the exactness for $j$ and suppose that $k^{2(j+1)} \in \Lambda$.  In the diagram
\begin{equation} \label{cd:tiesBetweenFamilies1}
\begin{tikzpicture}[scale=1/3,baseline=(current bounding box.center)]
\node (tt3) at (17,8) [] {$0$};
\node (t1) at (5,5) [] {$0$};
\node (t2) at (10,5) [] {$\TheB{k^{++}}{2(j-1)}$};
\node (t3) at (17,5) [] {$\TheB{k}{2j}$};
\node (t4) at (24,5) [] {$\Cost k$};
\node (t5) at (29,5) [] {$0$\phantom{ ,}};
\node (m1) at (5,0) [] {$0$};
\node (m2) at (10,0) [] {$\TheB{k^{++}}{2(j-1)}$};
\node (m3) at (17,0) [] {$\TheB{k}{2(j+1)}$};
\node (m4) at (24,0) [] {$\coker \beta$};
\node (m5) at (29,0) [] {$0$ ,};
\node (b3) at (17,-5) [] {$\Stan{k^{2j+1}}$};
\node (bb3) at (17,-8) [] {$0$};
\draw [->] (tt3) -- (t3);
\draw [->] (t1) -- (t2);
\draw [->] (t2) to node [above] {$\gamma$} (t3);
\draw [->] (t3) to node [above] {$\overline\gamma$} (t4);
\draw [->] (t4) -- (t5);
\draw [->] (m1) -- (m2);
\draw [->] (m2) to node [above] {$\beta$}  (m3);
\draw [->] (m3) to node [above] {$\overline{\beta}$}  (m4);
\draw [->] (m4) -- (m5);
\draw [->] (t2) to node [left] {$\id$} (m2);
\draw [->] (t3) to node [left] {$\alpha$} (m3);
\draw [dashed,->] ([xshift=-0.2cm]t4.south) to node [left] {$\delta$} ([xshift=-0.2cm]m4.north);
\draw [dashed,->] ([xshift=0.2cm]m4.north) to node [right] {$\delta^*$} ([xshift=0.2cm]t4.south);
\draw [->] (m3) to node [left] {$\overline\alpha$} (b3);
\draw [->] (b3) -- (bb3);
\end{tikzpicture}
\end{equation}
the first row is thus assumed exact. We moreover define $\beta = \alpha\gamma$, so that the left square commutes, and then $\overline\beta$, so that the second row is exact.  Finally, the second column is exact, by the definition \eqref{es:DefBandT} of $\TheB k{2(j+1)}$. This setup guarantees that there exists a morphism $\delta$ making the right square commute. The snake lemma then gives $\ker\delta=0$ and $\coker\delta\simeq\Stan{k^{2j+1}}$, proving that $\ses{\Cost{k}}{\coker\beta}{\Stan{k^{2j+1}}}$ is exact. Since $\Ext(\Stan{k^{2j+1}},\Cost k)=0$ (\cref{prop:SIExte}), this sequence splits and there exists a surjection $\delta^*\colon\coker\beta\to \Cost k$ such that $\delta^*\delta$ is the identity on $\Cost{k}$.

Consider now the diagram
\begin{equation} \label{cd:tiesBetweenFamilies2}
\begin{tikzpicture}[scale=1/3,baseline=(A.base)]
\node (t1) at (5,5) [] {$0$};
\node (t2) at (10,5) [] {$\TheB{k^{++}}{2(j-1)}$};
\node (t3) at (15,5) [] {$\TheB{k}{2j}$};
\node (t4) at (20,5) [] {$\Cost k$};
\node (t5) at (25,5) [] {$0$};
\node (m1) at (5,0) [] {$0$};
\node (m2) at (10,0) [] {$\ker f$};
\node (m3) at (15,0) [] {$\TheB{k}{2(j+1)}$};
\node (m4) at (20,0) [] {$\Cost k$};
\node (m5) at (25,0) [] {$0$};
\draw [->] (t1) -- (t2);
\draw [->] (t2) to node [above] {$\gamma$} (t3);
\draw [->] (t3) to node [above] {$\overline\gamma$} (t4);
\draw [->] (t4) -- (t5);
\draw [->] (m1) -- (m2);
\draw [->] (m2) -- (m3);
\draw [->] (m3) to node [above] {$f$}  (m4);
\draw [->] (m4) -- (m5);
\draw [dashed,->] (t2) to node [left] {$\phi$} (m2);
\draw [->] (t3) to node [left] (A) {$\alpha$} (m3);
\draw [->] (t4) to node [right] {$\id$} (m4);
\end{tikzpicture}
\ ,
\end{equation}
in which $f = \delta^* \overline\beta$.  Both rows are thus exact and the right square commutes:  $f \alpha = \delta^* \overline\beta \alpha = \delta^* \delta \overline\gamma = \overline\gamma$. It now follows from $f \alpha \gamma = \gamma \overline\gamma = 0$ that there is a morphism $\phi$ making the left square commute. The snake lemma then gives the exactness of $\ses{\TheB{k^{++}}{2(j-1)}}{\ker f}{\Stan{k^{2j+1}}}$, as before.  \cref{lem:indecExtension} and $\Ext(\Stan{k^{2j+1}}, \TheB{k^{++}}{2(j-1)}) = \CC$ (\cref{sub:definitionBT}) now imply that there are only two possibilities:  either this sequence is non-split, in which case it gives $\ker f \cong \TheB{k^{++}}{2j}$, by the definition of the latter, or it splits and $\ker f \cong \TheB{k^{++}}{2(j-1)} \oplus \Stan{k^{2j+1}}$.  However, if $\ker f \cong \TheB{k^{++}}{2(j-1)} \oplus \Stan{k^{2j+1}}$, then $\Stan{k^{2j+1}}$ would be a submodule of $\ker f$ and thus also of $\TheB{k}{2(j+1)}$. It would then follow that $\TheB{k}{2(j+1)}$ is the sum of two submodules, $\TheB{k}{2j}$ and $\Stan{k^{2j+1}}$, whose intersection is zero, contradicting its indecomposability.  We conclude that $\ker f \cong \TheB{k^{++}}{2j}$, hence that the second row of \eqref{cd:tiesBetweenFamilies2} is the desired exact sequence.

The proof of the exactness of the second sequence proceeds in a similar inductive fashion.  As the two sequences of \eqref{eq:BevenWithCostAndIrre} coincide for $j=1$, the base case has already been established.  Assuming the exactness for $j-1$, we consider the following diagram, similar to \eqref{cd:tiesBetweenFamilies1}:
\begin{equation} \label{cd:tiesBetweenFamilies3}
\begin{tikzpicture}[scale=1/3,baseline=(current bounding box.center)]
\node (tt3) at (16,8) [] {$0$};
\node (t1) at (5,5) [] {$0$};
\node (t2) at (10,5) [] {$\Irre{k^{2j}}$};
\node (t3) at (16,5) [] {$\TheB{k^{++}}{2(j-1)}$};
\node (t4) at (22,5) [] {$\TheB{k^{++}}{2j-3}$};
\node (t5) at (28,5) [] {$0$\phantom{ .}};
\node (m1) at (5,0) [] {$0$};
\node (m2) at (10,0) [] {$\Irre{k^{2j}}$};
\node (m3) at (16,0) [] {$\TheB{k}{2j}$};
\node (m4) at (22,0) [] {$\coker\beta$};
\node (m5) at (28,0) [] {$0$ .};
\node (b3) at (16,-5) [] {$\Cost k$};
\node (bb3) at (16,-8) [] {$0$};
\draw [->] (tt3) -- (t3);
\draw [->] (t1) -- (t2);
\draw [->] (t2) to node [above] {$\gamma$} (t3);
\draw [->] (t3) to node [above] {$\overline\gamma$} (t4);
\draw [->] (t4) -- (t5);
\draw [->] (m1) -- (m2);
\draw [->] (m2) to node [above] {$\beta$}  (m3);
\draw [->] (m3) -- (m4);
\draw [->] (m4) -- (m5);
\draw [->] (t2) to node [left] {$\id$} (m2);
\draw [->] (t3) to node [left] {$\alpha$} (m3);
\draw [dashed,->] (t4) to node [right] {$\delta$} (m4);
\draw [->] (m3) to node [left] {$\overline\alpha$} (b3);
\draw [->] (b3) -- (bb3);
\end{tikzpicture}
\end{equation}
The first row is exact by assumption and the second column is the first sequence of \eqref{eq:BevenWithCostAndIrre} (whose exactness has just been proved). Since the left square commutes, there exists a morphism $\delta$ that makes the right square commute. The snake lemma computes its kernel and cokernel which give the exact sequence $\ses{\TheB{k^{++}}{2j-3}}{\coker\beta}{\Cost{k}}$.
This time, \cref{lem:indecExtension} and $\Ext(\Cost{k}, \TheB{k^{++}}{2j-3}) = \CC$ (\cref{sub:definitionBT}) show that either $\coker \beta \cong \TheB{k}{2j-1}$, by the definition of the latter, or $\coker \beta$ splits as $\Cost{k} \oplus \TheB{k^{++}}{2j-3}$.  Now, $\coker \beta$ splitting would entail the existence of a non-zero morphism from $\Cost{k}$ to $\coker \beta$.  However, \cref{prop:SIPHomo,prop:SIExte} imply that the second row of \eqref{cd:tiesBetweenFamilies3} yields the following exact sequence:
\begin{equation}
0=\Hom(\Cost{k}, \Irre{k^{2j}}) \lra \Hom(\Cost{k}, \TheB k{2j}) \lra \Hom(\Cost{k}, \coker\beta) \lra \Ext(\Cost{k}, \Irre{k^{2j}})=0.
\end{equation}
Thus, $\Hom(\Cost{k}, \TheB k{2j})\simeq \Hom(\Cost{k}, \coker\beta)$ could not be $0$.  But, any non-zero map of $\Hom(\Cost{k}, \TheB k{2j})$ must be injective since $\head \Cost{k} \cong \Irre{k^+}$ is not a submodule of $\TheB k{2j}$ (\cref{prop:socRadHdBT}).  Thus, $\coker \beta$ splitting would imply that the indecomposable $\TheB k{2j}$ is the sum of two submodules, $\Cost{k}$ and $\TheB{k^{++}}{2(j-1)}$, whose intersection is zero.  This contradiction means that $\coker \beta \cong \TheB{k}{2j-1}$, hence that the second row of \eqref{cd:tiesBetweenFamilies3} is the desired exact sequence.
\end{proof}

With these lemmas in hand, we now turn to projective (injective) presentations.  More specifically, we compute the kernels (cokernels) of the projections (inclusions) that define the projective covers (injective hulls) of the modules $\TheB{k}{l}$ and $\TheT{k}{l}$.  As usual, Loewy diagrams provide an intuitive description of the result to come.
\begin{equation}
\begin{tikzpicture}[baseline={(current bounding box.center)},scale=1/3]
\node (t2) at (4,4) [] {};
\node (t4) at (8,4) [] {};
\node (t6) at (12,4) [] {};
\node (t8) at (16,4) [] {};
\node (m1) at (2,2) [] {};
\node (m3) at (6,2) [] {};
\node (m5) at (10,2) [] {};
\node (m7) at (14,2) [] {};
\node (m9) at (18,2) [] {};
\node (b2) at (4,0) [] {};
\node (b4) at (8,0) [] {};
\node (b6) at (12,0) [] {};
\node (b8) at (16,0) [] {};
\dodu{t2}\dodu{t4}\dodu{t6}\dodu{t8}
\dodu{m1}\dodu{m9}
\dodu{b2}\dodu{b4}\dodu{b6}\dodu{b8}
\dodu{5.8,2}\dodu{6.2,2}
\dodu{9.8,2}\dodu{10.2,2}
\dodu{13.8,2}\dodu{14.2,2}
\draw [->] (t2) -- (m1); \draw [->] (t2) -- (m3);
\draw [->] (t4) -- (m3); \draw [->] (t4) -- (m5);
\draw [->] (t6) -- (m5); \draw [->] (t6) -- (m7);
\draw [->] (t8) -- (m7); \draw [->] (t8) -- (m9);
\draw [->] (m1) -- (b2); \draw [->] (m3) -- (b2);
\draw [->] (m3) -- (b4); \draw [->] (m5) -- (b4);
\draw [->] (m5) -- (b6); \draw [->] (m7) -- (b6);
\draw [->] (m7) -- (b8); \draw [->] (m9) -- (b8);
\node at (0.9,2) {$k$};
\node at (5.6,4.5) {$k^+$};
\node at (17.6,4.5) {$k^7$};
\node at (19.4,2) {$k^8$};
\node (theM) at (0.5,4) [] {$\TheB k{8}$};
\node (flecheM) at (3.25,3.4) [] {};
\node (theA) at (0.5,0.0) [] {$\TheB{k^{+}}{6}$};
\node (flecheA) at (3.8,0) [] {};
\draw [->] (theM) to [bend right=-25] (flecheM);
\draw [->] (1.5,-0.5) to [bend right=25] (flecheA);
\draw[rouge,densely dashed,thick,rounded corners] (3.25,4) -- (4,4.75) -- (6,2.75) --  (8,4.75) -- (10,2.75) -- (12,4.75) -- (14,2.75) -- (16,4.75) -- (18.75,2) -- (18,1.25)
-- (16,3.25) -- (14,1.25) -- (12,3.25) -- (10,1.25) -- (8,3.25) --  (6,1.25) -- (4,3.25) -- (2,1.25) -- (1.25,2) -- (3.25,4) ;
\draw[bleu,densely dotted,thick,rounded corners] (5.25,2) -- (6,2.75) -- (6.75,2) -- (8,0.75) -- (10,2.75) -- (12,0.75) -- (14,2.75) -- (16.75,0) -- (16,-0.75) -- (14,1.25) -- (12,-0.75) -- (10,1.25) -- (8,-0.75) -- (6,1.25) -- (4,-0.75) -- (3.25,0) -- (5.25,2) ;
\end{tikzpicture}
\qquad
\begin{tikzpicture}[baseline={(current bounding box.center)},scale=1/3]
\node (t2) at (4,4) [] {};
\node (t4) at (8,4) [] {};
\node (t6) at (12,4) [] {};
\node (t8) at (16,4) [] {};
\node (m1) at (2,2) [] {};
\node (m3) at (6,2) [] {};
\node (m5) at (10,2) [] {};
\node (m7) at (14,2) [] {};
\node (m9) at (18,2) [] {};
\node (b2) at (4,0) [] {};
\node (b4) at (8,0) [] {};
\node (b6) at (12,0) [] {};
\node (b8) at (16,0) [] {};
\dodu{t2}\dodu{t4}\dodu{t6}\dodu{t8}
\dodu{m1}\dodu{m9}
\dodu{b2}\dodu{b4}\dodu{b6}\dodu{b8}
\dodu{5.8,2}\dodu{6.2,2}
\dodu{9.8,2}\dodu{10.2,2}
\dodu{13.8,2}\dodu{14.2,2}
\draw [->] (t2) -- (m1); \draw [->] (t2) -- (m3);
\draw [->] (t4) -- (m3); \draw [->] (t4) -- (m5);
\draw [->] (t6) -- (m5); \draw [->] (t6) -- (m7);
\draw [->] (t8) -- (m7); \draw [->] (t8) -- (m9);
\draw [->] (m1) -- (b2); \draw [->] (m3) -- (b2);
\draw [->] (m3) -- (b4); \draw [->] (m5) -- (b4);
\draw [->] (m5) -- (b6); \draw [->] (m7) -- (b6);
\draw [->] (m7) -- (b8); \draw [->] (m9) -- (b8);
\node at (0.9,2) {$k$};
\node at (5.6,4.5) {$k^+$};
\node at (17.6,4.5) {$k^7$};
\node at (19.4,2) {$k^8$};
\node (theM) at (1.2,4) [] {$\TheT{k^+}{6}$};
\node (flecheM) at (3.65,4) [] {};
\node (theA) at (0.6,0.0) [] {$\TheT{k}{8}$};
\node (flecheA) at (3.8,0) [] {};
\draw [->] (2.0,4.4) to [bend right=-25] (flecheM);
\draw [->] (1.5,-0.5) to [bend right=25] (flecheA);
\draw[rouge,densely dashed,thick,rounded corners] (4.75,4) -- (6,2.75) --  (8,4.75) -- (10,2.75) -- (12,4.75) -- (14,2.75) -- (16,4.75) -- (16.75,4) -- (14,1.25) -- (12,3.25) -- (10,1.25) -- (8,3.25) --  (6,1.25) -- (3.25,4) -- (4,4.75) -- (4.75,4) ;
\draw[bleu,densely dotted,thick,rounded corners] (2,1.25) -- (1.25,2) -- (2,2.75) --  (4,0.75) -- (6,2.75) -- (6.75,2) -- (8,0.75) -- (10,2.75) -- (12,0.75) -- (14,2.75) -- (16,0.75) --(18,2.75) -- (18.75,2) -- (16,-0.75) -- (14,1.25) -- (12,-0.75) -- (10,1.25) -- (8,-0.75) -- (6,1.25) -- (4,-0.75) -- (2,1.25) ;
\end{tikzpicture}
\end{equation}
The diagram on the left depicts the inclusion of $\TheB{k^+}{6}$ in its injective hull $\Inje{}[\TheB{k^+}{6}]$, the one on the right that of $\TheT k{8}$ in $\Inje{}[\TheT k{8}]$. These injective hulls are direct sums of the indecomposable injectives $\Inje{k'}$, with indices increasing in steps of $2$, each bringing (generically) four composition factors to the hull. Some are repeated, for example the composition factor $\Irre{k^{++}}$ has multiplicity $2$ in both $\Inje{}[\TheB{k^+}{6}]$ and $\Inje{}[\TheT k{8}]$, and we indicate this above by drawing two dots close together. The images of the inclusion maps are depicted by dotted lines.  Where it passes through a pair of ``double dots'', this image will contain a proper subspace, equivalent to one dot, of the subspace represented by these dots.  The cokernels of these inclusions are depicted by dashed lines with the same proviso regarding their passing through double dots. We note that the cokernel $\Inje{}[\TheB{k^+}{6}] / \TheB{k^+}{6}$ has the same composition factors as $\TheB k8$ and that $\Inje{}[\TheT k{8}] / \TheT k8$ has the same composition factors as $\TheT{k^+}{6}$.

\begin{proposition}\label{prop:cokerBT}
For $k\in\Lambda_0$ and $k^l\in\Lambda$, let $f_k^l$ and $g_k^l$ denote the natural inclusions $\TheB kl \ira \Inje{}[\TheB kl]$ and $\TheT kl \ira \Inje{}[\TheT kl]$, respectively. With a few exceptions, the cokernels of these inclusions are
\begin{subequations} \label{eq:cokernels}
\begin{equation} \label{eq:cokernels1}
\begin{alignedat}{4} 
\ &\Coker f_k^{2j}&\simeq &\ \TheB{k^{2\delta_L-1}}{2(j+1-\delta_L)-\delr{2j}}, &\qquad&
\Coker g_k^{2j}&\simeq &\ \TheT{k^+}{2(j-1)},\\
\ &\Coker f_k^{2j+1}&\simeq &\ \TheB{k^{2\delta_L-1}}{2(j-\delta_L)+1}, &\qquad&
\Coker g_k^{2j+1}&\simeq &\ \TheT{k^+}{2j+1-\delr{2j+1}},
\end{alignedat}
\end{equation}
where $\delta_L \equiv \delta_{k,k_L}$ and $\delr{l} \equiv \delta_{k^{l},k_R}$. The exceptions occur for $\Alg{} = \tl{}$ with $n$ even and $\beta = 0$, specifically
\begin{equation} \label{eq:cokernels2}
\Coker f_2^{2j}\simeq \TheT{2}{2j+1-\delr{2j}},\qquad\Coker f_2^{2j+1}\simeq \TheT{2}{2j}.
\end{equation}
\end{subequations}
\end{proposition}
\begin{proof}
We prove, by induction on $j$, the result for $\Coker f_k^{2j}\simeq \Inje{}[\TheB{k}{2j}] / \TheB{k}{2j}$, ignoring the exceptional cases at first. When $j=0$, the goal is to identify the cokernel in
\begin{equation}
\dses{\Irre k}{f_k^0}{\Inje k}{}{\Coker f_k^0}.
\end{equation}
If $k=k_L$, then $\Inje{k} \cong \Cost{k}$ and $\Coker f_k^0\cong\Irre{k^+}=\TheB{k^+}0$, by \eqref{es:ICI}. If $k_L<k<k_R$, then $\Inje{k} \cong \Proj{k}$ and we obtain $\Coker f_k^0 \cong \TheB{k^-}{2}$ from \eqref{es:DefAV} and \cref{lem:theAandV}.  If $k=k_R$, then $\Inje{k} \cong \Proj{k}$ again, but $\Irre{k} \cong \Cost{k}$, thus \eqref{es:CPC} gives $\Coker f_k^0 \cong \Cost{k^-} = \TheB{k^-}{1}$.  Note that if $k_L = k = k_R$, then $\TheB{k^+}{0} = \TheB{k^-}{1} = 0$ which is the correct cokernel ($f_k^0$ is an isomorphism in this case).  It is easy to check that the result in \eqref{eq:cokernels1} for $\Coker f_k^0$ unifies all of these cases.  The exceptional case $\Alg{} = \tl{}$ with $n$ even and $\beta = 0$ uses instead the exact sequence \eqref{eq:unbearable} to conclude that $\Coker f_2^0\simeq \Stan 2\simeq \TheT 2{1-\delr{0}}$, if $n>2$, and $\Irre 2\simeq \TheT 20$, if $n=2$ (see also the Loewy diagrams \eqref{ld:ProjDegenerate}).

Suppose now that $j\ge1$. The top row of the diagram
\begin{equation} \label{cd:cokerBT}
\begin{tikzpicture}[scale=1/3,baseline=(current bounding box.center)]
\node (tt2) at (10,8) [] {$0$};
\node (tt3) at (18,8) [] {$0$};
\node (tt4) at (26,8) [] {$0$};
\node (t1) at (5,5) [] {$0$};
\node (t2) at (10,5) [] {$\TheB{k}{2(j-1)}$};
\node (t3) at (18,5) [] {$\TheB{k}{2j}$};
\node (t4) at (26,5) [] {$\Stan{k^{2j-1}}$};
\node (t5) at (31,5) [] {$0$};
\node (m1) at (5,0) [] {$0$};
\node (m2) at (10,0) [] {$\Inje{}[\TheB{k}{2(j-1)}]$};
\node (m3) at (18,0) [] {$\Inje{}[\TheB{k}{2j}]$};
\node (m4) at (26,0) [] {$\Inje{k^{2j}}$};
\node (m5) at (31,0) [] {$0$};
\node (b1) at (5,-5) [] {$0$};
\node (b2) at (10,-5) [] {$\Coker f_k^{2(j-1)}$};
\node (b3) at (18,-5) [] {$\Coker f_k^{2j}$};
\node (b4) at (26,-5) [] {$\Stan{k^{2j}}$};
\node (b5) at (31,-5) [] {$0$};
\node (bb2) at (10,-8) [] {$0$};
\node (bb3) at (18,-8) [] {$0$};
\node (bb4) at (26,-8) [] {$0$};
\draw [->] (tt2) -- (t2);
\draw [->] (tt3) -- (t3);
\draw [->] (tt4) -- (t4);
\draw [->] (t1) -- (t2);
\draw [->] (t2) to node [above] {$a$} (t3);
\draw [->] (t3) to node [above] {$\alpha$} (t4);
\draw [->] (t4) -- (t5);
\draw [->] (m1) -- (m2);
\draw [dashed,->] (m2) to node [above] {$b$}  (m3);
\draw [dashed,->] (m3) to node [above] {$\beta$}  (m4);
\draw [->] (m4) -- (m5);
\draw [->] (t2) to node [left] {$f_k^{2(j-1)}$} (m2);
\draw [->] (t3) to node [left] {$f_k^{2j}$} (m3);
\draw [->] (t4) to node [left] {$\iota$} (m4);
\draw [->] (b1) -- (b2);
\draw [dashed,->] (b2) to node [above] {$c$} (b3);
\draw [dashed,->] (b3) to node [above] {$\gamma$} (b4);
\draw [->] (b4) -- (b5);
\draw [->] (m2) to node [left] {$h_k^{2(j-1)}$} (b2);
\draw [->] (m3) to node [left] {$h_k^{2j}$} (b3);
\draw [->] (m4) to node [left] {$\pi$} (b4);
\draw [->] (b2) -- (bb2);
\draw [->] (b3) -- (bb3);
\draw [->] (b4) -- (bb4);
\end{tikzpicture}
\end{equation}
is the exact sequence \eqref{es:DefBandT} defining $\TheB k{2j}$, the two leftmost columns describe the inclusions of the appropriate $\TheB{k}{l}$ into their injective hulls (the projections $h_k^l$ being chosen to make these columns exact), and the rightmost column is the exact sequence \eqref{eq:ProjExact} (note that $k^{2j} \neq k_L$).  Given $\coker f_k^{2(j-1)}$, the goal is to identify $\coker f_k^{2j}$.

Because $f_k^{2(j-1)}$ is injective, the injectivity of $\Inje{}[\TheB{k}{2j}]$ ensures that there exists $b$ making the top-left square of \eqref{cd:cokerBT} commute.  If $b$ were not injective, then it would annihilate some composition factor in $\soc \Inje{}[\TheB{k}{2(j-1)}] \cong \soc \TheB{k}{2(j-1)}$.  But then, $f_k^{2j} a$ would annihilate this factor in $\TheB{k}{2(j-1)}$, contradicting the injectivity of the latter morphism.  Thus, $b$ is injective.  Similarly, the injectivity of $f_k^{2j}$ and $\Inje{k^{2j}}$ ensures that there exists $\beta$ making the top-right square commute. If $\beta$ were not surjective, then it must annihilate the composition factor of $\soc \Inje{}[\TheB{k}{2j}]$ that is isomorphic to $\head \Inje{k^{2j}} \cong \Irre{k^{2j}}$.  As $\iota \alpha$ does not annihilate this factor in $\TheB{k}{2j}$, this is a contradiction, hence $\beta$ is surjective.  Finally, we conclude that the middle row of the diagram \eqref{cd:cokerBT} is exact by comparing composition factors.

The snake lemma now gives the exactness of the bottom row by defining morphisms $c$ and $\gamma$ that make the bottom two squares commute. If the bottom row of the commutative diagram \eqref{cd:cokerBT} does not split, then $\Coker f_k^{2j}$ will be a non-trivial extension of $\Stan{k^{2j}}$ by $\Coker f_k^{2(j-1)}$. By the induction hypothesis, this latter cokernel will be isomorphic to $\TheB{k^{2\delta_L-1}}{2(j-\delta_L)}$ (as $k^{2(j-1)} \neq k_R$), so the bottom row will give $\coker f_k^{2j} = \TheB{k^{2\delta_L-1}}{2(j+1-\delta_L)-\delr{2j}}$, by definition, the correction $-\delr{2j}$ being necessary when $k^{2j} = k_R$, hence $\Stan{k^{2j}} \cong \Irre{k^{2j}}$. It thus remains to prove that the bottom row of \eqref{cd:cokerBT} does not split.

Suppose then that the bottom row does split, so that there exists an injection $\gamma^{*} \colon \Stan{k^{2j}} \to \Coker f_k^{2j}$ such that $\gamma \gamma^{*}$ is the identity on $\Stan{k^{2j}}$. Moreover, this splitting means that
\begin{equation} \label{eq:HomJCoker}
	\Hom(\Inje{k^{2j}},\Coker f_k^{2j}) \simeq \Hom(\Inje{k^{2j}},\Coker f_k^{2(j-1)}) \oplus \Hom(\Inje{k^{2j}},\Stan{k^{2j}}) \simeq \mathbb{C},
\end{equation}
by \cref{prop:SIPHomo} and the induction hypothesis ($\Coker f_k^{2(j-1)} \cong \TheB{k^{2\delta_L-1}}{2(j-\delta_L)}$).  Since the modules in the middle row are all injective, this row splits and there exists an injection $\beta^{*}\colon\Inje{k^{2j}} \to \Inje{}[\TheB{k}{2j}]$ such that $\beta \beta^{*}$ is the identity on $\Inje{k^{2j}}$.  If $h_k^{2j} \beta^*$ were identically zero, then we would have $\Inje{k^{2j}} \cong \im \beta^* \subseteq \ker h_k^{2j} = \im f_k^{2j} \cong \TheB{k}{2j}$.  However, $\Inje{k^{2j}}$ has two composition factors isomorphic to $\Irre{k^{2j}}$ whereas $\TheB{k}{2j}$ has but one, a contradiction.  It follows that $h_k^{2j} \beta^* \neq 0$.

As $\gamma^* \pi \neq 0$, by the surjectivity of $\pi$, \eqref{eq:HomJCoker} shows that $\gamma^* \pi$ and $h_k^{2j} \beta^*$ are equal, up to some non-zero multiplicative constant. In particular, $h_k^{2j} \beta^{*} \iota = 0$, so $\im (\beta^* \iota) \subseteq \im f_k^{2j}$ and there exists a morphism $\alpha^{*}\colon\Stan{k^{2j-1}} \to \TheB{k}{2j}$ such that $f_k^{2j} \alpha^{*} = \beta^{*} \iota$. But now, $\iota = \beta \beta^* \iota = \beta f_k^{2j} \alpha^* = \iota \alpha \alpha^*$ and, since $\iota$ is injective, $\alpha \alpha^*$ is the identity on $\Stan{k^{2j-1}}$.  The top row therefore splits, contradicting the indecomposability of $\TheB k{2j}$. This contradiction shows that the bottom row is not split, completing the identification of $\coker f_k^{2j}$.

This identification of $\coker f_k^{2j}$ proceeds inductively until $k^{2j}$ approaches $k_R$.  It is easy to check that the above argument requires no significant changes if $k^{2j} = k_R$; however, changes are required if $k^{2j-1} = k_R$.  In the latter case, the recursive construction of \cref{sub:definitionBT} produces $\TheB{k}{2j-1}$ from $\TheB{k}{2(j-1)}$ and so the first row of \eqref{cd:cokerBT} has to be replaced by the defining exact sequence $\ses{\TheB{k}{2(j-1)}}{\TheB{k}{2j-1}}{\Irre{k^{2j-1}}}$.  Unfortunately, the first two injective hulls of the second row then become $\Inje{}[\TheB{k}{2(j-1)}]$ and $\Inje{}[\TheB{k}{2j-1}]$, which are isomorphic by \cref{cor:BTProjInj}. We can still deduce that there exists an injection, hence an isomorphism, $b$ making the top left square commute.  However, to make the second row of \eqref{cd:cokerBT} exact, we must replace $\Inje{k^{2j}}$ by $0$.  Applying the snake lemma to these replaced rows yields the morphism $c$ that makes the bottom left square commute and the short exact sequence $\ses{\Irre{k^{2j-1}}}{\coker f_k^{2(j-1)}}{\coker f_k^{2j-1}}$. But, induction identifies $\coker f_k^{2(j-1)}$ as $\TheB{k^{2 \delta_L -1}}{2(j-\delta_L)}$.  As this module has a unique submodule isomorphic to $\Irre{k^{2j-1}}$, comparing this exact sequence with the second of \cref{lem:tiesBetweenFamilies} yields the desired conclusion: $\coker f_k^{2j-1}\simeq \TheB{k^{2\delta_L-1}}{2(j-\delta_L)-1}$.

Similar arguments and duality identify the other cokernels in \eqref{eq:cokernels}. We remark that the proof for $\coker f_k^{2j+1}$ is somewhat easier because the restriction put on $k$ in the definition of the remaining $\TheB{k}{2j+1}$ avoids the technicalities that would arise should $k^{2j+1}$ approach $k_R$.
\end{proof}

These cokernels give injective presentations for the $\TheB{k}{l}$, for example $\ses{\TheB{k}{2j}}{\Inje{}[\TheB{k}{2j}]}{\TheB{k^{2\delta_L-1}}{2(j+1-\delta_L)-\delr{2j}}}$.  Analogous projective presentations now follow by taking duals.
\begin{proposition}\label{prop:kerBT}
For $k\in\Lambda_0$ and $k^l\in\Lambda$, let $p_k^l$ and $q_k^l$ denote the natural projections $\Proj{}[\TheB kl] \sra \TheB kl$ and $\Proj{}[\TheT kl] \sra \TheT kl$, respectively. With a few exceptions, the kernels of these projections are
\begin{subequations} \label{eq:kernels}
\begin{equation}
\begin{alignedat}{4}
\ &\ker p_k^{2j}&\simeq &\ \TheB{k^+}{2(j-1)}, &\qquad&
\ker q_k^{2j}&\simeq &\ \TheT{k^{2\delta_L-1}}{2(j+1-\delta_L)-\delr{2j}},\\
\ &\ker p_k^{2j+1}&\simeq &\ \TheB{k^+}{2j+1-\delr{2j+1}}, &\qquad&
\ker q_k^{2j+1}&\simeq &\ \TheT{k^{2\delta_L-1}}{2(j-\delta_L)+1},
\end{alignedat}
\end{equation}
in the notation of \cref{prop:cokerBT}. The exceptions occur for $\Alg{} = \tl{}$ with $n$ even and $\beta = 0$, specifically
\begin{equation}
\ker q_2^{2j}\simeq \TheB{2}{2j+1-\delr{2j}},\qquad\ker q_2^{2j+1}\simeq \TheB{2}{2j}.
\end{equation}
\end{subequations}
\end{proposition}

%
%
\subsection{Their extension groups with irreducible modules}\label{sub:projectiveBT}

\cref{lem:trivialExtensions}\ref{it:ExtByIrr} has shown that computing $\Ext(\Mod{M},\Irre{k})$ and $\Ext(\Irre{k},\Mod{M})$ is sufficient to see when a module $\Mod{M}$ can appear in any non-trivial extension. We shall therefore limit ourselves to these extension groups. With the presentations derived in the previous section, it is easy to compute them.
\begin{proposition}\label{prop:extensionsBT}
Let $k, k'\in\Lambda_0$. The extension groups of the $\TheB kl$, for $l\ge 2$, with an irreducible are given by
\begin{equation} \label{eq:ExtIB}
\begin{aligned}
&\textup{(a)} & \Ext(\Irre{k'},\TheB k{2j}) &\simeq \bigoplus_{i=0}^{j+1} \delta_{k',k^{2i-1}} \CC, \\
&\textup{(b)} & \Ext(\TheB k{2j},\Irre{k'}) &\simeq \bigoplus_{i=1}^{j-1} \delta_{k',k^{2i}} \CC,
\end{aligned}
\qquad \qquad
\begin{aligned}
&\textup{(c)} & \Ext(\Irre{k'},\TheB k{2j+1}) & \simeq \bigoplus_{i=0}^{j} \delta_{k',k^{2i-1}} \CC, \\
&\textup{(d)} & \Ext(\TheB k{2j+1},\Irre{k'}) & \simeq \bigoplus_{i=1}^{j+1} \delta_{k',k^{2i}} \CC.
\end{aligned}
\end{equation}
Those of the $\TheT kl$ with an irreducible are then given by $\Ext(\Irre{k'},\TheT kl) \simeq \Ext(\TheB kl,\Irre{k'})$ and $\Ext(\TheT kl,\Irre{k'})\simeq \Ext(\Irre{k'},\TheB kl)$.  The extension groups for $l=0$ and $1$ were given in \cref{prop:SIExte}.
\end{proposition}
\begin{proof}
The relationship between the extension groups involving the $\TheB{k}{l}$ and the $\TheT{k}{l}$ is just duality.  The computations for each of the four cases above are very similar, so we shall only present the details for (a), ignoring the exceptional case where $\Alg{} = \tl{}$ with $n$ even, $\beta = 0$ and $k=2$ (the results for these exceptional cases are also as given in \eqref{eq:ExtIB} above).

Recall that \cref{prop:cokerBT} gives an injective presentation of $\TheB{k}{2j}$:
\begin{equation}
\dses{\TheB k{2j}}{}{\Inje{}[\TheB{k}{2j}]}{}{\TheB{k^{2\delta_L-1}}{2(j+1-\delta_L)-\delr{2j}}}.
\end{equation}
This gives the long exact sequence
\begin{equation} \label{es:ExtIB}
0\lra \Hom(\Irre{k'},\TheB k{2j})
\lra \Hom(\Irre{k'},\Inje{}[\TheB{k}{2j}])
\lra \Hom(\Irre{k'},\TheB{k^{2\delta_L-1}}{2(j+1-\delta_L)-\delr{2j}})
\lra \Ext(\Irre{k'},\TheB k{2j})
\lra 0,
\end{equation}
where we recall that $\Ext(\Irre{k'},\Inje{}[\TheB{k}{2j}])=0$, by injectivity. Now, $\Hom(\Irre{k'}, \Mod{M}) \cong \Hom(\Irre{k'}, \soc \Mod{M})$, for any module $\Mod{M}$, and $\soc \TheB k{2j} \cong \soc \Inje{}[\TheB{k}{2j}]$. The two leftmost Hom-groups of \eqref{es:ExtIB} are therefore isomorphic, hence we have
\begin{equation}
\Ext(\Irre{k'}, \TheB k{2j})
\simeq \Hom(\Irre{k'}, \TheB{k^{2\delta_L-1}}{2(j+1-\delta_L)-\delr{2j}})
\simeq \Hom(\Irre{k'}, \soc \TheB{k^{2\delta_L-1}}{2(j+1-\delta_L)-\delr{2j}})
\simeq \bigoplus_{i=0}^{j+1-\delta_L-\delr{2j}} \delta_{k',k^{2i-1+2\delta_L}} \CC,
\end{equation}
by \cref{prop:SIPHomo,prop:socRadHdBT}.  It is easy to check that the condition $k' \in \Lambda_0$ allows this to be simplified to the statement of (a).  The proof of (c) follows the same argument whilst (b) and (d) instead use projective presentations.
\end{proof}

%
%
\subsection{Identification of non-trivial extensions and completeness} \label{sub:completeness}

\cref{prop:extensionsBT} shows that the modules $\TheB{k}{l}$ and $\TheT{k}{l}$ have non-trivial extensions with irreducibles. The next proposition shows that, in fact, every one of these non-trivial extensions is isomorphic to one (or a direct sum) of the modules that have already been introduced.  This will then imply that we have identified a complete list of indecomposable $\Alg{}$-modules.

\begin{proposition}\label{prop:extensionRevealed}
Let $k,k'\in\Lambda_0$ and $k^{2j}\in\Lambda$. The following short exact sequences identify a representative, unique up to isomorphism, of the non-trivial extensions of the $\TheB{k}{l}$ and $\Irre{k'}$ defined by \textup{(a)}, \textup{(b)}, \textup{(c)} and \textup{(d)} of \cref{prop:extensionsBT}:
\begin{equation} \label{eq:AllBIExts}
\begin{alignedat}{2}
\textup{(a1)}\quad\ & \dses{\TheB k{2j}}{}{\TheB k{2i-1}\oplus\TheT {k'}{2(j-i)+1}}{}{\Irre{k'}},&\qquad & \text{for \(k'=k^{2i-1}\), \(1\le i\le j\),}\\
\textup{(a2)}\quad\ & \dses{\TheB k{2j}}{}{\TheT{k'}{2j+1}}{}{\Irre{k'}}, &\qquad & \text{for \(k'=k^-\),}\\
\textup{(a3)}\quad\ & \dses{\TheB k{2j}}{}{\TheB{k}{2j+1}}{}{\Irre{k'}}, &\qquad & \text{for \(k'=k^{2j+1}\),}\\
\textup{(b)}\quad\ & \dses{\Irre{k'}}{}{\TheB k{2i}\oplus \TheB{k'}{2(j-i)}}{}{\TheB k{2j}},&\qquad & \text{for \(k'=k^{2i}\), \(1\le i\le j-1\),}\\
\textup{(c1)}\quad\ & \dses{\TheB k{2j+1}}{}{\TheB k{2i-1}\oplus\TheT {k'}{2(j-i+1)}}{}{\Irre{k'}},&\qquad & \text{for \(k'=k^{2i-1}\), \(1\le i\le j\),}\\
\textup{(c2)}\quad\ & \dses{\TheB k{2j+1}}{}{\TheT{k'}{2(j+1)}}{}{\Irre{k'}}, &\qquad & \text{for \(k'=k^-\),}\\
\textup{(d1)}\quad\ & \dses{\Irre{k'}}{}{\TheB k{2i}\oplus \TheB{k'}{2(j-i)+1}}{}{\TheB k{2j+1}},&\qquad & \text{for \(k'=k^{2i}\), \(1\le i\le j\),}\\
\textup{(d2)}\quad\ & \dses{\Irre{k'}}{}{\TheB k{2(j+1)}}{}{\TheB k{2j+1}},&\qquad & \text{for \(k'=k^{2(j+1)}\).}
\end{alignedat}
\end{equation}
Dualising these sequences gives representatives for the non-trivial extensions of the $\TheT{k}{l}$ and $\Irre{k'}$.
\end{proposition}
\noindent We remark that we have already established (d2) as the second exact sequence of \eqref{eq:BevenWithCostAndIrre}.

Before turning to the proof, we exemplify one of these exact sequences with Loewy diagrams. Consider (c1), with $i=2$, $j=3$ and $k'=k^3$.  It takes the form $\ses{\TheB k7}{\TheB k3\oplus\TheT{k^3}4}{\Irre{k^3}}$ and may be depicted thus:
\begin{equation} \label{es:(c1)}
0\ \longrightarrow\
\begin{tikzpicture}[baseline={(t.center)},scale=0.3]
\node (m3) at (6,2) [] {};
\node (m5) at (10,2) [] {};
\node (m7) at (14,2) [] {};
\node (m9) [label=above:$k^7$] at (18,2) [] {};
\node (b2) [label=below:$k$] at (4,0) [] {};
\node (b4) at (8,0) [] {};
\node (b6) at (12,0) [] {};
\node (b8) at (16,0) [] {};
\dodu{m3}\dodu{m5}\dodu{m7}\dodu{m9}
\dodu{b2}\dodu{b4}\dodu{b6}\dodu{b8}
\draw [->] (m3) -- (b2);
\draw [->] (m3) -- (b4); \draw [->] (m5) -- (b4);
\draw [->] (m5) -- (b6); \draw [->] (m7) -- (b6);
\draw [->] (m7) -- (b8); \draw [->] (m9) -- node (t) {} (b8);
\end{tikzpicture}
\ \longrightarrow\
\begin{tikzpicture}[baseline={(t.center)},scale=0.3]
\node (m3) at (6,2) [] {};
\node (m5) [label=above:$k^3$] at (10,2) [] {};
\node (b2) [label=below:$k$] at (4,0) [] {};
\node (b4) at (8,0) [] {};
\dodu{m3}\dodu{m5}
\dodu{b2}\dodu{b4}
\draw [->] (m3) -- (b2);
\draw [->] (m3) -- (b4); \draw [->] (m5) -- node (t) {} (b4);
\end{tikzpicture}
\ \oplus\
\begin{tikzpicture}[baseline={(t.center)},scale=0.3]
\node (m1) [label=above:$k^3$] at (2,2) [] {};
\node (m3) at (6,2) [] {};
\node (m5) [label=above:$k^7$] at (10,2) [] {};
\node (b2) at (4,0) [] {};
\node (b4) at (8,0) [] {};
\dodu{m1}\dodu{m3}\dodu{m5}
\dodu{b2}\dodu{b4}
\draw [->] (m1) -- (b2);
\draw [->] (m3) -- (b2);
\draw [->] (m3) -- (b4); \draw [->] (m5) -- node (t) {} (b4);
\end{tikzpicture}
\ \longrightarrow \
\begin{tikzpicture}[baseline={(m1.base)},scale=0.3]
\node (m1) [label=right:$k^3$] at (2,1) [] {\vphantom{$k^3$}};
\dodu{m1}
\end{tikzpicture}
\ \longrightarrow\ 0.
\end{equation}
The composition factors are depicted by dots, as usual, and only the extreme ones are labelled. It is easy to see that morphisms from $\TheB k7$ to $\TheB k3$ and $\TheT{k^3}{4}$ exist (the latter are quotients of the former).  If we can show that the direct sum ${\TheB k7}\to{\TheB k3\oplus\TheT{k^3}4}$ of these morphisms is injective, then its cokernel must be isomorphic to the irreducible $\Irre{k^3}$, by counting composition factors. There cannot be a non-zero morphism $\Irre{k^3}\to {\TheB k3\oplus\TheT{k^3}4}$ since $\soc(\TheB k3\oplus\TheT{k^3}4)$ has no submodule isomorphic to $\Irre{k^3}$. The sequence \eqref{es:(c1)} is therefore non-split and ${\TheB k3\oplus\TheT{k^3}4}$ represents the (isomorphism class of the) non-trivial extensions of $\Ext(\Irre{k^3},\TheB k7)\cong\mathbb C$, even though it is reducible.

\begin{proof}[Proof of \cref{prop:extensionRevealed}]
Each of these sequences are established using similar arguments, though the proofs differ slightly depending on whether the irreducible $\Irre{k'}$ is a submodule or a quotient. We will illustrate the method of  proof by detailing the arguments for (b), where the irreducible is a submodule, and only comment on the changes required when the irreducible is a quotient.

To prove (b), take $1 \le i \le j-1$ and consider the diagram
\begin{equation} \label{cd:extensionB1}
\begin{tikzpicture}[scale=1/3,baseline=(current bounding box.center)]
\node (t1) at (5,5) [] {$0$};
\node (t2) at (10,5) [] {$\TheB{k^+}{2(j-1)}$};
\node (t3) at (17,5) [] {$\Proj{}[\TheB{k}{2j}]$};
\node (t4) at (24,5) [] {$\TheB k{2j}$};
\node (t5) at (29,5) [] {$0$\phantom{,}};
\node (b1) at (5,0) [] {$0$};
\node (b2) at (10,0) [] {$\Irre{k^{2i}}$};
\node (b3) at (17,0) [] {$\TheB{k}{2i}\oplus\TheB{k^{2i}}{2(j-i)}$};
\node (b4) at (24,0) [] {$\Coker \iota$};
\node (b5) at (29,0) [] {$0$,};
\draw [->] (t1) -- (t2);
\draw [->] (t2) to node [above] {$r$} (t3);
\draw [->] (t3) to node [above] {$p$} (t4);
\draw [->] (t4) -- (t5);
\draw [->] (b1) -- (b2);
\draw [dashed,->] (b2) to node [above] {$\iota$}  (b3);
\draw [dashed,->] (b3) to node [above] {$\pi$}  (b4);
\draw [->] (b4) -- (b5);
\draw [->] (t2) to node [left] {$\alpha$} (b2);
\draw [->] (t3) to node [left] {$\beta$} (b3);
\draw [dashed,->] (t4) to node [left] {$\gamma$} (b4);
\end{tikzpicture}
\end{equation}
in which the top row is the projective presentation of $\TheB{k}{2j}$, given in \cref{prop:kerBT}, and is therefore exact.  As $\Irre{k^{2i}}$ is a composition factor of $\head \TheB{k^+}{2(j-1)}$, we may choose $\alpha$ to be non-zero, hence surjective.  Since $\TheB{k}{2i}\oplus\TheB{k^{2i}}{2(j-i)}$ and $\TheB{k}{2j}$ have isomorphic heads, their projective covers are isomorphic.  We may therefore choose $\beta$ to also be surjective in \eqref{cd:extensionB1}.  Indeed, we may choose the kernel so that
\begin{equation}
\TheB{k^+}{2(i-1)} \oplus \TheB{k^{2i+1}}{k^{2(j-i-1)}} \cong \ker \beta \subset \ker p \cong \TheB{k^+}{2(j-1)}
\end{equation}
and the kernels only differ in that $\ker p$ has one composition factor isomorphic to $\Irre{k^{2i}}$ while $\ker \beta$ has none.

As $\soc \TheB{k}{2i}\oplus\TheB{k^{2i}}{2(j-i)}$ has two composition factors isomorphic to $\Irre{k^{2i}}$, the inclusion $\iota$ belongs to a two-dimensional Hom-group.  We will choose $\iota$ so that the left square of \eqref{cd:extensionB1} commutes.  To see that this is possible, note that $\ker \beta \subset \ker p = \im r$, so that $\ker \beta r = \im r \cap \ker \beta = \ker \beta \cong \TheB{k^+}{2(i-1)} \oplus \TheB{k^{2i+1}}{k^{2(j-i-1)}}$, hence
\begin{equation}
\im \beta r \cong \frac{\TheB{k^+}{2(j-1)}}{\TheB{k^+}{2(i-1)} \oplus \TheB{k^{2i+1}}{k^{2(j-i-1)}}} \cong \Irre{k^{2i}}.
\end{equation}
It follows that there exists $\iota$ making the left square commute.  $\pi$ can now be chosen to make the bottom row of \eqref{cd:extensionB1} exact.  Since $\pi \beta r = \pi \iota \alpha = 0$, there is a map $\gamma$ that makes the right square of \eqref{cd:extensionB1} commute. It is surjective, as both $\beta$ and $\pi$ are, hence it is an isomorphism because the composition factors of $\TheB k{2j}$ and $\Coker \iota$ coincide.  The bottom row is thus the required short exact sequence (b).

When the irreducible $\Irre{k'}$ appearing in a sequence in \eqref{eq:AllBIExts} is a quotient, instead of a submodule, there is a minor change to the method of proof.  The diagram \eqref{cd:extensionB1} is replaced by one in which the bottom row is an injective presentation of the module that is extending $\Irre{k'}$ and the top row is the sequence whose exactness is to be established. For example the proof of the dual of (d2) would use the following diagram.
\begin{equation}\label{cd:extensionD2}
\begin{tikzpicture}[scale=1/3,baseline=(current bounding box.center)]
\node (t1) at (5,5) [] {$0$};
\node (t2) at (10,5) [] {$\ker\pi$};
\node (t3) at (17,5) [] {$\TheT k{2(j+1)}$};
\node (t4) at (24,5) [] {$\Irre{k^{2(j+1)}}$};
\node (t5) at (29,5) [] {$0$\phantom{.}};
\node (b1) at (5,0) [] {$0$};
\node (b2) at (10,0) [] {$\TheT k{2j+1}$};
\node (b3) at (17,0) [] {$\Inje{}[\TheT{k}{2j+1}]$};
\node (b4) at (24,0) [] {$\TheT{k^+}{2j+1}$};
\node (b5) at (29,0) [] {$0$.};
\draw [->] (t1) -- (t2);
\draw [dashed,->] (t2) to node [above] {$\iota$} (t3);
\draw [dashed,->] (t3) to node [above] {$\pi$} (t4);
\draw [->] (t4) -- (t5);
\draw [->] (b1) -- (b2);
\draw [->] (b2) to node [above] {$f$}  (b3);
\draw [->] (b3) to node [above] {$h$}  (b4);
\draw [->] (b4) -- (b5);
\draw [dashed,->] (t2) to node [left] {$\gamma$} (b2);
\draw [->] (t3) to node [left] {$\beta$} (b3);
\draw [->] (t4) to node [left] {$\alpha$} (b4);
\end{tikzpicture}
\end{equation}
Otherwise, one proves the commutativity of the diagram as before.
\end{proof}

\begin{theorem}\label{thm:classification}
Let $\Alg{}$ be $\tl{n}$ or $\dtl{n}$. Then, any finite-dimensional indecomposable module over $\Alg{}$ is isomorphic to one of the following:
\begin{enumerate}
\item $\Stan k$, for $k$ a critical integer in $\Lambda$; \label{it:IndecClassCrit}
\item $\Proj k$, for $k\in\Lambda$ non-critical and larger than $k_L$ in the orbit $[k]$; \label{it:IndecClassProj}
\item $\TheB k{l}$ or $\TheT k{l'}$, for $k$ non-critical in $\Lambda_0$ and $l\ge0$ and $l'>0$ such that $k^{l},\ k^{l'}\in\Lambda$. \label{it:IndecClassOther}
\end{enumerate}
These indecomposables are distinct in that there are no isomorphisms among different elements of the above list.
\end{theorem}
\begin{proof}
The modules appearing in the above list have already been shown to be indecomposable and pairwise non-isomorphic.  To show that the list given is complete, note that any finite-dimensional indecomposable module $\Mod{M}$ may be constructed iteratively by adding one composition factor at a time. Indeed, one could start from the head of $\Mod{M}$, which is semisimple, and then add the composition factors of the head of its radical, and then add those of the head of the radical of the radical, and so on. Therefore, every finite-dimensional indecomposable $\Alg{}$-module may be constructed by adding, one at time, irreducible modules to a direct sum of modules in the above list. However, \cref{prop:SIExte,prop:extensionRevealed} show that any such extension of these modules is already a direct sum of modules in the list. This list therefore constitutes a complete set of finite-dimensional indecomposable $\Alg{}$-modules, up to isomorphism.

The above list is also complete for the exceptional case $\Alg{} = \tl{}$ with $n$ even and $\beta=0$. Two remarks are useful to reach this conclusion in this case. First, the sublist \ref{it:IndecClassCrit} is empty. Second, $k_L$ is omitted from \ref{it:IndecClassProj} not to avoid coincidence with $\TheB{k_L}1$, as in the generic case, but because $k_L$ is then $0$ and does not belong to $\Lambda_0$ ($\Proj{0}$ is not defined). The leftmost irreducible module of the orbit is then $\Irre{2}$ and its projective cover $\Proj{2}$ is distinct from $\TheB{2}1$.
\end{proof}

%
%

\section{The restriction and induction of the modules $\Cost{n,k}$, $\TheB{n,k}l$ and $\TheT{n,k}l$} \label{sec:restrictionInduction}

The action of the induction and restriction functors on the standard, irreducible and projective modules was obtained (or recalled) in \cref{sub:restriction}. This section extends those calculations to the remaining classes of indecomposable modules, namely the costandards $\Cost{k}$ and the $\TheB{k}{l}$ and $\TheT{k}{l}$, with $l>1$, thus completing the description of these functors on all indecomposable $\Alg{}$-modules. In this section, the algebra label $n$ will be made explicit, so we shall write, for example, $\TheB{n,k}{j}$ instead of $\TheB{k}{j}$. We also write $\Lambda_{n,0}$ for the set $\Lambda_{0}$ corresponding to $\Alg{n}$.

\subsection{Restriction} \label{sec:CBTRes}

We begin with the restriction of the $\Cost{n+1,k}$.  This follows immediately from \cref{prop:StanRest} and the fact that restriction commutes with duality.
\begin{proposition} \label{prop:CostRest}
For all non-critical $k \in \Lambda_{n+1,0}$, the restriction of the costandard modules is given by
\begin{equation}
\Res{\Cost{n+1,k}} \cong
\begin{cases*}
\Cost{n,k-1}\oplus\Cost{n,k+1}, & if \(\Alg{n} = \tl{n}\), \\
\Cost{n,k-1}\oplus\Cost{n,k}\oplus\Cost{n,k+1}, & if \(\Alg{n} = \dtl{n}\).
\end{cases*}
\end{equation}
For critical $k$, $\Cost{n+1,k} \cong \Stan{n+1,k}$ and the result was given in \cref{prop:StanRest}\ref{it:SResCrit}.
\end{proposition}
Identifying the restrictions of the $\TheB{n+1,k}{l}$ requires considerably more work.
\begin{proposition}\label{prop:res_bmod}
For $k$ non-critical and $k, k^{l} \in \Lambda_{n+1,0}$, the restriction of $\TheB{n+1,k}{l}$ is given by
\begin{align} \label{eq:BRes}
\Res{\TheB{n+1,k}{l}} \simeq &
\begin{dlrcases*}
\bigoplus_{j=0}^{\lfloor l/2 \rfloor} \Proj{n,k^{2j}-1}, & if $k-1$ is critical, \\
\TheB{n,k-1}{l}, & otherwise
\end{dlrcases*}
\oplus
\begin{lrcases*}
\TheB{n,k}{l}, & if \(\Alg{} = \dtl{}\), \\
0, & otherwise
\end{lrcases*}
\notag\\
&\oplus
\begin{dlrcases*}
\bigoplus_{j=0}^{\lfloor (l-1)/2 \rfloor} \Proj{n,k^{2j}+1}, & if $k+1$ is critical,\\
\TheB{n,k+1}{l}, & otherwise
\end{dlrcases*}
,
\end{align}
where it is understood that each summand of the form $\TheB{n,\kappa}{l}$ should be replaced by $\TheB{n,\kappa}{l-1}$ whenever ${\kappa}^{l} \notin \Lambda_{n,0}$.
\end{proposition}
\noindent We remark that the projectives appearing in \eqref{eq:BRes} are all critical, hence irreducible.
\begin{proof}
The proofs for $l$ even and $l$ odd are (slightly) different; we first detail that for $l$ even and then explain how the arguments may be changed for $l$ odd.

The proof for $l$ even proceeds by induction on $l$. The case $l=0$ is covered by \cref{prop:IrreRest}, so assume that $l\ge 2$. Since restriction is an exact covariant functor, \eqref{eq:exactBk2j} gives the exact sequence
	\begin{equation}
		\dses{\Res{\TheB{n+1,k}{l-2}}}{}{\Res{\TheB{n+1,k}{l}}}{}{\Res{\Stan{n+1,k^{l-1}}}}.
	\end{equation}
As in the proof of \cref{prop:ProjRest}, this sequence can be decomposed into two or three exact sequences, according as to whether $\Alg{} = \tl{}$ or $\dtl{}$, respectively, by selecting a parity for the direct summands of the modules (for $\dtl{}$) and distinguishing their $F_{n}$-eigenvalues. These sequences are analysed using similar arguments, so we shall only focus on one of them, namely
	\begin{equation}\label{ses:resbeven2}
		\dses{(\Res{\TheB{n+1,k}{l-2}})_{+1}}{}{(\Res{\TheB{n+1,k}{l}})_{+1}}{}{\Stan{n,k^{l-1}-1}},
	\end{equation}
where $(\Mod{M})_{i}$, for $i\in\set{-1,0,+1}$, is the direct summand of $\Mod{M}$ on which $F_{n}$ has the same eigenvalue as on $\Stan{n,k+i}$ (see the proof of \cref{prop:ProjRest} where this notation was first introduced). This choice of $F_{n}$-eigenvalue will lead to the third direct summand (enclosed in braces) of \eqref{prop:res_bmod}.

The diagram below illustrates the argument. It assumes that $\ell=4$ and describes the restriction of the $\tl{n+1}$-module $\TheB{n+1,2}4$, hence $k=2$ and $l=4$. The indices of the composition factors of this module are typeset in bold and appear in the bottom line; those of $\Res{\TheB{n+1,2}4}$ appear in the top line.
\begin{equation*}
\begin{tikzpicture}[baseline={(current bounding box.center)},every node/.style={fill=white,circle,inner sep=2pt},scale=1/3]
	\draw[dashed]  (6,-1.5) -- (6,1.0);
	\draw[dashed] (14,-1.5) -- (14,1.);
	\draw[dashed] (22,-1.5) -- (22,1.);
	\draw[dashed] (30,-1.5) -- (30,1.);
	\draw[dashed] (38,-1.5) -- (38,1.);
	\draw[dashed]  (6,1.8) -- (6,2.8);
	\draw[dashed] (14,1.8) -- (14,2.8);
	\draw[dashed] (22,1.8) -- (22,2.8);
	\draw[dashed] (30,1.8) -- (30,2.8);
	\draw[dashed] (38,1.8) -- (38,2.8);
	\node at (0,0) {$0$};
	\node at (2,1) {$1$};
	\node at (4,0) {$\mathbf2$};    \node at (4,-1.75) {$k$};
	\node at (6,1) {$3$};
	\node at (8,0) {$\mathbf4$};    \node at (8,-1.75) {$k^+$};
	\node at (10,1) {$5$};
	\node at (12,0) {$6$};
	\node at (14,1) {$7$};
	\node at (16,0) {$8$};
	\node at (18,1) {$9$};
	\node at (20,0) {$\mathbf{10}$};  \node at (20,-1.75) {$k^{l-2}$};
	\node at (22,1) {$11$};
	\node at (24,0) {$\mathbf{12}$};  \node at (24,-1.75) {$k^{l-1}$};
	\node at (26,1) {$13$};
	\node at (28,0) {$14$};
	\node at (30,1) {$15$};
	\node at (32,0) {$16$};
	\node at (34,1) {$17$};
	\node at (36,0) {$\mathbf{18}$};  \node at (36,-1.75) {$k^{l}$};
	\node at (38,1) {$19$};
	\node at (40,0) {$20$};
	\node at (42,1) {$\dots$};
	\node at (44,0) {$\dots$};
\end{tikzpicture}
\end{equation*}

We first discuss the subcase in which $k+1$ is critical (as in the diagram). Then so is $k^{l-1}-1 = k^{l-2}+1$, hence the standard module $\Stan{n,k^{l-1}-1}$ is projective and \eqref{ses:resbeven2} must split. The induction hypothesis therefore gives
\begin{equation}
(\Res{\TheB{n+1,k}{l}})_{+1} \simeq \Stan{n,k^{l-1}-1} \oplus (\Res{\TheB{n+1,k}{l-2}})_{+1}
\simeq \Proj{n,k^{l-2}+1} \oplus \bigoplus_{j=0}^{l/2-2} \Proj{n,k^{2j}+1} = \bigoplus_{j=0}^{l/2-1} \Proj{n,k^{2j}+1},
\end{equation}
as in \eqref{prop:res_bmod}.  If $k+1$ is not critical, then the sequence \eqref{ses:resbeven2} cannot split because
\begin{equation}
	\Hom({\Stan{n,k^{l-1}-1}},\Res{\TheB{n+1,k}{l}})
	\overset{(1)}{\simeq}  \Hom(\Ind{\Stan{n,k^{l-1}-1}},\TheB{n+1,k}{l})
	\overset{(2)}{\simeq}  \Hom(\Stan{n+1,k^{l-1}},\TheB{n+1,k}{l})
	\overset{(3)}{=} 0.
\end{equation}
Here, the isomorphism $(1)$ is Frobenius reciprocity, $(2)$ follows from \cref{prop:StanRest} and the eigenvalues of $F_{n}$, and (3) amounts to noting that any such non-zero morphism would have to be injective, contradicting the fact that $\TheB{n+1,k}{l}$ has no submodule isomorphic to $\Stan{n+1,k^{l-1}}$.  (If it did, then the inclusion would split the defining exact sequence \eqref{eq:exactBk2j} and $\TheB{n,k}{l}$ would be decomposable.)  The induction hypothesis gives $(\Res{\TheB{n+1,k}{l-2}})_{+1} \cong \TheB{n,k+1}{l-2}$, so the non-split exact sequence \eqref{ses:resbeven2} is then that defining  $\TheB{n,k+1}{l}$, whence we conclude that $(\Res{\TheB{n+1,k}{l}})_{+1} \simeq \TheB{n,k+1}{l}$. Note that this assumes that $(k+1)^{l} \le n$, for otherwise \eqref{ses:resbeven2} gives instead$(\Res{\TheB{n+1,k}{l}})_{+1} \simeq \TheB{n,k+1}{l-1}$, by \cref{prop:extensionRevealed}.  This completes the identification of $(\Res{\TheB{n+1,k}{l}})_{+1}$ and similar arguments identify $(\Res{\TheB{n+1,k}{l}})_{0}$ and $(\Res{\TheB{n+1,k}{l}})_{-1}$, completing the proof for $l$ even.

When $l$ is odd, the induction instead starts at $l=1$, for which the statement is obtained from \cref{prop:StanRest} by noting that
\begin{equation}
\Res{\Cost{n+1,k}} \simeq \Res{\left(\twdu{\Stan{n+1,k}}\right)} \simeq \twdu{\left(\Res{\Stan{n+1,k}}\right)},
\end{equation}
since restriction and duality commute. In the inductive step, the sequence $\ses{\TheB{n+1,k^{2}}{l-2}}{\TheB{n+1,k}{l}}{\Cost{n+1,k}}$ is used instead of \eqref{eq:exactBk2j} (see \cref{prop:3point1}) and the rest of the arguments proceed as before.
\end{proof}
The restrictions of the $\TheT{n+1,k}{l}$ are now obtained by duality.
\begin{proposition}\label{prop:res_tmod}
For $k$ non-critical and $k, k^{l} \in \Lambda_{n+1,0}$, the restriction of $\TheT{n+1,k}{l}$ is given by
\begin{align}
\Res{\TheT{n+1,k}{l}} \simeq &
\begin{dlrcases*}
\bigoplus_{j=0}^{\lfloor l/2 \rfloor } \Proj{n,k^{2j}-1}, & if $k-1$ is critical, \\
\TheT{n,k-1}{l}, & otherwise
\end{dlrcases*}
\oplus
\begin{lrcases*}
\TheT{n,k}{l}, & if \(\Alg{} = \dtl{}\), \\
0, & otherwise
\end{lrcases*}
\notag\\
&\oplus
\begin{dlrcases*}
\bigoplus_{j=0}^{\lfloor (l-1)/2 \rfloor} \Proj{n,k^{2j}+1}, & if $k+1$ is critical, \\
\TheT{n,k+1}{l}, & otherwise
\end{dlrcases*}
,
\end{align}
where it is understood that each summand of the form $\TheT{n,\kappa}{l}$ should be replaced by $\TheT{n,\kappa}{l-1}$ whenever $\kappa^{l} \notin \Lambda_{n,0}$.
\end{proposition}

\subsection{Interlude:  Tor-groups} \label{sec:CBTTor}

The method used to compute these restrictions, as well as those of \cref{sub:restriction}, relies on the fact that the restriction functor is exact, so restricting each module in a short exact sequence results in another short exact sequence. But, as we saw at the end of \cref{sub:restriction}, the induction functor is only right-exact, meaning that identifying induced modules will require more sophisticated arguments.  We also note that induction does not commute with duality, in general, hence \cref{prop:StanRest} does not immediately identify, for instance, the induced costandard modules.

Inducing a $\Alg{n}$-module $\Mod{M}$ to an $\Alg{n+1}$-module amounts to taking the tensor product $\Alg{n+1}\otimes_{\Alg n} \Mod{M}$, where $\Alg{n+1}$ is regarded as a left $\Alg{n+1}$-module and a right $\Alg n$-module.  Just as the failure of Hom-functors to be exact is measured by extension groups, the failure of tensor products to be exact is measured by \emph{torsion groups}.  In particular, a short exact sequence
\begin{equation}
\dses{\Mod{N}''}{}{\Mod{N}}{}{\Mod{N}'}
\end{equation}
of left $\Alg{n}$-modules gives rise, upon induction, to the long exact sequence
\begin{equation}\label{eq:TorES}
\cdots \lra
\Tor(\Res{\Alg{n+1}},\Mod{N}'') \lra
\Tor(\Res{\Alg{n+1}},\Mod{N}) \lra
\Tor(\Res{\Alg{n+1}},\Mod{N}') \lra
\Ind{\Mod{N}''} \lra
\Ind{\Mod{N}} \lra
\Ind{\Mod{N}'} \lra 0,
\end{equation}
in which $\Alg{n+1}$ is viewed as a right $\Alg n$-module.  As with the Hom-Ext long exact sequences \eqref{es:lcohomology} and \eqref{es:rcohomology}, this sequence continues with higher torsion groups $\Tor_m(\Res{\Alg{n+1}},\blank)$.  Because we have no need for these higher groups, we omit the index $m$ and write $\Tor_1 \equiv \Tor$, for brevity.  We remark that as $\Alg{n+1}$ is also a left $\Alg{n+1}$-module, each of the torsion groups in \eqref{eq:TorES} is also a left $\Alg{n+1}$-module (as are the induced modules).

The following two facts about Tor-groups will be essential. First, Tor-groups are trivial when either of their arguments is flat. A projective module $\Proj{}$ is always flat, so $\Tor(\Proj{},-)\simeq\Tor(-,\Proj{})\simeq 0$. Second, if $\Mod{M}$ is a finite-dimensional right module and $\Mod{N}$ a left one over the same finite-dimensional algebra $\Alg{}$ (over a field), then the Tor-groups and Ext-groups are related by \cite[Cor.~IX.4.12]{Assem-Fr}
\begin{equation}\label{eq:TorExt}
\Tor({\Mod{M}}, \Mod{N})\simeq \VectD{\Ext(\Mod{M},\VectD{\Mod{N}})} \cong \Ext(\Mod{N}, \dual{\Mod{M}}),
\end{equation}
where $\VectD{\Mod{M}}$ denotes the vector space dual module of $\Mod{M}$ (see \cref{sec:Dual}).\footnote{With no other hypotheses, this is an isomorphism of vector spaces (over the field), hence the additional dual on the right-hand side is superfluous.  However, in the case \eqref{eq:TorES} of interest, $\Mod{M}$ is also a left module over a different algebra $\alg{B}$.  This means that $\Tor({\Mod{M}}, \Mod{N})$ is naturally a left $\alg{B}$-module, whilst $\Ext({\Mod{M}},\VectD{\Mod{N}})$ is a right $\alg{B}$-module, whence the required extra dual.}

Now, $\Alg{n+1}$ is free, hence projective, as a left $\Alg{n+1}$-module.  However, it might not be projective as a right $\Alg{n}$-module.  If it is, then the torsion groups in \eqref{eq:TorES} vanish and the induction functor is exact.  To analyse this, note that the left and right representation theories of $\tl{n}$ and $\dtl{n}$ are identical,\footnote{This is clear from the diagrammatic definitions of these algebras.} so we may combine \cref{prop:StanRest,prop:ProjRest} with the (left $\Alg{n}$-module) decomposition
\begin{equation} \label{eq:ResA}
\Res{\Alg{n+1}}=\bigoplus_{k\in\Lambda_{n+1,0}}\, (\dim \Irre{n+1,k})\, \Res{\Proj{n+1,k}}
\end{equation}
to explore the projectivity of the left $\Alg{n}$-module $\Alg{n+1}$ (the answer will be the same as a right $\Alg{n}$-module).  Introducing the symbol $\overset{\Proj{}}{\cong}$ to indicate an isomorphism up to projective direct summands, the results of this exploration may be summarised by noting that $\Res{\Proj{n+1,k}}$ is always projective for $k<n$ and is otherwise given by
\begin{equation}
\Res{\Proj{n+1,n}} \overset{\Proj{}}{\cong}
\begin{cases*}
0, & if \(n+1\) is critical, \\
\Irre{n,(n+1)^-}, & otherwise,
\end{cases*}
\qquad
\Res{\Proj{n+1,n+1}} \overset{\Proj{}}{\cong}
\begin{cases*}
\Irre{n,n}, & if \(n+1\) is critical, \\
\Irre{n,(n+1)^-}, & if \(n+2\) is critical, \\
\Irre{n,(n+1)^-} \oplus \Irre{n,(n+2)^-}, & otherwise.
\end{cases*}
\end{equation}
Here, we of course ignore modules whose indices have different parities if $\Alg{} = \tl{}$.  Since $\dim \Irre{n+1,n} = n+1$ and $\dim\Irre{n+1,n+1} = 1$, \eqref{eq:ResA} becomes
\begin{equation}\label{eq:AlgeRest}
\Res{\Alg{n+1}} \overset{\Proj{}}{\cong}
\begin{cases*}
\Irre{n,n}, & if \(n+1\) is critical, \\
(n+2) \, \Irre{n,(n+1)^-}, & if \(n+2\) is critical, \\
(n+2) \, \Irre{n,(n+1)^-} \oplus \Irre{n,(n+2)^-}, & otherwise.
\end{cases*}
\end{equation}
Thus, $\Res{\dtl{n+1}}$ is projective as a $\dtl{n}$-module, hence the induction functor for $\dtl{n}$-modules is exact, if and only if $\dtl{n}$ is semisimple.  On the other hand, this holds for $\Res{\tl{n+1}}$ if and only if $\tl{n+1}$ is semisimple \emph{or} $n+2$ is critical.

In any case, the key observation to take away from these computations is that whenever $\Res{\Alg{n+1}}$ has non-projective summands, they are irreducibles $\Irre{n,k}$ with $k = k_R$. This observation will be crucial when we identify the inductions of the $\Cost{n,k}$, $\TheB{k}{l}$ and $\TheT{k}{l}$, a task to which we now turn.

\subsection{Induction} \label{sec:CBTInd}
%
%
We begin with the inductions of the costandard modules, recalling our convention that any module with an index $k$ not in $\Lambda_{n+1,0}$ is set to zero, as are those, for $\Alg{n+1}=\tl{n+1}$, whose indices $n+1$ and $k$ have different parities.

\begin{proposition}\label{prop:CostIndu}
If $k$ is non-critical and $k,k^+ \in \Lambda_{n,0}$, then the induction of $\Cost{n,k}$ is given by
\begin{align}\label{eq:CostIndu}
\Ind{\Cost{n,k}} \simeq &
\begin{lrcases*}
\Proj{n+1,k -1}, & if $k-1$ is critical, \\
\TheB{n+1,k-1}{2}, & if $k-1$ is non-critical and $k^{++}= n+1$ or $n+2$, \\
\Cost{n+1,k-1}, & otherwise
\end{lrcases*}
\notag \\
&\oplus
\begin{lrcases*}
\TheB{n+1,k}{2}, & if $k^{++} = n+1$, \\
\Cost{n+1,k}, & otherwise
\end{lrcases*}
\oplus
\begin{lrcases*}
\Proj{n+1,k+1}, & if $k+1$ is critical, \\
\Cost{n+1,k+1}, & otherwise
\end{lrcases*}
.
\end{align}
\end{proposition}
\begin{proof}
The long exact sequence \eqref{eq:TorES} derived from the exact sequence \eqref{es:ICI} is
\begin{equation} \label{es:TorICI}
\cdots\lra\Tor(\Res{\Alg{n+1}},\Irre{n,k})
\lra\Tor(\Res{\Alg{n+1}},\Cost{n,k})
\lra\Tor(\Res{\Alg{n+1}},\Irre{n,k^{+}})
\lra \Ind{\Irre{n,k}}
\lra \Ind{\Cost{n,k}}
\lra \Ind{\Irre{n,k^{+}}}
\lra 0.
\end{equation}
Since Tor-groups involving projective modules vanish, we may replace $\Res{\Alg{n+1}}$ by the right-hand side of \eqref{eq:AlgeRest}, or rather the right module version of it, when calculating the Tor-groups in this sequence.  Let $\Irre{}^*$ denote the right module version, that is the vector space dual, of the right-hand side of \eqref{eq:AlgeRest}.  Then, \eqref{eq:TorExt} gives $\Tor(\Res{\Alg{n+1}},\Mod{M})\simeq \Ext({\Mod{M}},\Mod{I})$. We therefore have to compute $\Ext(\Irre{n,k},\Mod{I})$, $\Ext(\Cost{n,k},\Mod{I})$ and $\Ext(\Irre{n,k^+},\Mod{I})$.

Recall from \cref{sec:CBTTor} that $\Irre{}$ is a direct sum of irreducibles whose indices are always the rightmost in their orbit. Consulting \cref{prop:SIExte}, we see that $\Ext(\Irre{n,k^+},\Mod{I})$ may only be non-zero if $k=k_R^{--}$.  But then, $\Ext(\Irre{n,k},\Mod{I}) = 0$ and $\Ext(\Cost{n,k},\Mod{I}) \cong \Ext(\Irre{n,k^+},\Mod{I}) \cong \Mod{E}$ (say).  Thus, if $k \neq k_R^{--}$, then \eqref{es:TorICI} reduces to the short exact sequence
\begin{equation}\label{eq:111}
\dses{\Ind{\Irre{n,k}}}{}{\Ind{\Cost{n,k}}}{}{\Ind{\Irre{n,k^+}}}.
\end{equation}
However, if $k=k_R^{--}$, then \eqref{es:TorICI} reduces to
\begin{equation}
0 \lra \Mod{E} \lra \Mod{E} \lra \Ind{\Irre{n,k}} \lra \Ind{\Cost{n,k}} \lra \Ind{\Irre{n,k^+}} \lra 0,
\end{equation}
which also implies the exactness of \eqref{eq:111}.

From this point on, the proof follows familiar arguments.  The exact sequence \eqref{eq:111} is decomposed into (two or) three exact sequences corresponding to the eigenvalues of the central element $F_{n}$. Here is an example. Suppose first that neither $k-1$ nor $k+1$ are critical and that $k^+ \neq k_R$, so that $\Stan{n,k^+} \ncong \Irre{n,k^+}$. Then, all three short exact sequences will have the same form, namely
\begin{equation} \label{es:3ICI}
\dses{\Irre{n+1,k+i}}{}{(\Ind{\Cost{n,k}})_i}{}{\Irre{n+1,k^+-i}} \qquad \text{(\(i=0, \pm 1\)),}
\end{equation}
by \cref{prop:IrreRest}. None of these sequences can split because
\begin{equation}\label{eq:1234}
\Hom(\Ind{\Cost{n,k}}, \Irre{n+1,k+i}) \simeq \Hom(\Cost{n,k}, \Res{\Irre{n+1,k+i}}) \simeq \Hom(\Cost{n,k}, \Irre{n,k+i-1} \oplus \Irre{n,k+i} \oplus \Irre{n,k+i+1}) \simeq 0,
\end{equation}
by Frobenius reciprocity and \cref{prop:IrreRest,prop:SIPHomo}. \cref{prop:SIExte,coro:StanUnique} now give the conclusion: $(\Ind{\Cost{n,k}})_i \simeq \Cost{n,k+i}$, hence $\Ind{\Cost{n,k}} \cong \Cost{n,k-1} \oplus \Cost{n,k} \oplus \Cost{n,k+1}$.

When either $k-1$ or $k+1$ is critical, but still $k^+ \neq k_R$, some of the irreducibles in the three exact sequences \eqref{es:3ICI} are critical, hence projective, and some are replaced by $0$, according to \cref{prop:IrreRest}.  The results are as before, except that $(\Ind{\Cost{n,k}})_{-1} \cong \Proj{n+1,k-1}$ or $(\Ind{\Cost{n,k}})_{+1} \cong \Proj{n+1,k^+-1} = \Proj{n+1,k+1}$, respectively.

Finally, the case $k^+ = k_R$ leads to the novel summands $\TheB{n+1,k'}{2}$ in \eqref{eq:CostIndu}. \cref{prop:IrreRest} may be used only when $\Irre{n,k^+}\not\simeq\Stan{n,k^+}$; otherwise, \cref{prop:StanRest} applies instead and $\Ind{\Irre{n,k^+}} \simeq \Stan{n+1,k^+-1} \oplus \Stan{n+1,k^+} \oplus \Stan{n+1,k^++1}$. The question is whether these standard modules are irreducible or not, for if so, then the analysis proceeds as above. Now, $\Stan{n+1,k^++1}$ is only reducible when $k^++1$ is non-critical and $(k^++1)^{+}=k^{++}-1 \in \Lambda_{n+1,0}$, that is, when $k^{++}=n+1$ or $n+2$. Similarly, $\Stan{n+1,k^+}$ is only reducible when $k^{++}=n+1$ and $\Stan{n+1,k^+-1}$ is never reducible. In these few cases, some of the three short exact sequences, obtained by decomposing \eqref{eq:111}, describe non-split extensions of an irreducible by a reducible standard. \cref{prop:SIExte,lem:indecExtension,prop:3point1} identify the extensions as being isomorphic to $\TheB{n+1,k-1}{2}$ or $\TheB{n+1,k}{2}$, completing the proof.
\end{proof}

%
%
\begin{proposition}\label{prop:TheBIndu}
Let $k$ be non-critical with $k,k^l \in \Lambda_{n,0}$.  If $l$ is even, then the induction of $\TheB{n,k}{l}$ is given by
\begin{subequations} \label{eq:IndB}
\begin{equation} \label{eq:IndBEven}
\Ind{\TheB{n,k}{l}} \simeq \Res{\TheB{n+2,k}{l}}.
\end{equation}
If $l=2i+1$ is odd, then the induction is instead given by
\begin{align} \label{eq:IndBOdd}
\Ind{\TheB{n,k}{l}} \simeq &
\begin{dlrcases*}
\bigoplus_{j=0}^{i} \Proj{n+1,k^{2j}-1}, & if $k-1$ is critical, \\
\TheB{n+1,k-1}{l+1}, & if $k-1$ is non-critical and $k^{l+1}=n+1$ or $n+2$, \\
\TheB{n+1,k-1}{l}, & otherwise
\end{dlrcases*}
\notag\\
&\oplus
\begin{lrcases*}
\TheB{n+1,k}{l+1}, & if $k^{l+1} = n+1$, \\
\TheB{n+1,k}{l}, & otherwise
\end{lrcases*}
\oplus
\begin{dlrcases*}
\bigoplus_{j=0}^{i-1} \Proj{n+1,k^{2j}+1}, & if $k+1$ is critical, \\
\TheB{n+1,k+1}{l}, & otherwise
\end{dlrcases*}
.
\end{align}
\end{subequations}
\end{proposition}
\begin{proof}
The proof is by induction on $l$, distinguishing the two parities. If $l=0$, then $\TheB{n,k}0=\Irre{n,k}$ and the result was given in \cref{prop:StanRest,prop:IrreRest}. If $l=1$, then $\TheB{n,k}1=\Cost{n,k}$ and the result was given in \cref{prop:CostIndu}.

So, let $l\ge2$ be an even integer. The long exact sequence obtained by inducing the defining sequence
\begin{equation} \label{es:BBS}
\dses{\TheB{n,k}{l-2}}{}{\TheB{n,k}{l}}{}{\Stan{n,k^{l-1}}}
\end{equation}
of \cref{prop:3point1} has $\Tor(\Res{\Alg{n+1}}, \Stan{n,k^{l-1}}) = 0$. Indeed, this torsion group is isomorphic to $\Ext(\Stan{n,k^{l-1}}, \Irre{})$, where $\Irre{}$ is a direct sum of non-critical irreducibles whose indices have the form $k'_R$, and $\Ext(\Stan{n,k^{l-1}}, \Irre{k'_R})$ is always zero, by \cref{prop:SIExte}. We therefore arrive at the short exact sequence
\begin{equation} \label{es:BBI}
\dses{\Ind{\TheB{n,k}{l-2}}}{}{\Ind{\TheB{n,k}{l}}}{}{\Ind{\Stan{n,k^{l-1}}}}.
\end{equation}
The submodule here is known, by the induction hypothesis, and the quotient is known by \cref{prop:StanRest}. The identification of $\Ind{\TheB{n,k}{l}}$ now follows the same arguments as in the proof of \cref{prop:CostIndu}.  In particular, \eqref{es:BBI} is decomposed into (at most) three exact sequences whose nature, split or non-split, is established through the presence of projectives or by calculating Hom-groups.  For example, if $k-1$ and $k$ are not critical, then
\begin{equation}
\Hom(\Ind{\TheB{n,k}l}, \TheB{n+1,k-1}{l-2})\simeq
\Hom(\TheB{n,k}l, \Res{\TheB{n+1,k-1}{l-2}})\simeq\Hom(\TheB{n,k}l,\TheB{n,k}{l-2})=0,
\end{equation}
as otherwise $\TheB{n,k}l$ would be decomposable.  \cref{lem:indecExtension,prop:3point1} then complete the identification.

One case with $l$ odd also fits into this induction argument, that with $k^l = k_R$.  Then, \eqref{es:BBS} is replaced by the defining exact sequence
\begin{equation}
\dses{\TheB{n,k}{l-1}}{}{\TheB{n,k}{l}}{}{\Stan{n,k^{l}}}
\end{equation}
and the subsequent analysis follows similar lines to that performed for $\Ind{\Cost{n,k}}$ (for these values of $k$).  This case leads to the result for $\Ind{\TheB{n,k}{l}}$, when $k^{l+1}=n+1$ or $n+2$, in \eqref{eq:IndBOdd}.

The induction process for $l=2i+1 \ge 3$ odd starts with the computation of the Tor-groups related to the defining sequence of $\TheB{n,k}{2i+1}$:
\begin{equation}
\dses{\TheB{n,k^{++}}{2i-1}}{}{\TheB{n,k}{2i+1}}{}{\Cost{n,k}}.
\end{equation}
\cref{lem:trivialExtensions} indicates that each $\Ext(\Cost{n,k},\Irre{k'_R})$ is always zero because the composition factors of the two modules are sufficiently separated in their orbit (for example, $k^{++}$ stands between them). The long exact sequence thus reduces to
\begin{equation}
\dses{\Ind{\TheB{n,k}{2i-1}}}{}{\Ind{\TheB{n,k}{2i+1}}}{}{\Ind{\Cost{n,k}}}
\end{equation}
and, from this point on, the proof closely follows that for $l$ even.
\end{proof}

Our final induction result does not require any new techniques, so we omit the proof.
%
%
\begin{proposition}
Let $k$ be non-critical with $k,k^l \in \Lambda_{n,0}$.  If $l$ is odd, then the induction of $\TheT{n,k}{l}$ is given by
\begin{subequations}
\begin{equation}
\Ind{\TheT{n,k}{l}} \simeq \Res{\TheT{n+2,k}{l}}.
\end{equation}
If $l=2i$ is even, then the induction is instead given by
\begin{align}
\Ind{\TheT{n,k}{l}} \simeq &
\begin{dlrcases*}
\bigoplus_{j=0}^i \Proj{n+1,k^{2j}-1}, & if $k-1$ is critical, \\
\TheT{n+1,k-1}{l}, & otherwise
\end{dlrcases*}
\oplus
\begin{lrcases*}
\TheT{n+1,k}{l+1}, & if $k^{l+1} = n+1$, \\
\TheT{n+1,k}{l}, & otherwise
\end{lrcases*}
\notag \\
&\oplus
\begin{dlrcases*}
\bigoplus_{j=0}^{i-1} \Proj{n+1,k^{2j}+1}, & if $k+1$ is critical, \\
\TheT{n+1,k+1}{l+1}, & if $k+1$ is non-critical and $k^{l+1}=n+1$ or $n+2$, \\
\TheT{n+1,k+1}{l}, & otherwise
\end{dlrcases*}
.
\end{align}
\end{subequations}
\end{proposition}

%
%

\section{The Auslander-Reiten quiver for $\tl{n}$ and $\dtl{n}$}\label{sec:ARquiver}

\cref{sub:completeness} completed the classification of all indecomposable modules, up to isomorphism, over $\Alg{} = \tl{n}$ and $\dtl{n}$. The tools required included basic homological algebra and the representation theory of associative algebras: properties of injective and projective modules, extension groups and diagram chasing. The input, summed up in \cref{sec:standardProjective}, was the list of irreducible, projective and injective modules, their structures, and those of the standard and costandard modules, as captured either by short exact sequences or Loewy diagrams.

There are more advanced techniques to perform the same task. One of them is \AR{} theory. The method it underlies is purely algorithmic, at least in the case of $\tl{n}$ and $\dtl{n}$, and its input is again the data recalled in \cref{sec:standardProjective}.  The current section presents the application of this method to $\tl{n}$ and $\dtl{n}$.

First, we review the theoretical results of \AR{} theory. Then, we show how these abstract results may be applied algorithmically and carry this out on the simple case of an orbit $[k]$ of non-critical integers that contains only $3$ elements. We conclude by giving the result of applying this algorithm in the general case, providing only sketches of proofs, and thereby recover the classification of indecomposable $\Alg{}$-modules.

%
%
\subsection{The main results of \AR{} theory}

This subsection reviews the main ideas and results of \AR{} theory that we shall need to build a complete list of indecomposable modules of the algebras $\tl{n}$ and $\dtl{n}$. Our summary closely follows Chapter IV of \cite{Assem}, though the results are not necessarily presented in the same order that they are proved there.

\begin{center}
\emph{In this subsection, $\Alg{}$ stands for any finite-dimensional $\KK$-algebra, where $\KK$ is an algebraically closed field.}
\end{center}

Homological algebra studies modules through their Hom-groups. We first recall that a monomorphism (injective homomorphism) $f\colon\Mod{U}\to \Mod{V}$ is split if there exists $g\colon\Mod{V}\to \Mod{U}$ such that $gf$ is the identity on $\Mod{U}$. Similarly, an epimorphism (surjective homomorphism) $f\colon\Mod{U}\to \Mod{V}$ is split if there exists $g\colon\Mod{V}\to \Mod{U}$ such that $fg$ is the identity on $\Mod{V}$. A morphism is said to be split if it is either a split monomorphism or a split epimorphism. \AR{} theory studies refinements of these concepts.

\begin{definition}
	Let $f\colon\Mod{U} \to \Mod{V}$ be a morphism between two $\Alg{}$-modules.
	\begin{itemize}
		\item The morphism $f$ is \emph{left minimal almost split} if
		\begin{enumerate}
			\item $f$ is not a split monomorphism;
			\item for any morphism $g\colon \Mod{U} \to \Mod{W}$ that is not a split monomorphism, there exists a morphism $\bar{f}\colon \Mod{V} \to \Mod{W}$ such that $\bar{f} f = g$;
			\item for any $h \colon \Mod{V} \to \Mod{V}$, $hf=f$ implies that $h$ is an isomorphism.
		\end{enumerate}
		\item The morphism $f$ is \emph{right minimal almost split} if
		\begin{enumerate}
			\item $f$ is not a split epimorphism;
			\item for any morphism $g\colon \Mod{W} \to \Mod{V}$ that is not a split epimorphism, there exists a morphism $\bar{f}\colon \Mod{W} \to \Mod{U}$ such that $f \bar{f} = g$;
			\item for any $h \colon \Mod{U} \to \Mod{U}$, $fh=f$ implies that $h$ is an isomorphism.
		\end{enumerate}
		\item The morphism $f$ is \emph{irreducible} if
		\begin{enumerate}
			\item $f$ is not a split morphism;
			\item given any morphisms $f_1 \colon \Mod{U} \to \Mod{Z}$ and $f_2 \colon \Mod{Z} \to \Mod{V}$ satisfying $f = f_2 f_1$, then either $f_2$ is a split epimorphism or $f_1$ is a split monomorphism.
		\end{enumerate}
	\end{itemize}
\end{definition}
\begin{definition}
	A short exact sequence
	\begin{equation}
	\dses{\Mod{U}}{\iota}{\Mod{V}}{\pi}{\Mod{V}/\Mod{U}}
	\end{equation}
	is said to be an \emph{almost split exact sequence} if
	\begin{itemize}
		\item $\iota$ is left minimal almost split;
		\item $\pi$ is right minimal almost split.
	\end{itemize}
	In fact, it can be shown that either of these two conditions implies the other.
\end{definition}
It is useful to draw a parallel between irreducible morphisms and semisimple modules. Intuitively, a morphism is irreducible if it cannot be expressed as a composition of non-split morphisms, just as a module is semisimple if it cannot be expressed as a non-trivial extension of one module by another. Similarly, finite-dimensional modules can be studied through their Loewy diagrams in terms of extensions of simple modules, while morphisms between modules in a representation-finite algebra can be expressed as sums of compositions of irreducible morphisms. As for the left (right) minimal almost split morphisms, they turn out to be in one-to-one correspondence with the non-injective (non-projective) finite-dimensional indecomposable modules. In particular, the following proposition notes that the domain of every left minimal almost split morphism is indecomposable, while \cref{prop:weaving} below states that every non-injective finite-dimensional indecomposable module is the domain of a left minimal almost split morphism.
\begin{proposition}\label{prop:leftMinimal}
	Let $\Mod{U}$ and $\Mod{V}$ be $\Alg{}$-modules.  If $f\colon\Mod{U} \to \Mod{V}$ is a left minimal almost split morphism, then
	\begin{enumerate}
		\item $\Mod{U}$ is indecomposable;
		\item $f$ is irreducible;
		\item if $f'\colon\Mod{U} \to \Mod{V}'$ is also left minimal almost split, for some $\Alg{}$-module $\Mod{V}'$, then there exists an isomorphism $g\colon \Mod{V} \to \Mod{V}'$ such that $gf = f'$;
		\item a morphism $f'\colon\Mod{U} \to \Mod{V}'$ of $\Alg{}$-modules is irreducible if and only if $\Mod{V}' \neq 0$, $\Mod{V} \simeq \Mod{V}' \oplus \Mod{V}''$ for some module $\Mod{V}''$, and there exists $f''\colon\Mod{U} \to \Mod{V}''$ such that $f' \oplus f'' \colon \Mod{U} \to \Mod{V}' \oplus \Mod{V}''$ is left minimal almost split.  In particular, composing $f$ with the projection onto a (non-zero) direct summand always gives an irreducible morphism. \label{it:IrrDirSum}
	\end{enumerate}
\end{proposition}
\noindent A similar result holds for right minimal almost split morphisms.

Note that while irreducible and almost split morphisms are very useful, identifying and constructing them is far from trivial. Fortunately \AR{} theory provides a few irreducible morphisms to get one started (the next result), a way of building new ones out of old ones, and tools for classifying them.  First, however, it is convenient to introduce two more definitions.
\begin{definition}
	\leavevmode
	\begin{itemize}
		\item For two $\Alg{}$-modules $\Mod{U}$ and $\Mod{V}$, define the \emph{radical} of $\Hom_{\Alg{}}(\Mod{U},\Mod{V})$ to be the vector space $r(\Mod{U},\Mod{V})$ spanned by the morphisms $f\colon\Mod{U} \to \Mod{V}$ such that for all indecomposable $\Alg{}$-modules $\Mod{Z}$ and all morphisms $g\colon\Mod{Z} \to \Mod{U}$ and $h\colon\Mod{V} \to \Mod{Z}$, the composition $hfg$ is not an isomorphism.
		\item Define the \emph{second radical} of $\Hom_{\Alg{}}(\Mod{U},\Mod{V})$ to be the vector subspace $r^{2}(\Mod{U},\Mod{V}) \subseteq r(\Mod{U},\Mod{V})$ consisting of the morphisms $f\colon\Mod{U} \to \Mod{V}$ which may be factored as $f=gh$, where $g\in r(\Mod{Z},\Mod{V})$ and $h\in r(\Mod{U},\Mod{Z})$, for some $\Alg{}$-module $\Mod{Z}$.
	\end{itemize}
\end{definition}
\noindent We remark that if $\Mod{U}$ and $\Mod{V}$ are indecomposable, then $r(\Mod{U},\Mod{V})$ is the subspace of non-isomorphisms in $\Hom(\Mod{U},\Mod{V})$.
\begin{proposition}\label{prop:someIrrMorph}
	\leavevmode
	\begin{enumerate}
		\item If $\Proj{}$ is a non-simple projective indecomposable and $\Inje{}$ a non-simple injective indecomposable, then the canonical inclusion and projection,
		\begin{equation}
		\rad \Mod{P} \lira \Mod{P} \qquad \text{and} \qquad \Mod{J} \lsra \frac{\Mod{J}}{\soc \Mod{J}},
		\end{equation}
		are irreducible. Furthermore, these are the only irreducible morphisms with target $P$ and source $J$, respectively, up to rescaling. \label{it:RadSocIrr}
		\item If $\Mod{Q}$ is a non-simple projective and injective indecomposable with $\rad \Mod{Q} \neq \soc \Mod{Q}$, then the canonical projection and inclusion,
		\begin{equation}
		\rad \Mod{Q} \lsra \frac{\rad \Mod{Q}}{\soc \Mod{Q}} \qquad \text{and} \qquad
		\frac{\rad \Mod{Q}}{\soc \Mod{Q}} \lira \frac{\Mod{Q}}{\soc \Mod{Q}},
		\end{equation}
		are irreducible. \label{it:Rad/SocIrr}
		\item When $U$ and $V$ are indecomposable, $f\colon\Mod{U}\to \Mod{V}$ is irreducible if and only if $f\in r(\Mod{U},\Mod{V})$ but $f \notin r^2(\Mod{U},\Mod{V})$. \label{it:Irr=r-r^2}
	\end{enumerate}
\end{proposition}
\noindent Recall that the radicals of the projectives and the socles of the injectives are known for $\Alg{}=\tl{n}$ and $\dtl{n}$. Moreover, if $[k]$ is a non-critical orbit, then each projective $\Proj{k'}$, with $k' \in [k] \setminus \set{k_L}$, is also injective. The previous proposition therefore provides several irreducible morphisms for these algebras.

\begin{definition}
	Let $\Mod{U}$ be a left (or right) $\Alg{}$-module.
	\begin{itemize}
		\item The vector space dual $\VectD{\Mod{U}} \equiv \Hom_{\KK} \left( \Mod{U}, \KK \right)$ (introduced in \cref{sec:Dual} for $\KK = \CC$) is a right (or left) $\Alg{}$-module with action $(fa)(x) = f(ax)$ (or $(af)(x) = f(xa)$), for all $f \in \VectD{\Mod{U}}$ and $a \in \Alg{}$.
		\item The algebra dual $\AlgD{\Mod{U}} \equiv \Hom_{\Alg{}}\left(\Mod{U},\Alg{}\right)$ is a right (or left) $\Alg{}$-module with action $(fa)(x) = f(x)a$ (or $(af)(x) = a f(x)$), for all $f \in \AlgD{\Mod{U}}$ and $a \in \Alg{}$.
	\end{itemize}
\end{definition}
\noindent As with the twisted dual of \cref{sec:Dual}, the vector space dual defines a contravariant exact functor.  The functor for the algebra dual, on the other hand, is contravariant but only left-exact.

Recall that a projective presentation of a module $\Mod{U}$ is a short exact sequence $0 \ra \ker p \ra \Mod{P} \xrightarrow{p} \Mod{U} \ra 0$ in which $\Mod{P}$ is projective.  Replacing $\ker p$ by another projective $\Mod{Q}$, its projective cover for example, gives another exact sequence:  $\Mod{Q} \xrightarrow{q} \Mod{P} \xrightarrow{p} \Mod{U} \ra 0$.  This sequence is said to be a \emph{minimal} projective presentation of $\Mod{U}$ if $p \colon \Mod{P} \to \Mod{U}$ is a projective cover of $\Mod{U}$ and $q \colon \Mod{Q} \to \ker p$ is a projective cover of $\ker p$. Apply the functor $\AlgD{(\blank)}$ to obtain one last exact sequence:
\begin{equation}
0 \lra \AlgD{\Mod{U}} \overset{\AlgD{p}}{\lra} \AlgD{\Mod{P}} \overset{\AlgD{q}}{\lra} \AlgD{\Mod{Q}} \lra \Coker \AlgD{q} \lra 0.
\end{equation}
\begin{definition}
	The \emph{\AR{} transpose} of the left (right) $\Alg{}$-module $\Mod{U}$ is the right (left) $\Alg{}$-module
	\begin{equation}
	\ARtransp{\Mod{U}} = \Coker \AlgD{q}.
	\end{equation}
\end{definition}
\noindent We remark that the isomorphism class of the \AR{} transpose does not depend upon the choice of minimal projective presentation.

The algebra dual $\AlgD{\Mod{U}}$ may be quite different to the vector space dual $\VectD{\Mod{U}}$, even as vector spaces. Here is a simple example for $\Alg{}=\tl{n}$. Consider the irreducible left $\tl{n}$-module $\Irre{k_L}$ corresponding to the smallest integer in a non-critical orbit $[k]$ that contains at least two integers. Of course, $\VectD{\Irre{k_L}}$ is the corresponding irreducible right $\tl{n}$-module, a fact that we exploited in \cref{sec:CBTInd}. On the other hand, $\AlgD{\Irre{k_L}}$ is the space of all homomorphisms from $\Irre{k_L}$ to the left $\tl{n}$-module $\tl{n}$. However $\tl{n}$ is a direct sum of indecomposable projectives, none of which contain $\Irre{k_L}$ in their socles. Therefore, the only such homomorphism is zero and $\AlgD{\Irre{k_L}}=0$.

Computing $\AlgD{\Mod{U}}$ can thus be tricky. Fortunately, we shall only need to compute this dual when $\Mod{U}$ is projective and, in this case, the next result gives the answer.
\begin{proposition}\label{prop:ARtranspose}
	\leavevmode
	\begin{enumerate}
		\item If $\Mod{P}$ is a finitely generated projective $\Alg{}$-module, then so is $\AlgD{\Mod{P}}$. In particular, if $\Mod{P} = \Alg{} e$ for some idempotent $e \in \Alg{}$, then $\AlgD{\Mod{P}} \simeq e \Alg{}$. \label{it:ARTr-Proj}
		\item An indecomposable $\Mod{U}$ is projective if and only if $\ARtransp{\Mod{U}} = 0$; \label{it:ARTr-Vanishing}
		\item If $\Mod{U}$ is indecomposable and not projective, with minimal projective presentation $\rdses{\Mod{Q}}{q}{\Mod{P}}{p}{\Mod{U}}$, then  \label{it:ARTr-Involution}
		\begin{itemize}
			\item the exact sequence
			\begin{equation}
			\rdses{\AlgD{\Mod{P}}}{\AlgD{q}}{\AlgD{\Mod{Q}}}{}{\ARtransp{\Mod{U}}}
			\end{equation}
			is a minimal projective presentation of $\ARtransp{\Mod{U}}$;
			\item the module $\ARtransp{\Mod{U}}$ is indecomposable and not projective;
			\item $\ARtransp{(\ARtransp{\Mod{U}})} \simeq \Mod{U}$.
		\end{itemize}
	\end{enumerate}
\end{proposition}

Of course, the \AR{} transpose of a left module is a right module, when we are really only interested in classifying left modules.  We therefore modify this construction one last time.
\begin{definition}
	Let $\Mod{U}$ be an indecomposable left (right) module. The \emph{\AR{} translation} of $\Mod{U}$ is the left (right) module
	\begin{subequations}
	\begin{equation}
	\ARTau{\Mod{U}} = \VectD{(\ARtransp{\Mod{U}})}
	\end{equation}
	and the \emph{inverse \AR{} translation} of $\Mod{U}$ is the left (right) module
	\begin{equation}
	\iARTau{\Mod{U}} = \ARtransp{(\VectD{\Mod{U}})}.
	\end{equation}
	\end{subequations}
\end{definition}

Since $\dual{(\blank)}$ preserves indecomposability, \cref{prop:ARtranspose} shows that a non-projective (non-injective) indecomposable $\Mod{U}$ has an indecomposable translation $\ARTau{\Mod{U}}$ (inverse translation $\iARTau{\Mod{U}}$). This is a key feature of \AR{} translation because it allows one to construct new indecomposable modules from known ones. Another key feature is that it can also reveal new irreducible morphisms. The next proposition sums up these features.
%
%
\begin{proposition}\label{prop:weaving}
	Let $\Mod{U}$ and $\Mod{V}$ be indecomposable $\Alg{}$-modules and denote the set of isomorphism classes of finite-dimensional indecomposable $\Alg{}$-modules by $\Omega$.
	\begin{enumerate}
		\begin{subequations}
		\item If $\Mod{V}$ is not projective, then
		\begin{itemize}
			\item $\ARTau{\Mod{V}}$ is indecomposable and not injective, with $\iARTau{(\ARTau{\Mod{V}})} \simeq \Mod{V}$;
			\item As vector spaces,
			$r(\Mod{U},\Mod{V})/r^{2}(\Mod{U},\Mod{V}) \simeq r(\Mod{\ARTau{\Mod{V}}},\Mod{U})/r^{2}(\Mod{\ARTau{\Mod{V}}},\Mod{U})$;
			\item there exists a unique, up to isomorphism, almost split short exact sequence
			\begin{equation} \label{es:CheckSum1}
			\dses{\ARTau{\Mod{V}}}{}{\bigoplus_{\Mod{M} \in \Omega} s(\Mod{M})\cdot \Mod{M}}{}{\Mod{V}},
			\end{equation}
			where $s(\Mod{M}) = \dim (r(\Mod{M},\Mod{V})/r^{2}(\Mod{M},\Mod{V}))$.
		\end{itemize}
		\item If $\Mod{U}$ is not injective, then
		\begin{itemize}
			\item $\iARTau{\Mod{U}}$ is indecomposable and not projective, with $\ARTau{(\iARTau{\Mod{U}})} \simeq \Mod{U}$;
			\item As vector spaces,
			$r(\Mod{U},\Mod{V})/r^{2}(\Mod{U},\Mod{V}) \simeq r(\Mod{V},\iARTau{\Mod{U}})/r^{2}(\Mod{V},\iARTau{\Mod{U}})$;
			\item there exists a unique, up to isomorphism, almost split short exact sequence
			\begin{equation} \label{es:CheckSum2}
			\dses{\Mod{U}}{}{\bigoplus_{\Mod{M} \in \Omega} t(\Mod{M})\cdot \Mod{M}}{}{\iARTau{\Mod{U}}},
			\end{equation}
			where $t(\Mod{M})= \dim (r(\Mod{U},\Mod{M})/r^{2}(\Mod{U},\Mod{M}))$.
		\end{itemize}
		\end{subequations}
	\end{enumerate}
\end{proposition}
\noindent As we shall see, the non-negative integers $s(\Mod{M})$ and $t(\Mod{M})$ usefully measure the ``sizes'' of the sets $r(\Mod{M},\Mod{V}) \setminus r^{2}(\Mod{M},\Mod{V})$ and $r(\Mod{U},\Mod{M}) \setminus r^{2}(\Mod{U},\Mod{M})$ of irreducible morphisms from $\Mod{M}$ to $\Mod{V}$ and from $\Mod{U}$ to $\Mod{M}$, respectively.\footnote{Recall that the set of irreducible morphisms does not form a vector space, see \cref{prop:someIrrMorph}\ref{it:Irr=r-r^2}, explaining the slightly awkward language employed here.} One of the key assertions of the previous proposition is then that the set of irreducible morphisms from $\Mod{U}$ to $\Mod{V}$ has the same size as that from $\ARTau{\Mod{V}}$ to $\Mod{U}$, if $\Mod{V}$ is not projective, and as that from $\Mod{V}$ to $\iARTau{\Mod{U}}$, if $\Mod{U}$ is not injective.

Starting with a single irreducible morphism $\Mod{U} \to \Mod{V}$, \AR{} translation therefore infers an iterative sequence of irreducible morphisms between indecomposable modules that may be composed to form a chain:
\begin{equation} \label{eq:ARChain}
\cdots \lra \ARTau{(\ARTau{\Mod{V}})} \lra \ARTau{\Mod{U}} \lra \ARTau{\Mod{V}} \lra \Mod{U} \lra \Mod{V} \lra \iARTau{\Mod{U}} \lra \iARTau{\Mod{V}} \lra \iARTau{(\iARTau{\Mod{U}})} \lra \cdots.
\end{equation}
This chain terminates if either $\PowerARTau{m}\Mod{V}$ or $\PowerARTau{m}\Mod{U}$ is projective or if $\PowerARTau{-m}\Mod{V}$ or $\PowerARTau{-m}\Mod{U}$ is injective. Note that it may also happen that $\PowerARTau m \Mod{U}\simeq \Mod{U}$ and $\PowerARTau m \Mod{V}\simeq \Mod{V}$, for some $m \in \ZZ$, in which case the chain \eqref{eq:ARChain} becomes a cycle of irreducible morphisms. The algebras $\tl{n}$ and $\dtl{n}$ will provide examples of both possibilities.

\cref{prop:weaving} gives more than just a way to construct new indecomposable modules and irreducible morphisms. It can also be used to verify the completeness of a set of indecomposable modules. To see this, suppose that a set $\Omega'$ of inequivalent finite-dimensional indecomposable modules has been identified, along with the irreducible morphisms between them.  We can test for its completeness as follows. For each non-projective module $\Mod{V} \in \Omega'$, one can list the known (linearly independent) irreducible morphisms whose targets are $\Mod{V}$.  If $s'(\Mod{M})$ is the number of these morphisms with source $\Mod{M}$, then one can count composition factors to check whether the sequence
\begin{equation}
\dses{\ARTau{\Mod{V}}}{}{\bigoplus_{\Mod{M} \in \Omega'} s'(\Mod{M})\cdot \Mod{M}}{}{\Mod{V}}
\end{equation}
could be exact.  If $\ARTau{\Mod{V}}$ and $\Mod{V}$ have too many composition factors, \eqref{es:CheckSum1} tells us that the set $\Omega'$ is not complete or that we have not found all the irreducible morphisms.

A second test, based on \eqref{es:CheckSum2}, checks if, for every non-injective $\Mod{U}$, the composition factors of $\Mod{U}$ and $\iARTau{\Mod{U}}$ match those of $\bigoplus_{\Mod{M} \in \Omega'} t'(\Mod{M})\cdot \Mod{M}$, where $t'(\Mod{M})$ is the number of known (linearly independent) irreducible morphisms with source $\Mod{U}$ and target $\Mod{M}$. Again, if this test fails for a single $\Mod{U}$, then either $\Omega'$ is not complete or some irreducible morphisms are missing.

Suppose however, that the algebra $\Alg{}$ is \emph{connected}, meaning that it cannot be written as a direct sum of more than one \emph{block} (non-trivial indecomposable two-sided ideal).  If both tests pass, for all non-injectives $\Mod{U}$ and non-projectives $\Mod{V}$, then $\Omega'$ indeed yields a complete set of isomorphism classes of indecomposable modules.  If $\Alg{}$ is not connected, then one simply restricts these tests to each block.

For clarity, the indecomposable modules and irreducible morphisms of $\Alg{}$ are often represented graphically as a quiver (graph).
\begin{definition}
	Let $\Alg{}$ be a finite-dimensional associative $\KK$-algebra. The \emph{\AR{} quiver} $\Gamma(\lmod)$ of the category $\lmod$ of finite-dimensional left $\Alg{}$-modules is defined as follows:
	\begin{itemize}
		\item The vertices of $\Gamma(\lmod)$ are the isomorphism classes $\sqbrac{\Mod{U}}$ of indecomposable modules $\Mod{U}$ in $\lmod$.
		\item The arrows $\sqbrac{\Mod{U}} \to \sqbrac{\Mod{V}}$, for $\Mod{U},\Mod{V}$ in $\lmod$, are in one-to-one correspondence with basis vectors of the $\KK$-vector space $r(\Mod{U},\Mod{V})/r^2(\Mod{U},\Mod{V})$.
	\end{itemize}
\end{definition}
\begin{proposition}\label{prop:ARisconnected}
	\leavevmode
	\begin{enumerate}
		\item If $\Alg{}$ is a connected finite-dimensional $\KK$-algebra, then $\Gamma(\lmod)$ is connected and the number of arrows between any two vertices is finite. \label{it:QuiverConnected}
		\item $\Alg{}$ is of finite-representation type if and only if $\Gamma(\lmod)$ is a finite quiver.
		\item If the algebras $\Alg{1}$ and $\Alg{2}$ have equivalent module categories, then their \AR{} quivers are identical.
		\item If $\Alg{}$ is not connected, then $\Gamma(\lmod)$ is the disjoint union of the \AR{} quivers of its blocks.
	\end{enumerate}
\end{proposition}
\noindent Of course, one can similarly study the \AR{} quiver $\Gamma(\rmod)$ of the category $\rmod$ of finite-dimensional right $\Alg{}$-modules.  For $\tl{n}$ and $\dtl{n}$, the left and right quivers are isomorphic.

%
%
\subsection{A detailed example}\label{sec:ARexample}

As an example of the algorithmic construction afforded by the abstract results of the previous subsection, we compute the \AR{} quiver for a block of $\Alg{}$ corresponding to a non-critical orbit containing three elements $\{k_1=k_L, k_2, k_3=k_R\}$.  (We omit the case where $n$ is even and $\beta=0$, for $\Alg{}=\tl{n}$, deferring its study to the end of the next subsection.) In what follows, we shall use the notation $\TheB{k}{l}$ and $\TheT{k}{l}$ for every module that is not projective, injective or irreducible, slightly modified for brevity so that the index $k = k_i$ is replaced by the label $i$.  The same modification will be applied to the projectives, injectives and irreducibles.  Thus, $\Irre{2}$, $\Proj{3}$, $\TheB{2}{1}$ and $\TheT{1}{2}$ now stand for $\Irre{k_2}$, $\Proj{k_3}$, $\TheB{k_2}{1}$ and $\TheT{k_1}{2}$, respectively, whilst we shall prefer $\Proj{1}$ and $\Inje{1}$ over $\Stan{1} \cong \TheT{1}{1}$ and $\Cost{1} \cong \TheB{1}{1}$, respectively.  This notation extends in an obvious fashion to orbits of arbitrary length.

We break down the construction of the \AR{} quiver into three steps: The computation of the translation $\ARTau{}$ on indecomposable modules, the identification of irreducible morphisms, and the drawing and check of completeness of the \AR{} quiver (and therefore of the list of isomorphism classes of indecomposable $\Alg{}$-modules).

\medskip

\noindent\emph{The action of $\ARTau{}$ on indecomposable modules}\ ---\ We assume the results of \cref{sec:standardProjective}, that is the existence of the irreducible, standard, costandard, injective and projective modules, as well as their Loewy diagrams. The translation $\ARTau{}$ will be applied to all non-projectives of this list, and then on the new ones thus obtained, until the process does not introduce any new non-projective indecomposable modules.

Let us start with the (left) injective module $\Inje{1} = \Cost 1 = \TheB 11$. The first step is to construct a minimal projective presentation $\Mod{Q}\to \Mod{P}\to \Inje{1}\to 0$. It is
\begin{equation}
\begin{tikzpicture}[baseline={(l3.base)},scale=1/2]
\node (l2) at (2,2) [] {$\scriptstyle{3}$};
\node (l3) at (1,1) [] {$\scriptstyle{2}$};
\node (l4) at (2,0) [] {$\scriptstyle{3}$};
\draw [-] (l2) -- (l3);
\draw [-] (l3) -- (l4);
\end{tikzpicture}
\overset q\lra
\begin{tikzpicture}[baseline={(l3.base)},scale=1/2]
\node (l2) at (2,2) [] {$\scriptstyle{2}$};
\node (l3) at (1,1) [] {$\scriptstyle{1}$};
\node (l5) at (3,1) [] {$\scriptstyle{3}$};
\node (l4) at (2,0) [] {$\scriptstyle{2}$};
\draw [-] (l2) -- (l3);
\draw [-] (l2) -- (l5);
\draw [-] (l3) -- (l4);
\draw [-] (l5) -- (l4);
\end{tikzpicture}
\overset p\lra
\begin{tikzpicture}[baseline={(t.base)},scale=1/2]
\node (l2) at (2,2) [] {$\scriptstyle{2}$};
\node (l3) at (1,1) [] {$\scriptstyle{1}$};
\draw [-] (l2) -- node (t) {$\vphantom{\scriptstyle 2}$} (l3);
\end{tikzpicture}
\lra 0 .
\end{equation}
(We have replaced the names of the modules by their Loewy diagrams and each composition factor by its index $i$.  We have also omitted the arrowheads on the Loewy diagrams for clarity --- they all point down.) To see that this is a minimal projective presentation, note first that the projective cover of $\Inje{1}$ is $\Proj{2} \overset{p}{\ra} \Inje{1}$, by \cref{prop:ProjInj}\ref{it:IdentProj}, and that $\ker p \cong \TheB{2}{1}$, by \eqref{es:CPC}. Then, $\Mod{Q} = \Proj{}[\TheB{2}{1}] \cong \Proj{3}$, again by \cref{prop:ProjInj}\ref{it:IdentProj}. The next step is to identify the cokernel of $\AlgD{\Mod P} \overset{\AlgD{q}}\to \AlgD{\Mod Q}$. By \cref{prop:ARtranspose}\ref{it:ARTr-Proj}, the modules $\AlgD{\Proj 2}$ and $\AlgD{\Proj 3}$ are the (right) modules $\Proj 2^*$ and $\Proj 3^*$, respectively.  Moreover, the morphism $\AlgD{q}$ is non-zero because $\AlgD{q}=0$ would imply that $\ARtransp{\Inje{1}}$ was projective, by part \ref{it:ARTr-Involution} of the same proposition, contradicting part \ref{it:ARTr-Vanishing}. The cokernel $\Coker \AlgD{q}$ is thus easily identified:
\begin{equation}
\ARtransp{\Inje{1}} = \coker \Bigg[
\begin{tikzpicture}[baseline={(l3.base)},scale=1/2]
\node (l2) at (2,2) [] {$\scriptstyle{2}$};
\node (l3) at (1,1) [] {$\scriptstyle{1}$};
\node (l5) at (3,1) [] {$\scriptstyle{3}$};
\node (l4) at (2,0) [] {$\scriptstyle{2}$};
\draw [-] (l2) -- (l3);
\draw [-] (l2) -- (l5);
\draw [-] (l3) -- (l4);
\draw [-] (l5) -- (l4);
\end{tikzpicture}
\overset{\AlgD{q}}\lra
\begin{tikzpicture}[baseline={(l3.base)},scale=1/2]
\node (l2) at (2,2) [] {$\scriptstyle{3}$};
\node (l3) at (1,1) [] {$\scriptstyle{2}$};
\node (l4) at (2,0) [] {$\scriptstyle{3}$};
\draw [-] (l2) -- (l3);
\draw [-] (l3) -- (l4);
\end{tikzpicture}
\Bigg] \cong \Irre{3}^*.
\end{equation}
The \AR{} translation of $\Inje{1}$ is therefore $\ARTau{\Inje{1}}=\VectD{(\ARtransp{\Inje{1}})}\cong\Irre3$.  Replacing $\Inje{1}$ by $\Irre{3}$, a similar computation shows that $\ARTau{\Irre 3}=\VectD{(\ARtransp{\Irre 3})}\cong\Proj 1$. Since $\Proj 1$ is projective, the process started with $\Inje{1}$ stops here. Moreover, since $\Inje{1}$ is injective, we cannot use $\iARTau{}$ to construct any other new indecomposables. If we indicate $\Mod{M}_1=\ARTau{\Mod{M}_2}$
by $\Mod{M}_1\tauarrow \Mod{M}_2$, then these computations may be summarised as the following chain:
\begin{equation} \label{eq:Ex-t0}
\Proj{1} \tauarrow \Irre 3\tauarrow \Inje{1}. \tag{$t_0$}
\end{equation}
(The label \eqref{eq:Ex-t0} will be explained in the next subsection.)

The choice of arrow direction, following the action of $\ARTauSymbol^{-1}$ rather than $\ARTauSymbol$, was made so as to agree with the direction of the irreducible morphisms in the chain \eqref{eq:ARChain}.  Indeed, we may redraw this chain in a zigzag pattern, adding squiggly arrows representing $\ARTauSymbol^{-1}$ as follows:
\begin{equation} \label{eq:UltimateWeave}
\begin{tikzpicture}[baseline={(current bounding box.center)},scale=1.25]
	\node (um2) at (-4,1) {$\PowerARTau{2}{\Mod{U}}$};
	\node (um1) at (-2,1) {$\ARTau{\Mod{U}}$};
	\node (u0) at (-0,1) {$\Mod{U}$};
	\node (u1) at (2,1) {$\PowerARTau{-1}{\Mod{U}}$};
	\node (u2) at (4,1) {$\PowerARTau{-2}{\Mod{U}}$};
	\node at (-4.5,0.5) {$\cdots$};
	\node (dm2) at (-3,0) {$\PowerARTau{2}{\Mod{V}}$};
	\node (dm1) at (-1,0) {$\ARTau{\Mod{V}}$};
	\node (d0) at (1,0) {$\Mod{V}$};
	\node (d1) at (3,0) {$\PowerARTau{-1}{\Mod{V}}$};
	\node (d2) at (5,0) {$\PowerARTau{-2}{\Mod{V}}$};
	\node at (5.5,0.5) {$\cdots$ .};
	\draw[->] (um2) edge (dm2)
	          (dm2) edge (um1)
	          (um1) edge (dm1)
	          (dm1) edge (u0)
	          (u0) edge (d0)
	          (d0) edge (u1)
	          (u1) edge (d1)
	          (d1) edge (u2)
	          (u2) edge (d2);
	\draw[squig] (um2) -- (um1);
	\draw[squig] (um1) -- (u0);
	\draw[squig] (u0) -- (u1);
	\draw[squig] (u1) -- (u2);
	\draw[squig] (dm2) -- (dm1);
	\draw[squig] (dm1) -- (d0);
	\draw[squig] (d0) -- (d1);
	\draw[squig] (d1) -- (d2);
\end{tikzpicture}
\end{equation}
In this way, we see how to combine the $\ARTauSymbol$-orbits of $\Mod{U}$ and $\Mod{V}$, when we have prior knowledge of an irreducible morphism from $\Mod{U}$ to $\Mod{V}$, into a chain of irreducible morphisms.  We will refer to this combination process as \emph{weaving}.

We continue with the identification of the $\ARTauSymbol$-orbits of indecomposables, illustrating the method one more time by computing the repeated action of $\ARTau{}$ on $\Irre2$. The projective cover of $\Irre{2}$ is $\Proj{2}$ and the kernel of the covering map $p$ has head $\Irre{1} \oplus \Irre{3}$ (we know that $\ker p \cong \TheT{1}{2}$, by \eqref{es:DefAV} and \cref{lem:theAandV}, but we are only using the results of \cref{sec:standardProjective} here).  The projective cover of $\ker p$ is therefore $\Proj{1} \oplus \Proj{3}$ (in agreement with \cref{cor:BTProjInj}), so the minimal projective presentation of $\Irre 2$ is
\begin{equation}
\Mod{Q}=
\begin{tikzpicture}[baseline={(t.base)},scale=1/2]
\node (l2) at (1,2) [] {$\scriptstyle{1}$};
\node (l3) at (2,1) [] {$\scriptstyle{2}$};
\draw [-] (l2) -- node (t) {$\vphantom{\scriptstyle 1}$} (l3);
\end{tikzpicture}
\oplus
\begin{tikzpicture}[baseline={(l3.base)},scale=1/2]
\node (l2) at (2,2) [] {$\scriptstyle{3}$};
\node (l3) at (1,1) [] {$\scriptstyle{2}$};
\node (l4) at (2,0) [] {$\scriptstyle{3}$};
\draw [-] (l2) -- (l3);
\draw [-] (l3) -- (l4);
\end{tikzpicture}
\quad \overset{q}{\lra}\quad
\Mod{P}=
\begin{tikzpicture}[baseline={(l3.base)},scale=1/2]
\node (l2) at (2,2) [] {$\scriptstyle{2}$};
\node (l3) at (1,1) [] {$\scriptstyle{1}$};
\node (l5) at (3,1) [] {$\scriptstyle{3}$};
\node (l4) at (2,0) [] {$\scriptstyle{2}$};
\draw [-] (l2) -- (l3);
\draw [-] (l2) -- (l5);
\draw [-] (l3) -- (l4);
\draw [-] (l5) -- (l4);
\end{tikzpicture}
\quad \overset p\lra \quad
\Irre 2
\quad \lra\quad  0 .
\end{equation}
We know that the transpose $\ARtransp{\Irre 2}$ is indecomposable, by \cref{prop:ARtranspose}\ref{it:ARTr-Involution}, and is characterised by the following exact sequence (of right $\Alg{}$-modules):
\begin{equation}
\AlgD{\Mod{P}}=
\begin{tikzpicture}[baseline={(l3.base)},scale=1/2]
\node (l2) at (2,2) [] {$\scriptstyle{2}$};
\node (l3) at (1,1) [] {$\scriptstyle{1}$};
\node (l5) at (3,1) [] {$\scriptstyle{3}$};
\node (l4) at (2,0) [] {$\scriptstyle{2}$};
\draw [-] (l2) -- (l3);
\draw [-] (l2) -- (l5);
\draw [-] (l3) -- (l4);
\draw [-] (l5) -- (l4);
\end{tikzpicture}
\quad \overset{\AlgD{q}}{\lra}\quad
\AlgD{\Mod{Q}}=\begin{tikzpicture}[baseline={(t.base)},scale=1/2]
\node (l2) at (1,2) [] {$\scriptstyle{1}$};
\node (l3) at (2,1) [] {$\scriptstyle{2}$};
\draw [-] (l2) -- node (t) {$\vphantom{\scriptstyle 1}$} (l3);
\end{tikzpicture}
\oplus
\begin{tikzpicture}[baseline={(l3.base)},scale=1/2]
\node (l2) at (2,2) [] {$\scriptstyle{3}$};
\node (l3) at (1,1) [] {$\scriptstyle{2}$};
\node (l4) at (2,0) [] {$\scriptstyle{3}$};
\draw [-] (l2) -- (l3);
\draw [-] (l3) -- (l4);
\end{tikzpicture}
\quad \lra \quad
\ARtransp{\Irre 2}
\quad \lra\quad  0 .
\end{equation}
As before, $\AlgD{q} \neq 0$, so $\ARtransp{\Irre{2}}$ must have precisely three composition factors: $\Irre1^*$ and $\Irre3^*$ in its head and $\Irre2^*$ in its socle.  Indeed, the only other possibility is that $\AlgD{q}$ maps $\head \AlgD{\Mod{P}}$ onto $\soc \Proj{1}^*$, contradicting the indecomposability of $\ARtransp{\Irre{2}}$.  $\ARtransp{\Irre{2}}$ is therefore not one of the indecomposables considered in \cref{sec:standardProjective}. It must be, of course, the (right version of the) module $\TheT12$ introduced in \cref{sec:newFamilies}.  Because taking the (vector space) dual of a module exchanges its socle and head, the dual of the right module version of $\TheT12$ is the left module $\TheB12$ and, thus, $\ARTau{\Irre2}=\VectD{(\ARtransp{\Irre2})}\cong\TheB12$. Iterating the action of $\ARTauSymbol$, we obtain
\begin{equation} \label{eq:Ex-i2}
\cdots \tauarrow \Irre2\tauarrow \TheT12 \tauarrow \TheB12 \tauarrow \Irre2 \tauarrow\cdots \tag{$i_2$}
\end{equation}
Note that the \AR{} translation of $\TheT12$ is the irreducible $\Irre2$ that we started with: this sequence of translated modules forms a cycle. The computations for the translations of $\Irre{1}$ are similar to those detailed above and result in
\begin{equation} \label{eq:Ex-t2}
\cdots \tauarrow\Irre 1\tauarrow \TheT21\tauarrow \TheB21 \tauarrow \Irre 1 \tauarrow \cdots \tag{$t_2$}
\end{equation}

At this point, all of the indecomposable modules known from \cref{sec:standardProjective} (the ``input data") have appeared in either \eqref{eq:Ex-t0}, \eqref{eq:Ex-i2} or \eqref{eq:Ex-t2}, except for those that are both projective and injective:  $\Proj{2}$ and $\Proj{3}$. By \cref{prop:ARtranspose}\ref{it:ARTr-Vanishing}, $\ARTau{\Mod{M}}$ and $\iARTau{\Mod{M}}$ are both the zero module if $\Mod{M}$ is projective and injective. We have thus exhausted the possibility of constructing new indecomposables from the ones we know. It is not clear at this point whether the list of indecomposables that we have constructed,
\begin{equation}
\set{\Irre{1}, \Irre{2}, \Irre{3}, \Proj{1}, \Proj{2}, \Proj{3}, \Inje{1}, \TheB{1}{2}, \TheB{2}{1}, \TheT{1}{2}, \TheT{2}{1}},
\end{equation}
is complete. Proving completeness is the goal of the third step. But first, the irreducible morphisms between the known indecomposables must be counted.

\medskip

\noindent \emph{Irreducible morphisms and weaving}\ ---\ \cref{prop:someIrrMorph}\ref{it:RadSocIrr} gives some irreducible morphisms: including the radical in a projective indecomposable and projecting an injective indecomposable onto the quotient by its socle. In the present case, we obtain six irreducibles morphisms this way, conveniently summarised thus:
\begin{subequations} \label{eq:IrreMorph}
\begin{equation} \label{eq:radsoc}
\Inje1 \lra \Irre2 \lra \Proj1, \qquad \TheT12 \lra \Proj2 \lra \TheB12, \qquad \TheT21 \lra \Proj3 \lra \TheB21.
\end{equation}
When an indecomposable module is projective and injective, with its socle strictly contained in its radical, \cref{prop:someIrrMorph}\ref{it:Rad/SocIrr} gives further irreducible morphisms. Only the $\Proj{i}$, with $i>1$, have these properties.  We compute the quotients $\rad{\Proj2}/\soc{\Proj2}\cong\Irre1\oplus\Irre3$ and $\rad{\Proj3}/\soc{\Proj3}\cong\Irre2$, thereby adding six irreducible morphisms to those of \eqref{eq:radsoc}:
\begin{equation}\label{eq:rad/soc}
\TheT12 \lra \Irre1 \lra \TheB12, \qquad \TheT{2}{1} \lra \Irre{2} \lra \TheB{2}{1}, \qquad \TheT12 \lra \Irre3 \lra \TheB12.
\end{equation}
\end{subequations}
We remark that the decomposability of $\rad{\Proj2}/\soc{\Proj2}$ allowed us to construct four irreducible morphisms for $\Proj{2}$, instead of two, using \cref{prop:leftMinimal}\ref{it:IrrDirSum}.

The non-projective modules $\Inje1$ and $\Irre2$ from the first irreducible morphism of \eqref{eq:radsoc} belong to the distinct translation sequences \eqref{eq:Ex-t0} and \eqref{eq:Ex-i2}, respectively.  Weaving these sequences, as in \eqref{eq:UltimateWeave},
thus gives many other irreducible morphisms:
\begin{equation} \label{eq:FirstWeave}
\begin{tikzpicture}[baseline={(current bounding box.center)},scale=1]
	\node (um2) at (-4,1) {$\Proj{1}$};
	\node (um1) at (-2,1) {$\Irre{3}$};
	\node (u0) at (-0,1) {$\Inje{1}$};
	\node (u1) at (2,1) {$\Proj{1}$};
	\node (u2) at (4,1) {$\Irre{3}$};
	\node (u3) at (6,1) {$\Inje{1}$};
	\node at (-5.5,0.5) {$\cdots$};
	\node (dm3) at (-5,0) {$\Irre{2}$};
	\node (dm2) at (-3,0) {$\TheT{1}{2}$};
	\node (dm1) at (-1,0) {$\TheB{1}{2}$};
	\node (d0) at (1,0) {$\Irre{2}$};
	\node (d1) at (3,0) {$\TheT{1}{2}$};
	\node (d2) at (5,0) {$\TheB{1}{2}$};
	\node (d3) at (7,0) {$\Irre{2}$};
	\node at (7.5,0.5) {$\cdots$ .};
	\draw[->] (dm3) edge (um2)
	          (um2) edge (dm2)
	          (dm2) edge (um1)
	          (um1) edge (dm1)
	          (dm1) edge (u0)
	          (u0) edge (d0)
	          (d0) edge (u1)
	          (u1) edge (d1)
	          (d1) edge (u2)
	          (u2) edge (d2)
	          (d2) edge (u3)
	          (u3) edge (d3);
	\draw[squig] (um2) -- (um1);
	\draw[squig] (um1) -- (u0);
	\draw[squig] (u1) -- (u2);
	\draw[squig] (u2) -- (u3);
	\draw[squig] (dm3) -- (dm2);
	\draw[squig] (dm2) -- (dm1);
	\draw[squig] (dm1) -- (d0);
	\draw[squig] (d0) -- (d1);
	\draw[squig] (d1) -- (d2);
	\draw[squig] (d2) -- (d3);
\end{tikzpicture}
\end{equation}
Here, the top row is the translation chain \eqref{eq:Ex-t0} (repeated) and the bottom row is the translation cycle \eqref{eq:Ex-i2}.  The irreducible morphisms form the following cycle:
\begin{equation} \label{eq:weaving1}
\cdots \lra \Irre{3} \lra \TheB{1}{2} \lra \Inje{1} \lra \Irre{2} \lra \Proj{1} \lra \TheT{1}{2} \lra \Irre{3} \lra \cdots.
\end{equation}
Of these irreducible morphisms, $\Inje{1} \to \Irre{2} \to \Proj{1}$ and $\TheT{1}{2} \to \Irre{3} \to \TheB{1}{2}$ have appeared already in \eqref{eq:IrreMorph}, but $\TheB{1}{2} \to \Inje{1}$ and $\Proj{1} \to \TheT{1}{2}$ are new.

We continue weaving to find further irreducible morphisms.  Of those in \eqref{eq:radsoc}, the first two were analysed above and the remaining four involve an indecomposable that is both projective and injective, so give nothing new.  We therefore turn to the first four morphisms in \eqref{eq:rad/soc}, the last two having also appeared in the previous analysis.  From $\TheT{1}{2} \to \Irre{1}$, we weave the translation cycles \eqref{eq:Ex-i2} and \eqref{eq:Ex-t2} together and arrive at the following cycle of irreducible morphisms:
\begin{equation}\label{eq:weaving2}
\cdots \lra \Irre{2} \lra \TheB21 \lra \TheT12 \lra \Irre1 \lra \TheB12 \lra \TheT21 \lra \Irre2 \lra \cdots.
\end{equation}
These include the remaining morphisms of \eqref{eq:rad/soc} as well as $\TheB{2}{1} \to \TheT{1}{2}$ and $\TheB{1}{2} \to \TheT{2}{1}$, which are new.  Having exhausted our stock of irreducible morphisms, it is reasonable to conjecture that we have found a complete set.  Testing this is the content of the last step.

\medskip

\noindent \emph{The \AR{} quiver and completeness}\ ---\ All indecomposable modules constructed using
\AR{} translation and all irreducible morphisms obtained in the previous step are drawn in the quiver of \cref{fig:quiver}, along with the quivers for orbits of lengths $2$, $4$ and $5$. This quiver combines the morphisms of the cycles \eqref{eq:weaving1} and \eqref{eq:weaving2} with those of the chains in \eqref{eq:radsoc} involving the projective injectives $\Proj2$ and $\Proj3$.

\cref{prop:weaving} gives two sets of checks to perform: whether, for all non-projective indecomposables $\Mod{V}$, every irreducible morphism with target $\Mod{V}$ has been obtained and whether, for all non-injective indecomposables $\Mod{U}$, every irreducible morphism with source $\Mod{U}$ has been obtained. Here is one example of the many verifications required by these checks. The module $\TheB21$ is not projective, its translation is $\ARTau{\TheB21}=\TheT21$, and it is the target of at least two irreducible morphisms: $\Proj3\to \TheB21$ and $\Irre2\to\TheB21$. \cref{prop:weaving} now states that there is an (almost split) exact sequence of the form
\begin{equation}
\dses{\TheT21}{}{\Irre2\oplus \Proj3\oplus \text{?}}{}{\TheB21},
\end{equation}
where ``?'' will be non-zero if and only there are additional independent irreducible morphisms with target $\TheB21$. However, the composition factors of $\TheT21$ and $\TheB21$ together are $\Irre2$ and $\Irre3$, both appearing with multiplicity two, which precisely matches the factors of $\Irre2\oplus\Proj3$. Thus, ``?" is the zero module and a complete set of irreducible morphisms with target $\TheB12$ has been obtained. The remaining verifications are numerous, but the quiver makes them expeditious.

We remark that as $\Proj2$ and $\Proj3$ are both projective and injective, they escape the checks of \cref{prop:weaving}. However, \cref{prop:someIrrMorph}\ref{it:RadSocIrr} assures us that all irreducible morphisms with these modules as source or target have already appeared in \eqref{eq:radsoc}.  Thus, the top right quiver drawn in \cref{fig:quiver} depicts all the irreducible morphisms between those indecomposables that have been identified, thus far.

Our last task is to ascertain whether the list of (isomorphism classes of) indecomposable modules that has been obtained is complete. Since the list of projectives is complete, any missing indecomposable $\Mod{V}$ would have to be non-projective. \cref{prop:weaving} would then state the existence of an almost split exact sequence
\begin{equation}
\dses{\ARTau{\Mod{V}}}{}{\oplus_i \Mod{M}_i}{}{\Mod{V}},
\end{equation}
for some modules $\Mod{M}_i$, and more irreducible morphisms. None of these new morphisms could have any of the indecomposables already found as source or target, because we have already determined that the second quiver in \cref{fig:quiver} already contains all such irreducible morphisms.  A new irreducible morphism would therefore imply that the quiver for this block of $\Alg{}$ is disconnected, contradicting \cref{prop:ARisconnected}\ref{it:QuiverConnected}.  Thus, we have obtained a complete list of isomorphism classes of indecomposables.
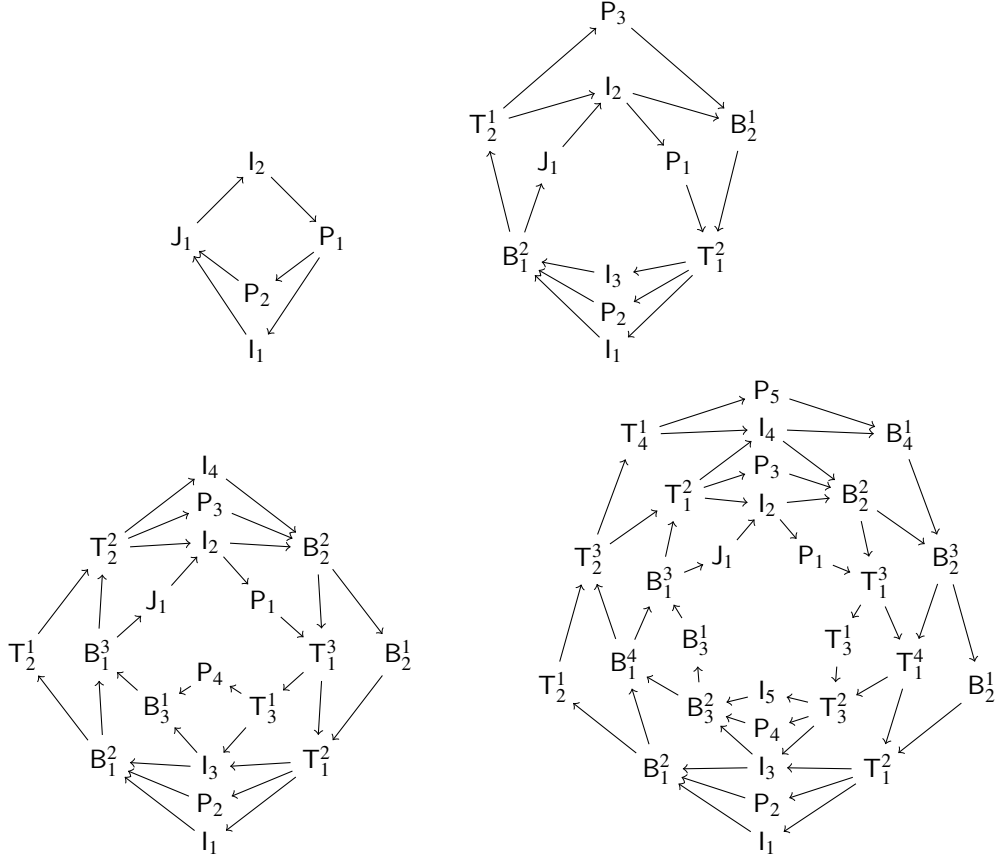
\begin{figure}
\begin{center}
\begin{tikzpicture}[->,every node/.style={circle, inner sep=0pt, outer sep=2pt}]
\node (j1) at (-1,0) [] {$\Inje1$};
\node (p1) at (1,0)  [] {$\Proj1$};
\node (i2) at (0,1)  [] {$\Irre2$};
\node (p2) at (0,-0.75) [] {$\Proj2$};
\node (i1) at (0,-1.5) [] {$\Irre1$};
\draw (p1) edge (p2)
      (p2) edge (j1)
      (j1) edge (i2)
      (i2) edge (p1)
      (p1) edge (i1)
      (i1) edge (j1);
\end{tikzpicture}
\qquad\qquad
%
%
\begin{tikzpicture}[every node/.style={circle, inner sep=0pt, outer sep=2pt}]
\node (u1) at (-0.866,0.5) [] {$\Inje1$};
\node (i3) at (0,-1)       [] {$\Irre3$};
\node (p1) at (0.866,0.5)  [] {$\Proj1$};
\node (i2) at (0.*\rr,1.*\rr)      [] {$\Irre2$};
\node (b12)at (-0.866*\rr,-0.5*\rr)[] {$\TheB12$};
\node (t12)at (0.866*\rr,-0.5*\rr) [] {$\TheT12$};
\node (i1) at (0*\rrr,-1*\rrr)       [] {$\Irre1$};
\node (u2) at (0.866*\rrr,0.5*\rrr)  [] {$\TheB21$};
\node (s2) at (-0.866*\rrr,0.5*\rrr) [] {$\TheT21$};
\node (p2) at (0*\rr,-1*\rr) [] {$\Proj2$};
\node (p3) at (0,1*\rrrr)   [] {$\Proj3$};
\draw[->] (t12) -- (i3);
\draw[->] (i3) -- (b12);
\draw[->] (b12) -- (u1);
\draw[->] (u1) -- (i2);
\draw[->] (i2) -- (p1);
\draw[->] (p1) -- (t12);
\draw[->] (t12) -- (i1);
\draw[->] (i1) -- (b12);
\draw[->] (b12) -- (s2);
\draw[->] (s2) -- (i2);
\draw[->] (i2) -- (u2);
\draw[->] (u2) -- (t12);
\draw[->] (t12) -- (p2);
\draw[->] (p2) -- (b12);
\draw[->] (s2) -- (p3);
\draw[->] (p3) -- (u2);
\end{tikzpicture}\qquad \phantom{m}
\end{center}
%
%
\begin{center}
\begin{tikzpicture}[every node/.style={circle, inner sep=0pt, outer sep=2pt}]
\node (u1) at (-0.7071,0.7071) [] {$\Inje1$};
\node (u3) at (-0.7071,-0.7071) [] {$\TheB31$};
\node (s3) at (0.7071,-0.7071) [] {$\TheT31$};
\node (p1) at (0.7071,0.7071) [] {$\Proj1$};
\node (i2) at (0,1*\rr) [] {$\Irre2$};
\node (b13) at (-1*\rr,0) [] {$\TheB13$};
\node (i3) at (0,-1*\rr) [] {$\Irre3$};
\node (t13) at (1*\rr,0) [] {$\TheT13$};
\node (t22) at (-0.7071*\rrr,0.7071*\rrr) [] {$\TheT22$};
\node (b12) at (-0.7071*\rrr,-0.7071*\rrr) [] {$\TheB12$};
\node (t12) at (0.7071*\rrr,-0.7071*\rrr) [] {$\TheT12$};
\node (b22) at (0.7071*\rrr,0.7071*\rrr) [] {$\TheB22$};
\node (i4) at (0,1*\rrrr) [] {$\Irre4$};
\node (s2) at (-1*\rrrr,0) [] {$\TheT21$};
\node (i1) at (0,-1*\rrrr) [] {$\Irre1$};
\node (u2) at (1*\rrrr,0) [] {$\TheB21$};
\node (p2) at (0,-1*\rrr) [] {$\Proj2$};
\node (p3) at (0,1*\rrr)  [] {$\Proj3$};
\node (p4) at (0,-1*\petitr) [] {$\Proj4$};
\draw[->] (t13) -- (s3);
\draw[->] (s3) -- (i3);
\draw[->] (i3) -- (u3);
\draw[->] (u3) -- (b13);
\draw[->] (b13) -- (u1);
\draw[->] (u1) -- (i2);
\draw[->] (i2) -- (p1);
\draw[->] (p1) -- (t13);
\draw[->] (t13) -- (t12);
\draw[->] (t12) -- (i3);
\draw[->] (i3) -- (b12);
\draw[->] (b12) -- (b13);
\draw[->] (b13) -- (t22);
\draw[->] (t22) -- (i2);
\draw[->] (i2) -- (b22);
\draw[->] (b22) -- (t13);
\draw[->] (u2) -- (t12);
\draw[->] (t12) -- (i1);
\draw[->] (i1) -- (b12);
\draw[->] (b12) -- (s2);
\draw[->] (s2) -- (t22);
\draw[->] (t22) -- (i4);
\draw[->] (i4) -- (b22);
\draw[->] (b22) -- (u2);
\draw[->] (t12) -- (p2);
\draw[->] (p2) -- (b12);
\draw[->] (t22) -- (p3);
\draw[->] (p3) -- (b22);
\draw[->] (s3) -- (p4);
\draw[->] (p4) -- (u3);
\end{tikzpicture}
\qquad\qquad
%
%
\begin{tikzpicture}[every node/.style={circle, inner sep=0pt, outer sep=2pt}]
\node (u1) at (-0.587785, 0.809017)  [] {$\Inje1$};
\node (u3) at (-0.951057, -0.309017) [] {$\TheB31$};
\node (i5) at (0,-1)                 [] {$\Irre5$};
\node (s3) at (0.951057, -0.309017)  [] {$\TheT31$};
\node (p1) at (0.587785, 0.809017)   [] {$\Proj1$};
\node (i2)  at (0., 1.*\rr)                  [] {$\Irre2$};
\node (b13) at (-0.951057*\rr, 0.309017*\rr) [] {$\TheB13$};
\node (b32) at (-0.587785*\rr, -0.809017*\rr)[] {$\TheB32$};
\node (t32) at (0.587785*\rr, -0.809017*\rr) [] {$\TheT32$};
\node (t13) at (0.951057*\rr, 0.309017*\rr)  [] {$\TheT13$};
\node (t22) at (-0.587785*\rrr, 0.809017*\rrr)  [] {$\TheT12$};
\node (b14) at (-0.951057*\rrr, -0.309017*\rrr) [] {$\TheB14$};
\node (i3)  at (0,-1*\rrr)                      [] {$\Irre3$};
\node (t14) at (0.951057*\rrr, -0.309017*\rrr)  [] {$\TheT14$};
\node (b22) at (0.587785*\rrr, 0.809017*\rrr)   [] {$\TheB22$};
\node (i4)  at (0., 1.*\rrrr)                  [] {$\Irre4$};
\node (t23) at (-0.951057*\rrrr, 0.309017*\rrrr) [] {$\TheT23$};
\node (b12) at (-0.587785*\rrrr, -0.809017*\rrrr)[] {$\TheB12$};
\node (t12) at (0.587785*\rrrr, -0.809017*\rrrr) [] {$\TheT12$};
\node (b23) at (0.951057*\rrrr, 0.309017*\rrrr)  [] {$\TheB23$};
\node (s4) at (-0.587785*\rrrrr, 0.809017*\rrrrr)  [] {$\TheT41$};
\node (s2) at (-0.951057*\rrrrr, -0.309017*\rrrrr) [] {$\TheT21$};
\node (i1) at (0,-1*\rrrrr)                      [] {$\Irre1$};
\node (u2) at (0.951057*\rrrrr, -0.309017*\rrrrr)  [] {$\TheB21$};
\node (u4) at (0.587785*\rrrrr, 0.809017*\rrrrr)   [] {$\TheB41$};
\node (p2) at (0,-1.*\rrrr) [] {$\Proj2$};
\node (p3) at (0,1*\rrr) [] {$\Proj3$};
\node (p4) at (0,-1.*\rr) [] {$\Proj4$};
\node (p5) at (0,1*\rrrrr) [] {$\Proj5$};
\draw[->] (t13) -- (s3);
\draw[->] (s3) -- (t32);
\draw[->] (t32) -- (i5);
\draw[->] (i5) -- (b32);
\draw[->] (b32) -- (u3);
\draw[->] (u3) -- (b13);
\draw[->] (b13) -- (u1);
\draw[->] (u1) -- (i2);
\draw[->] (i2) -- (p1);
\draw[->] (p1) -- (t13);
\draw[->] (t13) -- (t14);
\draw[->] (t14) -- (t32);
\draw[->] (t32) -- (i3);
\draw[->] (i3) -- (b32);
\draw[->] (b32) -- (b14);
\draw[->] (b14) -- (b13);
\draw[->] (b13) -- (t22);
\draw[->] (t22) -- (i2);
\draw[->] (i2) -- (b22);
\draw[->] (b22) -- (t13);
\draw[->] (b23) -- (t14);
\draw[->] (t14) -- (t12);
\draw[->] (t12) -- (i3);
\draw[->] (i3) -- (b12);
\draw[->] (b12) -- (b14);
\draw[->] (b14) -- (t23);
\draw[->] (t23) -- (t22);
\draw[->] (t22) -- (i4);
\draw[->] (i4) -- (b22);
\draw[->] (b22) -- (b23);
\draw[->] (b23) -- (u2);
\draw[->] (u2) -- (t12);
\draw[->] (t12) -- (i1);
\draw[->] (i1) -- (b12);
\draw[->] (b12) -- (s2);
\draw[->] (s2) -- (t23);
\draw[->] (t23) -- (s4);
\draw[->] (s4) -- (i4);
\draw[->] (i4) -- (u4);
\draw[->] (u4) -- (b23);
\draw[->] (p2) -- (b12);
\draw[->] (p3) -- (b22);
\draw[->] (p4) -- (b32);
\draw[->] (p5) -- (u4);
\draw[->] (t12) -- (p2);
\draw[->] (t22) -- (p3);
\draw[->] (t32) -- (p4);
\draw[->] (s4) -- (p5);
\end{tikzpicture}
\end{center}
\caption{The \AR{} quivers of $\tl{n}$ and $\dtl{n}$ for orbits $[k]$ of lengths $2$, $3$, $4$ and $5$.  The quivers for $\tl{n}$, with $n$ even and $\beta = 0$, are different and are instead illustrated in \cref{fig:BadQuivers}.  Reflecting about vertical bisectors amounts to taking the twisted dual of \cref{sec:Dual}.}\label{fig:quiver}
\end{figure}

%
%
\subsection{Algorithmic construction of \AR{} quivers for $\tl{n}$ and $\dtl{n}$}

The goal of this last subsection is to present an algorithmic construction of the \AR{} quiver for the algebras $\tl{n}$ and $\dtl{n}$. A detailed proof of this result is long, but does not involve any argument not yet covered in the example of the previous subsection. We shall only indicate how to construct the quiver and sketch the proofs. As usual, we shall assume that the case of $\tl{n}$, with $n$ even and $\beta = 0$, is not under consideration, its study being deferred to the end of the section.

Let $s$ denote the length of the non-critical orbit of $k \in \Lambda_n$.  If $s=1$, then $\Proj{k} = \Inje{k} = \Irre{k}$ and this is the only indecomposable module for the corresponding block.  We shall therefore assume that $s\ge 2$. Order the integers in $[k]$ from the smallest to the largest and label them by $k_1=k_L$, $k_2$, \dots, $k_s=k_R$. Our goal is to construct the \AR{} quiver of the corresponding block.  Its vertices are precisely the indecomposable $\Alg{}$-modules which have the property that each of their composition factors is isomorphic to one of the irreducibles $\Irre 1$, $\Irre 2$, \dots, or $\Irre s$ (as in \cref{sec:ARexample}, we shall replace $k_j$ by $j$ for clarity). We will also find it convenient to write $s=2i$, if $s$ is even, and $s=2i+1$, if it is odd.

\medskip

\noindent\emph{The orbits of the translation $\ARTauSymbol$}\ ---\ As in the previous example, the first step is to characterise the $\ARTauSymbol$-orbits. Recall that the \AR{} translation $\ARTau{\Mod{M}}$ of an indecomposable module $\Mod{M}$ is zero if (and only if) $\Mod{M}$ is projective. The use of the word ``orbit'' for the action of the translation $\ARTauSymbol$ is thus somewhat abusive. We shall use the word regardless, qualifying it as a \emph{$\ARTauSymbol$-orbit} to avoid confusion with the orbit of a non-critical $k$.  We will also describe $\ARTauSymbol$-orbits as chains and cycles, as in \cref{sec:ARexample}.  We note that $\ARTauSymbol$-orbits are disjoint:  no indecomposable appears in more than one.

It is easy to show that there is always a $\ARTauSymbol$-orbit, denoted by $(t_0)$, that takes the form of the following chain:
\begin{equation}\label{t0}
\Proj1=\TheT11
\tauarrow \TheT31
\tauarrow \TheT51
\tauarrow \cdots
\left\{
\begin{matrix}
\TheT{2i-1}1\tauarrow \TheB{2i-1}1 \\
\TheT{2i-1}1\tauarrow\Irre{2i+1}\tauarrow\TheB{2i-1}1
\end{matrix}
\right\}\tauarrow
\cdots\tauarrow
\TheB51\tauarrow
\TheB31\tauarrow
\TheB11=\Inje{1},
\qquad \left\{\begin{matrix}
s\textrm{\ even},\\
s\textrm{\ odd}.\end{matrix}\right.
\tag{$t_0$}
\end{equation}
We recall that $\tauarrow$ denotes the action of $\ARTauSymbol^{-1}$. For $s=3$ ($i=1$), this chain indeed reduces to the chain \eqref{eq:Ex-t0} obtained in the previous subsection. Explicit examples of $\ARTauSymbol$-orbits are given in \cref{app:s6s7}.

The description of the other $\ARTauSymbol$-orbits requires a systematic use of the following ``trick". Let $\TheB kj$, with $k\in\mathbb Z$ and $j$ a positive odd integer, denote the zigzag represented graphically as
\begin{equation}
\begin{tikzpicture}[baseline={(l5.base)},scale=1/3]
\node (l1) at (0,0) [] {$\scriptstyle{k}$};
\node (l2) at (2,2) [] {$\scriptstyle{k+1}$};
\node (l3) at (4,0) [] {$\scriptstyle{k+2}$};
\node (l4) at (6,2) [] {$\phantom{\scriptstyle k+3}$};
\node (l5) at (7,1) [] {$\cdots$};
\node (l6) at (8,2) [] {$\phantom{\scriptstyle j+k-2}$};
\node (l7) at (10,0) [] {$\scriptstyle{j+k-1}$};
\node (l8) at (12,2) [] {$\scriptstyle{j+k}$};
\draw [-] (l1) -- (l2);
\draw [-] (l2) -- (l3);
\draw [-] (l3) -- (l4);
\draw [-] (l6) -- (l7);
\draw [-] (l7) -- (l8);
\end{tikzpicture}
\ .
\end{equation}
Of course, this need not be the Loewy diagram of an indecomposable module as, in general, the integers appearing in this zigzag might be smaller than $1$ or larger than $s$. In order to simplify the following discussion, we shall suppose that $j$ and $k$ satisfy
\begin{equation} \label{eq:Cond}
0 \le j \le 2s+1, \qquad -s \le k \le s, \qquad 1 \le j+k \le 2s+1.
\end{equation}
These conditions will be satisfied in the applications to come.  For all indices $m\le 0$ appearing in this zigzag, construct the pair $(m,r(m))$, where $r(m) =-m$ is the reflection of $m$ through a mirror at $0$. Similarly, for all $m\ge s+1$ in the zigzag, construct the pair $(m,r'(m))$, where $r'(m) = 2(s+1)-m$ is now the reflection of $m$ through a mirror at $s+1$. Finally, delete from the zigzag every integer that appears in one of the pairs $(m,r(m))$ or $(m,r'(m))$ and denote the resulting zigzag by $\rd{\TheB kj}$. The conditions \eqref{eq:Cond} ensure that the result is either empty or that it is the Loewy diagram of one of the indecomposable modules found in \cref{sec:newFamilies}.\footnote{It is clear that this reflection-deletion trick has its origins in a signed action of the affine Weyl group $\grp{A}_1^{(1)} = \ZZ_2 \ltimes \ZZ$ on the integers $\ZZ$ with ``fundamental alcove'' $\set{1, 2, \ldots, s}$ (see \cite{GoodWenzl93}).  This action reflects and translates each zigzag vertex into the fundamental alcove, picking up a formal sign with every reflection.  However, the resulting zigzag may have vertices appearing with negative coefficients in general, ruining its interpretation as the Loewy diagram of a module.  We have imposed the conditions \eqref{eq:Cond} in order to ensure that this does not happen.} In particular, the composition factors of this module have labels in the set $\{1, 2, \dots, s\}$. Using the nomenclature of \cref{sec:newFamilies}, this module can be of type $\TheB{}{}$, $\TheT{}{}$ or $\Irre{}$. We shall call the process of constructing $\rd{\TheB kj}$ from $\TheB kj$ the \emph{reflection-deletion trick}.

Here are two examples of this process. If $s=8$, then $\rd{\TheB{-3}9}$ is $\TheT42$ because the pairs $(0,0)$, $(-1,1)$, $(-2,2)$ and $(-3,3)$ are to be deleted, leaving only the composition factors labelled by $4, 5$ and $6$:
\begin{equation}
\begin{tikzpicture}[baseline={(t.base)},scale=1/3]
\node (l1) at (0,0) [] {$\scriptstyle{-3}$};
\node (l2) at (2,2) [] {$\scriptstyle{-2}$};
\node (l3) at (4,0) [] {$\scriptstyle{-1}$};
\node (l4) at (6,2) [] {$\scriptstyle{0}$};
\node (l5) at (8,0) [] {$\scriptstyle{1}$};
\node (l6) at (10,2) [] {$\scriptstyle{2}$};
\node (l7) at (12,0) [] {$\scriptstyle{3}$};
\node (l8) at (14,2) [] {$\scriptstyle{4}$};
\node (l9) at (16,0) [] {$\scriptstyle{5}$};
\node (l10) at (18,2) [] {$\scriptstyle{6}$};
\draw [-] (l1) -- (l2);
\draw [-] (l2) -- (l3);
\draw [-] (l3) -- (l4);
\draw [-] (l4) -- (l5);
\draw [-] (l5) -- (l6);
\draw [-] (l6) -- (l7);
\draw [-] (l7) -- (l8);
\draw [-] (l8) -- (l9);
\draw [-] (l9) -- node (t) {$\vphantom{\scriptstyle 2}$} (l10);
\draw [dashed,-] (6,2.5) -- (6,-0.5);
\draw [-] (-0.2,0.3) -- (0.6,-0.3);
\draw [-] (3.8,0.3) -- (4.6,-0.3);
\draw [-] (7.6,0.3) -- (8.4,-0.3);
\draw [-] (11.6,0.3) -- (12.4,-0.3);
\draw [-] (5.6,2.3) -- (6.4,1.7);
\draw [-] (1.8,2.3) -- (2.6,1.7);
\draw [-] (9.6,2.3) -- (10.4,1.7);
\end{tikzpicture}
\ .
\end{equation}
(The dashed lines indicate the mirror through which integers have been reflected.)  Similarly, $\rd{\TheB 06}=\TheT12$, if $s=4$, since the pairs to be deleted are now $(0,0)$, $(5,5)$ and $(6,4)$:
\begin{equation}
\begin{tikzpicture}[baseline={(t.base)},scale=1/3]
\node (l1) at (0,0) [] {$\scriptstyle{0}$};
\node (l2) at (2,2) [] {$\scriptstyle{1}$};
\node (l3) at (4,0) [] {$\scriptstyle{2}$};
\node (l4) at (6,2) [] {$\scriptstyle{3}$};
\node (l5) at (8,0) [] {$\scriptstyle{4}$};
\node (l6) at (10,2) [] {$\scriptstyle{5}$};
\node (l7) at (12,0) [] {$\scriptstyle{6}$};
\draw [-] (l1) -- (l2);
\draw [-] (l2) -- (l3);
\draw [-] (l3) -- (l4);
\draw [-] (l4) -- (l5);
\draw [-] (l5) -- (l6);
\draw [-] (l6) -- node (t) {$\vphantom{\scriptstyle 2}$} (l7);
\draw [dashed,-] (0,2.5) -- (0,-0.5);
\draw [dashed,-] (10,2.5) -- (10,-0.5);
\draw [-] (-0.4,0.3) -- (0.4,-0.3);
\draw [-] (11.6,0.3) -- (12.4,-0.3);
\draw [-] (7.6,0.3) -- (8.4,-0.3);
\draw [-] (9.6,2.3) -- (10.4,1.7);
\end{tikzpicture}
\ .
\end{equation}

We now construct the remaining $\ARTauSymbol$-orbits, beginning with those that we denote by \eqref{t2l}, where $1\le l\le \lfloor \frac{1}{2} (s-1) \rfloor$. First, for each $l$ in this range, construct the following sequence of $i$ or $i+1$ zigzags, according as to whether $s$ is even or odd, respectively:
\begin{equation} \label{eq:ZigZagt2l}
\underset{\text{if \(s=2i+1\) is odd}}{\underbrace{\TheB{2i+1-m}{2m+1} \tauarrow}} \TheB{2i-1-m}{2m+1} \tauarrow \cdots \tauarrow \TheB{3-m}{2m+1} \tauarrow \TheB{1-m}{2m+1}.
\end{equation}
Here, we write $m=2l$ for later convenience. Second, apply the reflection-deletion trick to each zigzag in \eqref{eq:ZigZagt2l}. Note that $\rd{\TheB{1-m}{2m+1}}=\TheT m2=\TheT{2l}{2}$ (except when $s$ is odd and $m=s-1$, in which case $\rd{\TheB{1-m}{2m+1}}=\TheT{m}{1}=\TheT{2l}{1}$), explaining our choice of label \eqref{t2l}. Note also that when $s$ is odd, $\rd{\TheB{2i+1-m}{2m+1}}$ is $\Irre{2i+1-m}$. Third, extend the reflected-deleted sequence to the left by adding the (twisted) duals of the $\rd{\TheB{j}{2m+1}}$ in reverse order (if $s$ is odd, the self-dual $\Irre{2i+1-m}$ is not repeated). The complete $\ARTauSymbol$-orbit is then obtained from this extended sequence by closing it to form a cycle:
\begin{equation}\label{t2l}
\begin{tikzpicture}[baseline={(I.base)}]
\node (I) at (-2,1) {$\Irre{2i+1-m}$};
\node (bmax) at (1,2) {$\rd{\TheB{2i-1-m}{2m+1}}$};
\node (bd) at (3.5,2) {$\cdots$};
\node (b3) at (6,2) {$\rd{\TheB{3-m}{2m+1}}$};
\node (b1) at (9,2) {$\rd{\TheB{1-m}{2m+1}}$};
\node (dbmax) at (1,0) {$\twdu{\rd{\TheB{2i-1-m}{2m+1}}}$};
\node (dbd) at (3.5,0) {$\cdots$};
\node (db3) at (6,0) {$\twdu{\rd{\TheB{3-m}{2m+1}}}$};
\node (db1) at (9,0) {$\twdu{\rd{\TheB{1-m}{2m+1}}}$};
\draw[squig] (bmax) -- (bd);
\draw[squig] (bd) -- (b3);
\draw[squig] (b3) -- (b1);
\draw[squig] (b1) -- (db1);
\draw[squig] (db1) -- (db3);
\draw[squig] (db3) -- (dbd);
\draw[squig] (dbd) -- (dbmax);
\draw[squig] (dbmax) -- node[right] {$\begin{smallmatrix} \text{if $s$ is} \\ \text{even} \end{smallmatrix}$} (bmax);
\draw[squig] (dbmax) -- node[below] {$\begin{smallmatrix} \text{is odd} \end{smallmatrix}$} (I);
\draw[squig] (I) -- node[above] {$\begin{smallmatrix} \text{if $s$} \end{smallmatrix}$} (bmax);
\end{tikzpicture}
\ . \tag{$t_{2l}$}
\end{equation}
It contains $s$ distinct indecomposables. Taking $s=3$ and $l=1$, hence $i=1$ and $m=2$, we note that $\rd{\TheB{-1}{5}} = \TheT{1}{2}$, thereby recovering the cycle \eqref{eq:Ex-t2} obtained in the previous subsection.

The remaining $\ARTauSymbol$-orbits will be denoted by $(i_{2l})$, where $1 \le l \le \lfloor \frac{1}{2} s \rfloor$, because they contain the irreducible module $\Irre{2l}$. Their construction is similar to that of the $(t_{2l})$. First, we have the following sequences of $i+1$ zigzags:
\begin{equation}
\TheB{2i+1-m}{2m-1} \tauarrow \dots \tauarrow \TheB{3-m}{2m-1} \tauarrow \TheB{1-m}{2m-1}.
\end{equation}
If $s$ is odd, then $m$ takes the values $m=2l$, for $1\le l\le i$. If it is even, then $m=2i+1-2l$ with, again, $1\le l\le i$. Second, apply the reflection-deletion trick to each zigzag of each sequence. Note that $\rd{\TheB{1-m}{2m-1}} = \Irre{m}$, which is $\Irre{2l}$, if $s$ is odd, and $\Irre{s+1-2l}$, if it is even. Moreover, $\rd{\TheB{2i+1-m}{2m-1}} =\Irre{2l}$, if $s$ is even. Third, we again extend the reflected-deleted sequence to the left by adding the (twisted) duals in reverse order.  When an endpoint of the sequence is irreducible, hence self-dual, it does not get repeated.  The result is the cycle
\begin{equation} \label{i2l}
\begin{tikzpicture}[baseline={(i.base)},yscale=1.25]
\node (bMax) at (-2,1.5) {$\rd{\TheB{2i+1-m}{2m-1}}$};
\node (I) at (1,1) {$\Irre{2l}$};
\node (bmax) at (1,2) {$\rd{\TheB{2i-1-m}{2m-1}}$};
\node (bd) at (3.5,2) {$\cdots$};
\node (b3) at (6,2) {$\rd{\TheB{5-m}{2m-1}}$};
\node (b1) at (9,2) {$\rd{\TheB{3-m}{2m-1}}$};
\node (i) at (11,1) {$\Irre{m}$};
\node (dbMax) at (-2,0.5) {$\twdu{\rd{\TheB{2i+1-m}{2m-1}}}$};
\node (dbmax) at (1,0) {$\twdu{\rd{\TheB{2i-1-m}{2m-1}}}$};
\node (dbd) at (3.5,0) {$\cdots$};
\node (db3) at (6,0) {$\twdu{\rd{\TheB{5-m}{2m-1}}}$};
\node (db1) at (9,0) {$\twdu{\rd{\TheB{3-m}{2m-1}}}$};
\draw[squig] (bmax) -- (bd);
\draw[squig] (bd) -- (b3);
\draw[squig] (b3) -- (b1);
\draw[squig] (b1) -- (i);
\draw[squig] (i) -- (db1);
\draw[squig] (db1) -- (db3);
\draw[squig] (db3) -- (dbd);
\draw[squig] (dbd) -- (dbmax);
\draw[squig] (dbmax) -- node[right] {$\begin{smallmatrix} \text{even} \end{smallmatrix}$} (I);
\draw[squig] (I) -- node[right] {$\begin{smallmatrix} \text{if $s$ is} \end{smallmatrix}$} (bmax);
\draw[squig] (dbmax) -- node[below] {$\begin{smallmatrix} \text{odd} \end{smallmatrix}$} (dbMax);
\draw[squig] (dbMax) -- node[left] {$\begin{smallmatrix} \text{$s$ is} \end{smallmatrix}$} (bMax);
\draw[squig] (bMax) -- node[above] {$\begin{smallmatrix} \text{if} \end{smallmatrix}$} (bmax);
\end{tikzpicture}
\ . \tag{$i_{2l}$}
\end{equation}
It likewise contains $s$ distinct indecomposables.

\begin{lemma}\label{thm:tauOrbits} Let $[k]$ be a non-critical orbit of length $s \ge 2$ (omitting the case of $\Alg{} = \tl{n}$, with $n$ even and $\beta = 0$) and write $s=2i$ or $2i+1$, according as to whether $s$ is even or odd, respectively.  Then, the $s$ $\ARTauSymbol$-orbits
\begin{equation}
(t_0), (t_2), \dots, (t_{2i-2}), \underbrace{(t_{2i}),}_{\mathclap{\text{if \(s\) is odd}}} (i_2), (i_4), \dots, (i_{2i})
\end{equation}
are all cycles, except for $(t_0)$ which is a chain.  They each contain $s$ distinct indecomposable modules and are disjoint in the sense that no indecomposable appears in more than one such $\ARTauSymbol$-orbit.
\end{lemma}
\begin{proof}[Sketch of proof]
First, the \AR{} translation of the modules that are ``far" from the boundary of the class $[k]$ are computed. For a module to be ``far enough", it is sufficient that neither the composition factor $k_L=k_1$ nor $k_R=k_s$ appear in the minimal projective presentation used to compute its \AR{} translation. This computation is then straightforward and does not depend on the parity of $s$. Second, the translation of the modules whose projective presentation involves $k_1$ or $k_s$ is studied. In most cases, the result depends on the parity of $s$. Consequently, the analysis involves many subcases (and is rather tedious). Third, when sufficiently many subcases have been computed, it is straightforward to check that the ``boundary cases'' in the $\ARTauSymbol$-orbits \eqref{t2l} and \eqref{i2l}, with $l >0$, are correctly predicted by the reflection-deletion trick. Finally, the disjointness, number and lengths of the $\ARTauSymbol$-orbits are obtained by inspection. Again, the parity of $s$ plays a role.
\end{proof}
\noindent It will turn out that every indecomposable $\Alg{}$-module, except the projective injectives $\Proj{j}$, $j>1$, will appear in one of the $\ARTauSymbol$-orbits \eqref{t2l} or \eqref{i2l}.  As was noted above, $\Proj{1}$ and $\Inje{1}$ both appear in the chain \eqref{t0}.

\medskip

\noindent\emph{Irreducible morphisms and the weaving of $\ARTauSymbol$-orbits}\ ---\ The second goal is to construct the irreducible morphisms between the indecomposable modules that have been constructed.

As $[k]$ contains $s$ elements, \cref{prop:someIrrMorph}\ref{it:RadSocIrr} gives $2s$ irreducible morphisms, namely $\rad(\Proj j)\ira \Proj j$ and $\Inje j\sra \Inje j/\soc(\Inje j)$, for $1\le j\le s$. Explicitly, these are
\begin{equation}\label{eq:godIrrMorphism}
\Irre2\lra \Proj1, \quad \Inje{1} \lra \Irre 2; \qquad
\TheT{j-1}2 \lra \Proj j \lra \TheB{j-1}2\quad \text{(\(1<j<s\));} \qquad
\TheT{s-1}1 \lra \Proj s \lra \TheB{s-1}1.
\end{equation}
Moreover, since the $\Proj j$, with $j>1$, are non-simple, projective and injective, part \ref{it:Rad/SocIrr} of the same proposition, combined with \cref{prop:leftMinimal}\ref{it:IrrDirSum}, gives another $4s-6$ irreducible morphisms:
\begin{equation}\label{eq:newIrrMorph}
\TheT{j-1}2\lra \Irre{j-1}\lra \TheB{j-1}2, \quad
\TheT{j-1}2\lra \Irre{j+1}\lra \TheB{j-1}2 \quad \text{(\(1<j<s\));} \qquad
\TheT{s-1}1\lra \Irre{s-1}\lra \TheB{s-1}1.
\end{equation}

The next step is to identify pairs of $\ARTauSymbol$-orbits that may be weaved together. This happens when there exists an irreducible morphism $U\to V$ from an indecomposable $U$ of one $\ARTauSymbol$-orbit to an indecomposable $V$ of another. Such pairs are then weaved together as in \eqref{eq:UltimateWeave}.

Suppose first that $s$ is odd.  Then, the sequence $(t_0)$ contains the irreducible $\Irre{2i+1} = \Irre{s}$ and the sequence $(i_2)$ contains the indecomposable $\rd{\TheB{s-2}{3}} = \TheB{s-2}{2}$ and its dual $\TheT{s-2}{2}$.  The irreducible morphisms $\TheT{s-2}{2} \to \Irre{s} \to \TheB{s-2}{2}$ of \eqref{eq:newIrrMorph} therefore allow us to weave $(t_0)$ and $(i_2)$ together.  Moreover, $(t_2)$ contains $\Irre{s-2}$ and \eqref{eq:newIrrMorph} includes $\TheT{s-2}{2} \to \Irre{s-2} \to \TheB{s-2}{2}$, hence we may also weave $(i_2)$ and $(t_2)$ together.  Continuing, we find that we can recursively
weave contiguous pairs in the following sequence, read from left to right:
\begin{equation} \label{eq:Contiguous}
(t_0)\llra (i_2) \llra (t_2) \llra (i_4) \llra (t_4)\llra \dots \llra (i_{2i}) \llra (t_{2i}).
\end{equation}

For $s$ even, we can also recursively weave contiguous pairs from \eqref{eq:Contiguous}, though the final $\ARTauSymbol$-orbit $(t_{2i})$ is omitted.  The justification for this weaving is slightly different to that for $s$ odd because the $(t_{2l})$ no longer contain any irreducibles.  To start, note that the irreducible morphisms $\TheT{s-1}{1} \to \Irre{s-1} \to \TheB{s-1}{1}$ of \eqref{eq:newIrrMorph} allow us to weave $(t_0)$ with $(i_2)$, as before.  To obtain the weave with $(t_2)$, we note that \cref{prop:weaving} implies the exactness of
\begin{equation}
0 \lra
\begin{tikzpicture}[baseline={(t.base)},scale=1/3]
\node (l2) at (2,2) [] {$\scriptstyle{s-3}$};
\node (l3) at (4,0) [] {$\scriptstyle{s-2}$};
\node (l4) at (6,2) [] {$\scriptstyle{s-1}$};
\node (l5) at (8,0) [] {$\scriptstyle{s}$};
\draw [-] (l2) -- node (t) {$\vphantom{\scriptstyle 1}$} (l3);
\draw [-] (l3) -- (l4);
\draw [-] (l4) -- (l5);
\end{tikzpicture}
\lra
\ ?\ \oplus
\begin{tikzpicture}[baseline={(t.base)},scale=1/3]
\node (l4) at (6,2) [] {$\scriptstyle{s-1}$};
\node (l5) at (8,0) [] {$\scriptstyle{s}$};
\draw [-] (l4) -- node (t) {$\vphantom{\scriptstyle 1}$} (l5);
\end{tikzpicture}
\lra \ \Irre{s-1} \ \lra 0,
\end{equation}
for some unknown module $?$, because $\Irre{s-1}$ is not projective, $\ARTau{\Irre{s-1}} \cong \TheT{s-3}{3}$, and there is an irreducible morphism from $\TheT{s-1}{1}$ to $\Irre{s-1}$.  The identity of the unknown module is now determined by the fact that its composition factors are $\Irre{s-3}$, $\Irre{s-2}$, $\Irre{s-1}$ and the fact that $? \oplus \TheT{s-1}{1}$ has a submodule isomorphic to $\TheT{s-3}{4}$: the only possible module is $? \cong \TheT{s-3}{2}$.  But now, \cref{prop:weaving} gives the existence of another irreducible morphism, this time from $\TheT{s-3}{2}$ to $\Irre{s-1}$.  This morphism allows us to weave $(i_2)$ and $(t_2)$.

The iteration for $s$ even now proceeds in the following fashion:  When weaving $(t_{2l}$) with $(i_{2(l+1)})$, the procedure is the same as that described in the $s$ odd case.  However, when weaving $(i_{2l})$ with $(t_{2l})$, the procedure follows that described in the previous paragraph.  We remark that this latter procedure would obviously fail when $s=2$; however, the only $\ARTauSymbol$-orbits in this case are $(t_0)$ and $(i_2)$, so no such weaving is required.

%
%
\begin{lemma}\label{thm:weaving}
The only weaves between $\ARTauSymbol$-orbits are those between contiguous pairs of the sequence
\begin{equation}
(t_0)\llra (i_2) \llra (t_2) \llra (i_4) \llra (t_4)\llra \dots \llra (i_{2i}) \underbrace{\llra (t_{2i})}_{\text{if \(s\) is odd}}.
\end{equation}
Together with the morphisms \eqref{eq:godIrrMorphism}, the morphisms obtained from these weaves form a complete list of all irreducible morphisms between indecomposable modules. The list of indecomposable modules is likewise complete.
\end{lemma}
\begin{proof}[Sketch of proof]
There are two steps to this proof. First, one has to check that the irreducible morphisms described in the above procedures do actually appear in the previously woven pair. Second, one has to check that the tests described after \cref{prop:weaving} are satisfied, for all non-projective and non-injective indecomposable modules.  While both steps are required to conclude that all the irreducible morphisms have been constructed, they each require nothing but patience. As in the example of \cref{sec:ARexample}, completeness follows from \cref{prop:ARisconnected} and the fact that blocks are connected, by definition.  The proof of completeness of the set of indecomposables also follows as in \cref{sec:ARexample}.
\end{proof}

\medskip

\noindent\emph{The complete \AR{} quiver of $\tl{n}$ and $\dtl{n}$}\ ---\ The last step of the algorithm is to put together the quivers for each block of $\Alg{}$.  Recall that a block corresponds to the partition into classes of critical and non-critical integers (see \cref{sub:basics}). Recall that if $k$ is critical or if its (non-critical) orbit has length $1$ ($[k]=\{k\}$), then all its extension groups are trivial and the only indecomposable with $\Irre k$ as a composition factor is $\Irre k$ itself.\footnote{In general, the case $\tl{2}$, with $\beta = 0$, provides the only counterexample to this statement.}

\begin{theorem}
The \AR{} quiver of $\Alg{}=\tl{n}$ or $\dtl{n}$ (omitting $\tl{n}$, when $n$ is even and $\beta=0$) is the disjoint union of $c$ connected graphs where $c$ is the number of distinct classes $[k]$, for $k \in \Lambda_n$.
\begin{enumerate}
\item If a class $[k]$ contains the single integer $k$, then the corresponding connected subgraph consists of a single vertex, labelled by the irreducible $\Irre{k}$, and has no arrows.
\item If a class $[k]$ contains $s\ge 2$ integers, then the connected subgraph associated with this class is constructed by computing the $\ARTauSymbol$-orbits, weaving them, and then adding the injective projectives and their irreducible morphisms, as described in \cref{thm:tauOrbits,thm:weaving}.
\end{enumerate}
\end{theorem}

We have put aside the case of $\Alg{}=\tl{n}$, with $n$ even and $\beta=0$. However, the result for this case is now easily stated. For $\Alg{}=\tl{n}$, with $n$ even and $\beta=0$, every $k \in \Lambda_{n,0}$ is even and belongs to the same class $\{2, 4, \dots, n\}$. In the notation of this section, these integers are replaced by the labels $1, 2, \dots, n/2$. The quiver for this class is then identical to the quiver described above for a class of length $s=n/2$, except for the following changes: $\Proj1$ and $\Inje1$ are replaced by $\TheT11$ and $\TheB11$, respectively (in this case, neither are projective nor injective), a new $\Proj1$ is introduced in the center of the quiver (under $\Irre2$), and the irreducible morphisms $\TheB11\to\Proj1\to\TheT11$ are added. For completeness, we draw the quivers for $n=2$, $4$, $6$ and $8$ in \cref{fig:BadQuivers}.  Interestingly, they coincide with those of a zigzag algebra discussed in \cite{Huerfano}.

\begin{figure}
\hspace{\stretch{1}}
%
%
\begin{tikzpicture}[->,every node/.style={circle, inner sep=0pt, outer sep=2pt}]
\node (p1) at (0,1) [] {$\Proj1$};
\node (i1) at (0,0) [] {$\Irre1$};
\draw[bend left] (p1) edge (i1)
      (i1) edge (p1);
\end{tikzpicture}
%
%
\hspace{\stretch{1}}
\begin{tikzpicture}[->,every node/.style={circle, inner sep=0pt, outer sep=2pt},scale=0.6]
\node (b11) at (180:\rr) [] {$\TheB11$};
\node (t11) at (0:\rr)   [] {$\TheT11$};
\node (p2) at (-90:\rr)  [] {$\Proj2$};
\node (p1) at (90:\rr)   [] {$\Proj1$};
\node (i2) at (90:\rrrr)  [] {$\Irre2$};
\node (i1) at (-90:\rrrr) [] {$\Irre1$};
\draw (t11) edge (p2)
      (p2) edge (b11)
      (b11) edge (i2)
      (i2) edge (t11)
      (t11) edge (i1)
      (i1) edge (b11)
      (b11) edge (p1)
      (p1) edge (t11);
\end{tikzpicture}
\hspace{\stretch{1}}
%
%
\begin{tikzpicture}[every node/.style={circle, inner sep=0pt, outer sep=2pt},scale=0.95]
\node (u1) at (150:\rr)   [] {$\TheB11$};
\node (i3) at (-90:\rr)   [] {$\Irre3$};
\node (p1) at (30:\rr)    [] {$\TheT11$};
\node (np1)at (90:\rr)    [] {$\Proj1$};
\node (b12)at (-150:\rr)  [] {$\TheB12$};
\node (t12)at (-30:\rr)   [] {$\TheT12$};
\node (i2) at (90:\rrr)   [] {$\Irre2$};
\node (p2) at (-90:\rrr)  [] {$\Proj2$};
\node (i1) at (-90:\rrrr) [] {$\Irre1$};
\node (u2) at (30:\rrrr)  [] {$\TheB21$};
\node (s2) at (150:\rrrr) [] {$\TheT21$};
\node (p3) at (90:\rrrr)  [] {$\Proj3$};
\draw[->] (t12) -- (i3);
\draw[->] (i3) -- (b12);
\draw[->] (b12) -- (u1);
\draw[->] (u1) -- (i2);
\draw[->] (i2) -- (p1);
\draw[->] (p1) -- (t12);
\draw[->] (t12) -- (i1);
\draw[->] (i1) -- (b12);
\draw[->] (b12) -- (s2);
\draw[->] (s2) -- (i2);
\draw[->] (i2) -- (u2);
\draw[->] (u2) -- (t12);
\draw[->] (t12) -- (p2);
\draw[->] (p2) -- (b12);
\draw[->] (s2) -- (p3);
\draw[->] (p3) -- (u2);
\draw[->] (u1) -- (np1);
\draw[->] (np1) -- (p1);
\end{tikzpicture}
\hspace{\stretch{1}}
%
%
\begin{tikzpicture}[every node/.style={circle, inner sep=0pt, outer sep=2pt}]
\node (u1) at (135:\rr)     [] {$\TheB11$};
\node (u3) at (-135:\rr)    [] {$\TheB31$};
\node (s3) at (-45:\rr)     [] {$\TheT31$};
\node (p1) at (45:\rr)      [] {$\TheT11$};
\node (b13) at (180:\rr)    [] {$\TheB13$};
\node (t13) at (0:\rr)      [] {$\TheT13$};
\node (np1)at (90:\rr)      [] {$\Proj1$};
\node (p4) at (-90:\rr)     [] {$\Proj4$};
\node (i3) at (-90:\rrr)    [] {$\Irre3$};
\node (i2) at (90:\rrr)     [] {$\Irre2$};
\node (t22) at (135:\rrrr)  [] {$\TheT22$};
\node (b12) at (-135:\rrrr) [] {$\TheB12$};
\node (t12) at (-45:\rrrr)  [] {$\TheT12$};
\node (b22) at (45:\rrrr)   [] {$\TheB22$};
\node (s2) at (180:\rrrr)   [] {$\TheT21$};
\node (u2) at (0:\rrrr)     [] {$\TheB21$};
\node (p2) at (-90:\rrrr)   [] {$\Proj2$};
\node (p3) at (90:\rrrr)    [] {$\Proj3$};
\node (i4) at (90:\rrrrr)   [] {$\Irre4$};
\node (i1) at (-90:\rrrrr)  [] {$\Irre1$};
\draw[->] (t13) -- (s3);
\draw[->] (s3) -- (i3);
\draw[->] (i3) -- (u3);
\draw[->] (u3) -- (b13);
\draw[->] (b13) -- (u1);
\draw[->] (u1) -- (i2);
\draw[->] (i2) -- (p1);
\draw[->] (p1) -- (t13);
\draw[->] (t13) -- (t12);
\draw[->] (t12) -- (i3);
\draw[->] (i3) -- (b12);
\draw[->] (b12) -- (b13);
\draw[->] (b13) -- (t22);
\draw[->] (t22) -- (i2);
\draw[->] (i2) -- (b22);
\draw[->] (b22) -- (t13);
\draw[->] (u2) -- (t12);
\draw[->] (t12) -- (i1);
\draw[->] (i1) -- (b12);
\draw[->] (b12) -- (s2);
\draw[->] (s2) -- (t22);
\draw[->] (t22) -- (i4);
\draw[->] (i4) -- (b22);
\draw[->] (b22) -- (u2);
\draw[->] (t12) -- (p2);
\draw[->] (p2) -- (b12);
\draw[->] (t22) -- (p3);
\draw[->] (p3) -- (b22);
\draw[->] (s3) -- (p4);
\draw[->] (p4) -- (u3);
\draw[->] (u1) -- (np1);
\draw[->] (np1) -- (p1);
\end{tikzpicture}
\hspace{\stretch{1}}
\caption{The \AR{} quivers of $\tl{n}$, with $\beta = 0$, for $n=2$, $4$, $6$ and $8$ (left to right).  Reflecting about vertical bisectors again amounts to taking the twisted dual of \cref{sec:Dual}.} \label{fig:BadQuivers}
\end{figure}
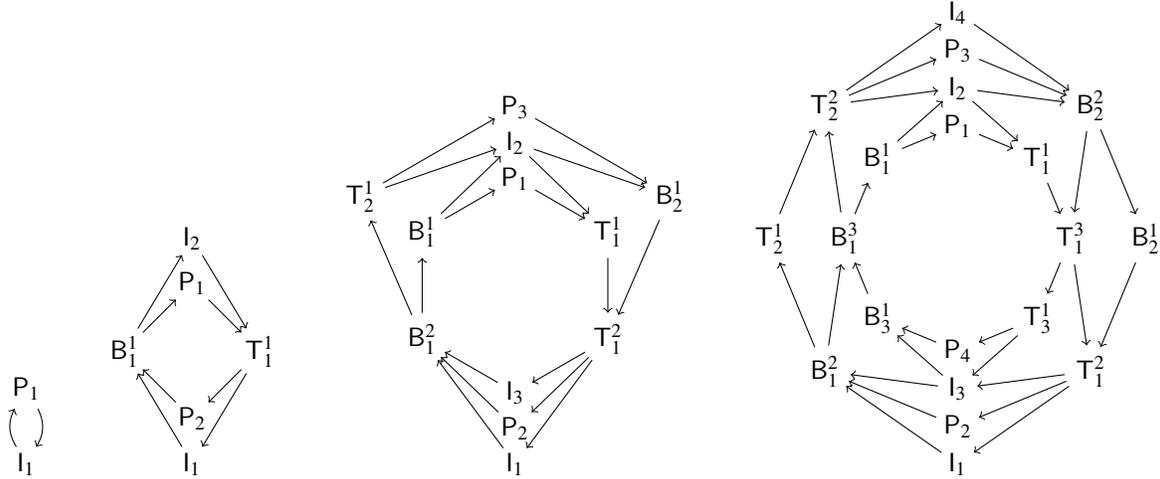

%
%
\newpage
\appendix

\section{Some physical applications of indecomposable Temperley-Lieb modules}\label{app:physics}

As discussed in the Introduction, we shall detail here some of the uses of non-semisimple Temperley-Lieb modules in specific examples of physically important integrable lattice models. First, it is important to recall why non-semisimple modules appear in the first place. In a semisimple algebra, like the complex group algebras of the symmetric groups or the Temperley-Lieb algebras with $\beta$ generic, determining the structure of any given module is equivalent to determining its composition factors. In the non-semisimple case, this is not always so. A simple example is provided by the two-dimensional algebra $\tl 2(\beta)$ with $\beta = 0$. A basis for this algebra is $\{\id, e_1\}$ with unit $\id$ and the relation $e_1\cdot e_1=0$. The regular module, wherein $\tl2$ acts on itself, is given in this basis by
\begin{equation}
\id \mapsto \begin{pmatrix}1 & 0 \\ 0 & 1\end{pmatrix}\qquad\textrm{and}\qquad e_1\mapsto\begin{pmatrix}0 & 0 \\ 1 & 0 \end{pmatrix}.
\end{equation}
Clearly $\mathbb Ce_1$ is a submodule, but it does not have a complement that would be itself a submodule. This representation is indecomposable, even though it has a proper submodule. It has two (isomorphic) composition factors in which $\id$ and $e_1$ act as $(1)$ and $(0)$ respectively. This information is however insufficient to completely characterise the regular module. For non-semisimple algebras, one must also determine how these composition factors are ``glued together'' to form indecomposable summands.  In favourable cases, including indecomposable Temperley-Lieb modules for $\beta = q+q^{-1}$ with $q$ a root of unity, this can be achieved by determining how the composition factors are arranged in the Loewy diagram of the module.

In many physical problems, the analysis proceeds most efficiently and elegantly when one is able to precisely identify which module of the symmetry algebra is defined by the states of the system.  This not only facilitates computations of physically interesting quantities, it may also demonstrate otherwise obscured relationships with other physical problems, perhaps even some with known solutions.  However, if the algebra is not semisimple and the physically relevant module has many composition factors, correctly identifying it through the pattern of its Loewy diagram can be very difficult. Having a complete list of all the patterns that can appear is therefore a very valuable aide that reduces this identification problem to a (hopefully) simple process of elimination. In particular, for the Temperley-Lieb algebras, there is always a finite number of (isomorphism classes of) modules possessing any given finite set of composition factors.  Moreover, these modules are each characterised (up to isomorphism) by a finite number of properties, so we can verify which of these properties are satisfied and thereby perform the identification algorithmically.

\subsection{Example I: the XXZ spin chain}\label{sec:XXZ}

The open XXZ spin chain is a well known physical model on the $n$ site Hilbert space $\CC^{2n} = (\CC^2)^{\otimes n}$ that is described by the Hamiltonian
\begin{equation}
	H_{XXZ} = -\sum_{i=1}^{n-1} e_{i}, \qquad
	e_{i} = -\frac{1}{2}\left(
	\sigma^{x}_{i}\sigma^{x}_{i+1} + \sigma^{y}_{i}\sigma^{y}_{i+1}
+ \frac{q+q^{-1}}{2}(\sigma^{z}_{i}\sigma^{z}_{i+1} - 1_{\mathbb{C}^{\otimes 2n}})
+ \frac{q-q^{-1}}{2}(\sigma^{z}_{i} - \sigma^{z}_{i+1})\right),
\end{equation}
where $\sigma^x_i$, $\sigma^y_i$ and $\sigma^z_i$ denote the Pauli matrices acting on the $i$-th copy of $\mathbb{C}^{2}$ while acting as the identity on the remaining $n-1$ copies. One can verify using the standard Pauli relations that the $e_{i}$ generate a representation of $\tl{n}(\beta)$ with $\beta = q+q^{-1}$.  According to the general strategy outlined above, we should ask how to decompose it into indecomposable Temperley-Lieb modules.

Martin \cite{MartinSchurWeyl} was the first to tackle this problem. More recently, Gainutdinov and Vasseur \cite{GainutdinovVasseur} used their lattice fusion product to determine the explicit decomposition and Provencher and Saint-Aubin \cite{Provencher} confirmed it by computing the projections onto each indecomposable summand. If $\ell$ denotes the smallest positive integer such that $q^{2 \ell} = 1$, the result is
\begin{equation}
	\mathbb{C}^{\otimes 2n} \simeq \overset{n}{\underset{k=0}{\bigoplus}} \Mod{M}(k,n-k),
\end{equation}
where we set, for $r$ and $s$ non-negative integers and $i$ and $j$ integers satisfying $0 \le i,j < \ell$,
\begin{align}
	\Mod{M}(r \ell - 1 + i,s \ell - 1 + j) &\simeq
	\bigoplus_{\substack{t=\abs{r-s}+1 \\ \text{step}=2}}^{r+s-1} \ \bigoplus_{\substack{k=\mu(\ell-\abs{i-j}-1) \\ \text{step}=2}}^{\ell-\abs{i-j}-1} \Proj{t \ell - 1 + k}
\oplus
	\bigoplus_{\substack{t=\abs{r-s-\sign(i-j)}+1 \\ \text{step}=2}}^{r+s} \ \bigoplus_{\substack{k=\mu(\abs{i-j}-1) \\ \text{step}=2}}^{\abs{i-j}-1} \Proj{t \ell - 1 + k} \notag \\
	& \qquad \oplus \bigoplus_{\substack{k=\mu(i+j-\ell-1) \\ \text{step}=2}}^{i+j-\ell-1} \Proj{(r+s) \ell - 1 + k}
	\oplus \bigoplus_{\substack{k=\abs{i-j}+1 \\ \text{step}=2}}^{\ell-\abs{\ell-i-j}-1} \Stan{(r+s)\ell - 1 + k},
\end{align}
with $\mu(x) \in \set{0,1}$ defined to agree with $x$ mod $2$.  Here, it is understood that sums over empty sets correspond to the zero module.  We give a few examples for $n = 8$:
\begin{equation}
	\begin{aligned}
		(\ell &= 4 & &\text{so} & \beta &= \pm\sqrt{2}) &&&&& \mathbb{C}^{\otimes 16} &\simeq 6 \Proj{8} \oplus \Proj{6} \oplus 3\Proj{4} \oplus 3\Stan{8} \oplus 2\Stan{4}, \\
		(\ell &= 3 & &\text{so} & \beta &= \pm1) &&&&& \mathbb{C}^{\otimes 16} &\simeq 9 \Proj{8} \oplus 4 \Proj{6} \oplus \Proj{4} \oplus 3\Proj{2} \oplus 3\Stan{6}, \\
		(\ell &= 2 & &\text{so} & \beta &= 0) &&&&& \mathbb{C}^{\otimes 16} &\simeq 4 \Proj{8} \oplus 3 \Proj{6} \oplus 2 \Proj{4} \oplus \Proj{2} \oplus 5\Stan{8}.
	\end{aligned}
\end{equation}
Note that many indecomposable representations appear multiple times. The number of times that each module appears is the dimension of an irreducible module of the quantum group $U_q(\mathfrak{sl}_{2})$, a manifestation of the so-called \emph{quantum Schur-Weyl duality}.

\subsection{Example II: the Dimer model}

Whereas the XXZ model decomposes into standard and projective modules, the decomposition of the state space of our next example requires the less familiar indecomposables introduced in \cref{sec:newFamilies}.  The dimer model, see \cite{LieSol67} for example, is another physical model that may be formulated \cite{dimer2016} in terms of a spin chain, this time on $n-1$ sites, with Hamiltonian
\begin{equation}
	H_{\text{dimer}} = - \sum_{j=1}^{n-1}e_{j}, \qquad
	e_{j} = \sigma^{-}_{j-1}\sigma^{+}_{j} + \sigma^{+}_{j}\sigma^{-}_{j+1}, \qquad 2\sigma^{\pm}_{j} \equiv \sigma^{x}_{j} \mp i \sigma^{y}_{j}, \quad \sigma^{\pm}_{0} \equiv \sigma^{\pm}_{n} \equiv 0.
\end{equation}
Here, the $\sigma^{k}_{j}$ are defined as for the XXZ example above. One can again verify that the $e_{j}$ generate a representation of the Temperley-Lieb algebra $\tl{n}(\beta)$, this time with $\beta = 0$. It was proven in \cite{dimer2016} that this spin chain representation decomposes as
\begin{equation}
\mathbb{C}^{\otimes 2(n-1)} \simeq \underset{\nu = -\frac{1}{2} (n-1)}{\overset{\frac{1}{2} (n-1)}{\bigoplus}}\Mod{E}^{\nu}_{n},
\end{equation}
where one sets
\begin{align}\label{eq:dimer.structure}
	\Mod{E}^{\nu}_{n} \simeq
	\begin{cases}
		\TheB{2\nu +1}{\frac{1}{2} (n-1)-\nu} & \text{ if } \nu \ge \frac{1}{2}, \\
		\TheT{-2\nu +1}{\frac{1}{2} (n-1)+\nu} & \text{ if } \nu \le -\frac{1}{2},
	\end{cases},
\intertext{if $n$ is even, and}
	\Mod{E}^{\nu}_{n} \simeq \bigoplus_{i=0}^{\left\lfloor \frac{1}{4} (n-1-2 \abs{\nu}) \right\rfloor} \Proj{2 \abs{\nu} +1 + 4i},
\end{align}
if $n$ is odd.  For example, we have the following decompositions:
\begin{equation}
	\begin{aligned}
		(n &= 8) &&&&& \mathbb{C}^{\otimes 14} &\simeq  \Stan{6} \oplus \TheT{4}{2} \oplus \TheT{2}{3}  \oplus \TheB{2}{3} \oplus \TheB{4}{2} \oplus \Cost{6} \oplus 2 \Irre{8}, \\
		(n &= 9) &&&&& \mathbb{C}^{\otimes 16} &\simeq 5 \Proj{9} \oplus 4 \Proj{7} \oplus 3 \Proj{5} \oplus 2 \Proj{3} \oplus \Proj{1}, \\
		(n &= 10) &&&&& \mathbb{C}^{\otimes 18} &\simeq \Stan{8} \oplus \TheT{6}{2} \oplus \TheT{4}{3} \oplus \TheT{2}{4} \oplus \TheB{2}{4} \oplus \TheB{4}{3} \oplus \TheB{6}{2} \oplus \Cost{8} \oplus 2 \Irre{10}.
	\end{aligned}
\end{equation}

Note that this model can also be seen has a module for another algebra: the \emph{type B Temperley-Lieb} algebra\footnote{This name comes from the fact that it can be obtained from the type $B_{n}$ Hecke algebra in the same way that regular Temperley-Lieb is obtained from the type $A_{n}$ Hecke algebra.} $\tilde{\mathsf{b}}_{n}(q,\delta)$, also called the \emph{one-boundary Temperley-Lieb} algebra. This algebra is very similar to $\tl{n}(\beta)$ except that it has another generator $b$ and another parameter $\delta$ that satisfy the following relations:
\begin{equation}
	e_{1} b e_{1} = (q \delta^{-1} - q^{-1}\delta)e_{1},
   \qquad b^{2} = (\delta - \delta^{-1})b, \qquad
	e_{i}b = be_{i} \quad \text{for all}\ i \ge 2.
\end{equation}
In the dimer model, one simply considers the same model on $n$ sites (instead of $n-1$), defining $b = 2\ii e_{1}$ and $u_{i} = e_{i+1}$ for $i=1, \hdots, n-1$. This produces a representation of the one-boundary Temperley-Lieb algebra with $\delta = 1$ and $q = \ii$.  If one ``forgets'' the boundary operator $b$, meaning that we restrict to the $\tl{n+1}(0)$ subalgebra generated by the $u_i$, then one can verify that combining the formulae in \cref{sec:CBTRes} with \cref{eq:dimer.structure} results in the ordinary XXZ spin chain at $q = \ii$. From a physical point of view, this shows that the dimer model on $n$ sites is isomorphic to an XXZ spin chain on which one imposes somewhat unusual boundary conditions by adding the operator $b$.

\subsection{Example III: the twisted XXZ spin chain}
Our final example discusses a well known physical model for which the Temperley-Lieb symmetry is augmented to a significantly larger associative algebra.  The twisted XXZ spin chain is a closed version of the open spin chain discussed in section \ref{sec:XXZ}. Its Hilbert space is also isomorphic to $\mathbb{C}^{\otimes 2n}$, but its Hamiltonian is given by
\begin{equation}
	H_{XXZ}^{\text{tw.}} = - \sum_{i=1}^{n} e_{i},
\end{equation}
where this time the $e_i$ are defined by
\begin{equation}
	e_{i} =
	\begin{dcases*}
		-\frac{1}{2}\left( \sigma^{x}_{i}\sigma^{x}_{i+1} + \sigma^{y}_{i}\sigma^{y}_{i+1}
+ \frac{q+q^{-1}}{2}(\sigma^{z}_{i}\sigma^{z}_{i+1} - 1_{\mathbb{C}^{\otimes 2n}})
+ \frac{q-q^{-1}}{2}(\sigma^{z}_{i} - \sigma^{z}_{i+1})\right), & for $i < n$, \\
		-\frac{1}{2} \ee^{-\ii \phi \sigma^{z}_{1} / 2} \left( \sigma^{x}_{n}\sigma^{x}_{1} + \sigma^{y}_{n}\sigma^{y}_{1}
 + \frac{q+q^{-1}}{2}(\sigma^{z}_{n}\sigma^{z}_{1} - 1_{\mathbb{C}^{\otimes 2n}})
+ \frac{q-q^{-1}}{2}(\sigma^{z}_{n} - \sigma^{z}_{1})\right) \ee^{\ii \phi \sigma^{z}_{1} / 2}, & for $i=n$.
	\end{dcases*}
\end{equation}
Here, $\phi$ is a twist parameter and, again, the $\sigma^{k}_{i}$ are defined as in the open case. One can show that this defines a representation of the \emph{affine Temperley-Lieb} algebra $\mathsf{aTL}_{n}$ \cite{GL-Aff}. The defining relations of this algebra are
\begin{equation}
	\begin{gathered}
		e_{i}e_{i} = (q+q^{-1})e_{i}, \qquad e_{i}e_{i\pm 1}e_{i} = e_{i}, \qquad e_{i}e_{j} = e_{j} e_{i} \quad \text{if}\ \abs{i-j}>1, \\
		u e_{i} = e_{i+1}u , \qquad u^{2}e_{n-1} = e_{1} e_{2} \hdots e_{n-1},
	\end{gathered}
\end{equation}
where we understand
that $e_{0} \equiv e_{n}$, $e_{n+1} \equiv e_{1}$, and so on.  In this representation, the generator $u$ is given by
\begin{equation}
	u = \ii^{n} \ee^{-\ii \phi \sigma^{z}_{1} / 2} s_{1}s_{2}\cdots s_{n}, \qquad s_{i} \equiv \frac{1}{2}(1_{\mathbb{C}^{\otimes 2n}} + \sigma^{x}_{i}\sigma^{x}_{i+1} +\sigma^{y}_{i}\sigma^{y}_{i+1}+\sigma^{z}_{i}\sigma^{z}_{i+1}).
\end{equation}

The representation theory of the affine Temperley-Lieb algebra is significantly more complicated than in the regular case, so we simply mention that its standard modules $\Mod{W}_{k,z}$ and its irreducible modules $\chi_{k,z}$ are indexed by two parameters: an integer $0 \le k \le n$ and a non-zero complex parameter $z$. If $q$ is generic, then one can show \cite{braidtrans} that the Hilbert space decomposes as
\begin{equation} \label{eq:twXXZdecomp}
	\mathbb{C}^{\otimes 2n} \simeq \bigoplus_{\substack{j=-n \\ \text{step}=2}}^n	\Mod{M}_{j,z},
\end{equation}
where $z =(-1)^{n} \ee^{-\ii \phi / 2}$ and $\Mod{M}_{j,z}$ is the standard module $\Mod{W}_{j,z}$, if $j \ge 0 $, and is its dual $\dual{\Mod{W}}_{j,z}$, if $j < 0$.  If $\ee^{-\ii \phi}$ is not an integer power of $q$, then the model is semisimple; if it is an integer power of $q$, then each standard module has, in general, one or two composition factors. If $q$ is a root of unity, the decomposition \eqref{eq:twXXZdecomp} still holds but the identification of the $\Mod{M}_{j,z}$ is presently unknown in general.  In this case, the structure of the standard modules is significantly more complicated, so it is harder to identify them or their various subquotients.  For instance, \cref{fig:atlreg} shows the Loewy diagrams for some standard modules with $n = 16$, and $\ee^{-\ii \phi/2} = q^{2} $; in contrast,  \cref{fig:atl} shows the Loewy diagrams of the same modules if $\ell = 3$ (so $q^{6} = 1$).
 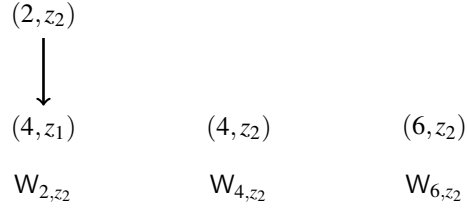
\begin{figure}
 	\begin{tikzpicture}[scale = 1/2, ->]
 		\node (top) at (0,6) {$(2, z_{2})$};
 		\node (bottom) at (0,3) {$(4, z_{1})$};
 		\draw[line width = 1pt, black]  (top) -- (bottom);
 		\node[anchor = north] (bottomlabel) at (0,2) {$\Mod{W}_{2,z_{2}}$};
 	\end{tikzpicture} \qquad \qquad
 	\begin{tikzpicture}[scale = 1/2, ->]
 		\node (top) at (0,3) {$(4,z_{2})$};
 		\node[anchor = north] (bottomlabel) at (0,2) {$\Mod{W}_{4,z_{2}}$};
 	\end{tikzpicture} \qquad \qquad
 	\begin{tikzpicture}[scale = 1/2, ->]
 		\node (top) at (0,3) {$(6,z_{2})$};
 		\node[anchor = north] (bottomlabel) at (0,2) {$\Mod{W}_{6,z_{2}}$};
 	\end{tikzpicture}
 	\caption{ Some Loewy diagrams for the standard modules $\Mod{W}_{j,z}$ for $n=16$ and $q$ generic. The notation $(j,z)$ denotes the composition factor isomorphic to the simple quotient of $\Mod{W}_{j,z}$.  We set $z_k = (-q)^k$ for convenience.}\label{fig:atlreg}
 \end{figure}
 \begin{figure}
 	\begin{tikzpicture}[scale = 1/2,->]
		\node (top) at (0,12) {$(2,z_{2})$}; 
		\node (left1) at (-3,9) {$(8,z_{-1})$};
		\node (right1) at (3,9) {$(4,z_{1})$};
		\node (left2) at (-3,6) {$(10,z_{-2}$)};
		\node (right2) at (3,6) {$(14,z_{2})$};
		\node (bottom) at (0,3) {$(16,z_{1})$};
		\draw[line width = 1pt, black]  (top) -- (left1);
		\draw[line width = 1pt, black]  (top) -- (right1);
		\draw[line width = 1pt, black] (left1) -- (right2);
		\draw[line width = 1pt, black] (left1) -- (left2);
		\draw[line width = 1pt, black] (right1) -- (right2);
		\draw[line width = 3pt, white] (right1) -- (left2);
		\draw[line width = 1pt, black] (right1) -- (left2);
		\draw[line width = 1pt, black] (right2) -- (bottom);
		\draw[line width = 1pt, black] (left2) -- (bottom);
		\node[anchor = north] (bottomlabel) at (0,2) {$\Mod{W}_{2,z_{2}}$};
	\end{tikzpicture}\qquad\qquad
	\begin{tikzpicture}[scale = 1/2,->]
		\node (top) at (0,12) {$(4,z_{2})$};
		\node (mid1) at (0,9) {$(10,z_{-1})$};	
		\node (mid2) at (0,6) {$(14,z_{1})$};	
		\node (bottom) at (0,3) {$(16,z_{2})$};	
		\draw[line width = 1pt, black]  (top) -- (mid1);
		\draw[line width = 1pt, black]  (mid1) -- (mid2);
		\draw[line width = 1pt, black]  (mid2) -- (bottom);
		\node[anchor = north] (bottomlabel) at (0,2) {$\Mod{W}_{4,z_{2}}$};
	\end{tikzpicture}\qquad\qquad
		\begin{tikzpicture}[scale = 1/2,->]
		\node (top) at (0,12) {$(6,z_{2})$};
		\node (left1) at (-3,9) {$(8, z_{3})$};		
		\node (right1) at (3,9) {$(10, z_{0})$};	
		\node (left2) at (-3,6) {$(12, z_{-2}$)};	
		\node (right2) at (3,6) {$(12, z_{2})$};	
		\node (left3) at (-3,3) {$(14, z_{3}$)};	
		\node (right3) at (3,3) {$(16, z_{0})$};	
		\draw[line width = 1pt, black]  (top) -- (left1);
		\draw[line width = 1pt, black]  (top) -- (right1);
		\draw[line width = 1pt, black] (left1) -- (right2);
		\draw[line width = 1pt, black] (left1) -- (left2);
		\draw[line width = 1pt, black] (right1) -- (right2);
		\draw[line width = 3pt, white] (right1) -- (left2);
		\draw[line width = 1pt, black] (right1) -- (left2);
		\draw[line width = 1pt, black] (left2) -- (left3);
		\draw[line width = 1pt, black] (left2) -- (right3);
		\draw[line width = 1pt, black] (right2) -- (right3);
		\draw[line width = 3pt, white] (right2) -- (left3);
		\draw[line width = 1pt, black] (right2) -- (left3);
		\node[anchor = north] (bottomlabel) at (0,2) {$\Mod{W}_{6,z}$};
	\end{tikzpicture}
	\caption{ Some Loewy diagrams for the standard modules $\Mod{W}_{j,z}$ for $n=16$ and $\ell = 3 $ (so $q^{6}= 1$). The notation $(j,z)$ denotes the composition factor isomorphic to the simple quotient of $\Mod{W}_{j,z}$.  We set $z_k = (-q)^k$ for convenience.}\label{fig:atl}
 \end{figure}
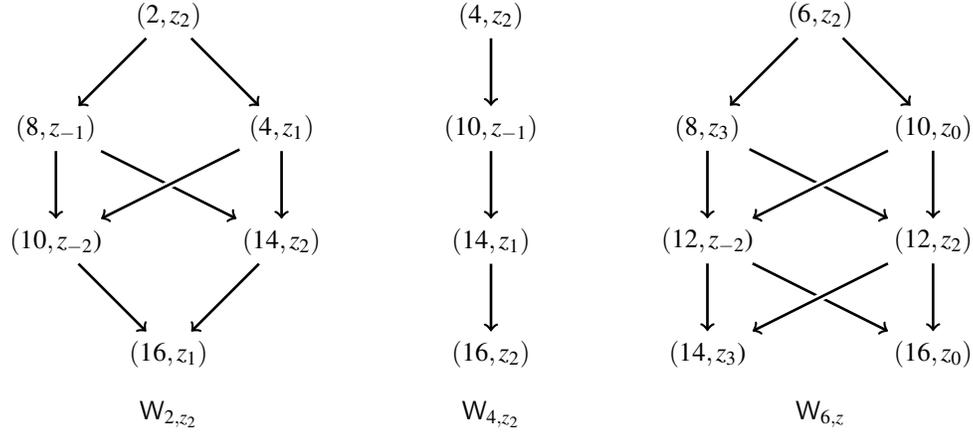

%
%
\section{$\ARTauSymbol$-orbits for $s=6$ and $7$}\label{app:s6s7}

The $\ARTauSymbol$-orbits for $s=6$ are
\begin{equation}
\begin{aligned}
&\Proj1 \tauarrow \TheT31 \tauarrow \TheT51 \tauarrow \TheB51 \tauarrow \TheB31 \tauarrow \Inje1, & &(t_0) \\
&\TheB22 \tauarrow \TheT15 \tauarrow \TheT32 \tauarrow \TheB32 \tauarrow \TheB15 \tauarrow \TheT22, & &(t_2) \\
&\TheB12 \tauarrow \TheT23 \tauarrow \TheT42 \tauarrow \TheB42 \tauarrow \TheB23 \tauarrow \TheT12, & &(t_4)
\end{aligned}
\qquad \qquad
\begin{aligned}
&\TheT13 \tauarrow \TheT33 \tauarrow \Irre5 \tauarrow \TheB33 \tauarrow \TheB13 \tauarrow \Irre2, & &(i_2) \\
&\TheB14 \tauarrow \TheT24 \tauarrow \Irre4 \tauarrow \TheB24 \tauarrow \TheT14 \tauarrow \Irre3, & &(i_4) \\
&\TheT21 \tauarrow \TheT41 \tauarrow \Irre6 \tauarrow \TheB41 \tauarrow \TheB21 \tauarrow \Irre1. & &(i_6)
\end{aligned}
\end{equation}
All $\ARTauSymbol$-orbits are cycles except $(t_0)$, that is, the $\ARTauSymbol$-translation of the leftmost module in each orbit is the rightmost one. Their weaves are as follows:
\begin{align}
\Proj1 \lra \TheT13 &\lra \TheT31 \lra \TheT33 \lra \TheT51 \lra \Irre5  \lra \TheB51 \lra \TheB33 \lra \TheB31 \lra \TheB13 \lra \Inje1 \lra \Irre2  \lra \Proj1, \tag{$(t_0)\leftrightarrow(i_2)$} \\
\TheB22 \lra \TheT13 &\lra \TheT15 \lra \TheT33 \lra \TheT32 \lra \Irre5  \lra \TheB32 \lra \TheB33 \lra \TheB15 \lra \TheB13 \lra \TheT22 \lra \Irre2  \lra \TheB22, \tag{$(i_2)\leftrightarrow(t_2)$}\\
\TheB22 \lra \TheB24 &\lra \TheT15 \lra \TheT14 \lra \TheT32 \lra \Irre3  \lra \TheB32 \lra \TheB14 \lra \TheB15 \lra \TheT24 \lra \TheT22 \lra \Irre4  \lra \TheB22, \tag{$(t_2)\leftrightarrow(i_4)$}\\
\TheB42 \lra \TheB24 &\lra \TheB23 \lra \TheT14 \lra \TheT12 \lra \Irre3  \lra \TheB12 \lra \TheB14 \lra \TheT23 \lra \TheT24 \lra \TheT42 \lra \Irre4  \lra \TheB42, \tag{$(i_4)\leftrightarrow(t_4)$}\\
\TheB42 \lra \TheB41 &\lra \TheB23 \lra \TheB21 \lra \TheT12 \lra \Irre1  \lra \TheB12 \lra \TheT21 \lra \TheT23 \lra \TheT41 \lra \TheT42 \lra \Irre6 \lra \TheB42. \tag{$(t_4)\leftrightarrow(i_6)$}
\end{align}

\bigskip

The $\ARTauSymbol$-orbits for $s=7$ are
\begin{equation}
\begin{aligned}
\Proj1 &\tauarrow \TheT31 \tauarrow \TheT51 \tauarrow \Irre7 \tauarrow \TheB51 \tauarrow \TheB31 \tauarrow \Inje1, & &(t_0) \\
\TheB22 &\tauarrow \TheT15 \tauarrow \TheT34 \tauarrow \Irre5 \tauarrow \TheB34 \tauarrow \TheB15 \tauarrow \TheT22, & &(t_2) \\
\TheB42 &\tauarrow \TheB25 \tauarrow \TheT14 \tauarrow \Irre3 \tauarrow \TheB14 \tauarrow \TheT25 \tauarrow \TheT42, & &(t_4) \\
\TheB61 &\tauarrow \TheB41 \tauarrow \TheB21 \tauarrow \Irre1 \tauarrow \TheT21 \tauarrow \TheT41 \tauarrow \TheT61, & &(t_6)
\end{aligned}
\qquad \qquad
\begin{aligned}
&\TheT13 \tauarrow \TheT33 \tauarrow \TheT52 \tauarrow \TheB52 \tauarrow \TheB33 \tauarrow \TheB13 \tauarrow \Irre2, & &(i_2) \\
&\TheB24 \tauarrow \TheT16 \tauarrow \TheT32 \tauarrow \TheB32 \tauarrow \TheB16 \tauarrow \TheT24 \tauarrow \Irre4, & &(i_4) \\
&\TheB43 \tauarrow \TheB23 \tauarrow \TheT12 \tauarrow \TheB12 \tauarrow \TheT23 \tauarrow \TheT43 \tauarrow \Irre6. & &(i_6)
\end{aligned}
\end{equation}
Again, all $\ARTauSymbol$-orbits are cyclic except $(t_0)$. Their weaves are:
\begin{align}
\Proj1 \ra \TheT13 &\ra \TheT31 \ra \TheT33 \ra \TheT51 \ra \TheT52 \ra \Irre7 \ra \TheB52  \ra \TheB51 \ra \TheB33 \ra \TheB31 \ra \TheB13 \ra \TheB11 \ra \Irre2  \ra \Proj1, \tag{$(t_0)\leftrightarrow(i_2)$} \\
\TheB22 \ra \TheT13 &\ra \TheT15 \ra \TheT33 \ra \TheT34 \ra \TheT52 \ra \Irre5  \ra \TheB52 \ra \TheB34 \ra \TheB33 \ra \TheB15 \ra \TheB13 \ra \TheT22 \ra \Irre2  \ra \TheB22, \tag{$(i_2)\leftrightarrow(t_2)$}\\
\TheB22 \ra \TheB24 &\ra \TheT15 \ra \TheT16 \ra \TheT34 \ra \TheT32 \ra \Irre5  \ra \TheB32 \ra \TheB34 \ra \TheB16 \ra \TheB15 \ra \TheT24 \ra \TheT22 \ra \Irre4  \ra \TheB22, \tag{$(t_2)\leftrightarrow(i_4)$}\\
\TheB42 \ra \TheB24 &\ra \TheB25 \ra \TheT16 \ra \TheT14 \ra \TheT32 \ra \Irre3  \ra \TheB32 \ra \TheB14 \ra \TheB16 \ra \TheT25 \ra \TheT24 \ra \TheT42 \ra \Irre4  \ra \TheB42, \tag{$(i_4)\leftrightarrow(t_4)$}\\
\TheB42 \ra \TheB43 &\ra \TheB25 \ra \TheB23 \ra \TheT14 \ra \TheT12 \ra \Irre3  \ra \TheB12 \ra \TheB14 \ra \TheT23 \ra \TheT25 \ra \TheT43 \ra \TheT42 \ra \Irre6 \ra \TheB42, \tag{$(t_4)\leftrightarrow(i_6)$}\\
\TheB61 \ra \TheB43 &\ra \TheB41 \ra \TheB23 \ra \TheB21 \ra \TheT12 \ra \Irre1  \ra \TheB12 \ra \TheT21 \ra \TheT23 \ra \TheT41 \ra \TheT43 \ra \TheT61 \ra \Irre6 \ra \TheB61. \tag{$(i_6)\leftrightarrow(t_6)$}
\end{align}

\raggedright
\singlespacing
\bibliography{indec}
\bibliographystyle{unsrt}

\end{document}